\documentclass[11pt, letterpaper]{article}

\usepackage{bbm}
\usepackage{makecell}
\usepackage{needspace}
\usepackage{pifont}     %
\newcommand{\cmark}{\ding{51}}   %
\usepackage{setspace}

\makeatletter
\newcommand\email[2][]%
   {\newaffiltrue\let\AB@blk@and\AB@pand
      \if\relax#1\relax\def\AB@note{\AB@thenote}\else\def\AB@note{\relax}%
        \setcounter{Maxaffil}{0}\fi
      \begingroup
        \let\protect\@unexpandable@protect
        \def\thanks{\protect\thanks}\def\footnote{\protect\footnote}%
        \@temptokena=\expandafter{\AB@authors}%
        {\def\\{\protect\\\protect\Affilfont}\xdef\AB@temp{#2}}%
         \xdef\AB@authors{\the\@temptokena\AB@las\AB@au@str
         \protect\\[\affilsep]\protect\Affilfont\AB@temp}%
         \gdef\AB@las{}\gdef\AB@au@str{}%
        {\def\\{, \ignorespaces}\xdef\AB@temp{#2}}%
        \@temptokena=\expandafter{\AB@affillist}%
        \xdef\AB@affillist{\the\@temptokena \AB@affilsep
          \AB@affilnote{}\protect\Affilfont\AB@temp}%
      \endgroup
       \let\AB@affilsep\AB@affilsepx
}
\makeatother

\usepackage{amsmath}
\makeatletter
  \g@addto@macro \normalsize{%
    \setlength\abovedisplayskip{5pt plus 0pt minus 0pt}%
    \setlength\belowdisplayskip{5pt plus 0pt minus 0pt}}%
\makeatother

\makeatletter
\renewcommand{\paragraph}{%
  \@startsection{paragraph}{4}%
  {\z@}{1.5ex \@plus 1ex \@minus .2ex}{-1em}%
  {\normalfont\normalsize\itshape\bfseries}%
}
\makeatother

\usepackage{paralist}
\let\OLDthebibliography\thebibliography
\renewcommand\thebibliography[1]{
  \OLDthebibliography{#1}
  \setlength{\parskip}{2pt}
  \setlength{\itemsep}{4pt plus 1pt}
}
\usepackage{amsthm}
\usepackage{thmtools}
\usepackage{amssymb}
\usepackage{geometry}
\usepackage{tikz-cd}
\usepackage{empheq}
\usepackage{booktabs} %
\usepackage{tabularx}
\usepackage[ruled, linesnumbered]{algorithm2e} %

\usepackage{thm-restate}
\usepackage{array}
\usepackage{xcolor}
\usepackage{multicol}
\usepackage[colorlinks, pagebackref, hypertexnames=false]{hyperref}
\def\tmp#1#2#3{%
  \definecolor{Hy#1color}{#2}{#3}%
  \hypersetup{#1color=Hy#1color}}
\tmp{link}{HTML}{800006}
\tmp{cite}{HTML}{2E7E2A}
\tmp{file}{HTML}{131877}
\tmp{url} {HTML}{8A0087}
\tmp{menu}{HTML}{727500}
\tmp{run} {HTML}{137776}
\def\tmp#1#2{%
  \colorlet{Hy#1bordercolor}{Hy#1color#2}%
  \hypersetup{#1bordercolor=Hy#1bordercolor}}
\tmp{link}{!60!white}
\tmp{cite}{!60!white}
\tmp{file}{!60!white}
\tmp{url} {!60!white}
\tmp{menu}{!60!white}
\tmp{run} {!60!white}

\definecolor{myred}{RGB}{230,0,0}
\definecolor{myblue}{rgb}{0.01, 0.28, 1.0}

\colorlet{lightpurple}{myblue}

\usepackage{backref}
\renewcommand*{\backref}[1]{}
\renewcommand*{\backrefalt}[4]{%
  \ifcase #1 %
     (Not cited.)%
  \or
     (p.~#2.)%
  \else
     (pp.~#2.)%
  \fi%
}

\newcommand{\apprefproof}[1]{\hyperref[#1]{{\mbox{\footnotesize\textsc{\upshape{[proof]}}}}}}

\usepackage[T1]{fontenc}
\usepackage{newtxtext,newtxmath}
\usepackage[scaled=.95]{helvet}

\usepackage[numbers, sort]{natbib}
\usepackage{comment}
\usepackage{url}

\SetAlFnt{\small}
\SetAlCapFnt{\small}
\SetAlCapNameFnt{\small\scshape}
\SetAlgoSkip{smallskip}

\SetAlFnt{\small}
\SetAlCapFnt{\small}
\SetAlCapNameFnt{\small}
\SetAlCapHSkip{0pt}

\definecolor{chatcolor}{RGB}{0,112,100}

\renewcommand{\vec}[1]{\boldsymbol{#1}}

\allowdisplaybreaks
\newtheorem{theorem}{Theorem}[section]
\newtheorem{lemma}{Lemma}[section]

\newtheorem{proposition}{Proposition}[section]
\newtheorem{corollary}{Corollary}[section]
\newtheorem{example}{Example}[section]
\theoremstyle{definition}
\newtheorem{definition}{Definition}[section]

\newtheorem{remark}{Remark}[section]

\DeclareMathOperator{\OPT}{OPT}

\DeclareMathOperator{\PNE}{PNE}
\DeclareMathOperator{\MNE}{MNE}

\DeclareMathOperator{\EQ}{EQ}

\DeclareMathOperator{\EQPOA}{EQ-POA}

\DeclareMathOperator{\LW}{LW}

\DeclareMathOperator{\CCEPOA}{CCE-POA}

\DeclareMathOperator{\CCE}{CCE}
\DeclareMathOperator{\CE}{CE}
\DeclareMathOperator{\LP}{LP}

\DeclareMathOperator*{\argmax}{arg\,max}

\DeclareMathOperator{\EX}{\mathbb{E}}
\DeclareMathOperator*{\E}{\mathop{\mathbb{E}}}

\newcommand{\isubscr}[2]{%
  \ifthenelse{\equal{#2}{}}
  {\ensuremath{#1}}
  {\ensuremath{#1,\,#2}}
}
\newcommand{\iexp}[3]{\ensuremath{\EX_{\isubscr{#1}{#2}}[#3]}}
\newcommand{\iexpl}[3]{\ensuremath{\EX_{\isubscr{#1}{#2}}\left[#3\right]}}

\newcommand{\ssubscr}[2]{%
  \ifthenelse{\equal{#2}{}}
  {\ensuremath{#1}}
  {\ensuremath{\substack{#1 \\ #2}}}
}

\newcommand{\myexp}[3]{\iexpl{#1}{#2}{#3}}
\newcommand{\myexpl}[3]{\iexpl{#1}{#2}{#3}}

\DeclareMathOperator*{\supp}{supp}

\newcommand{\set}[1]{\ensuremath{\{ #1 \}}}
\newcommand{\sset}[2]{\ensuremath{\{ #1 \, \mid \, #2 \}}}
\newcommand{\setl}[1]{\ensuremath{\left\{ #1 \right\}}}

\newcommand{\cA}{\ensuremath{\mathcal A}}
\newcommand{\cX}{\ensuremath{\mathcal X}}
\newcommand{\cI}{\ensuremath{\mathcal I}}

\newcommand{\cC}{\ensuremath{\mathcal C}}

\newcommand{\charge}{\ensuremath{C}}

\newcommand{\va}{\ensuremath{\vec{a}}}
\newcommand{\vA}{\ensuremath{\vec{A}}}

\newcommand{\vx}{\ensuremath{\vec{x}}}
\newcommand{\vchi}{\ensuremath{\vec{\chi}}}
\renewcommand{\vv}{\ensuremath{\vec{v}}}

\newcommand{\z}{\ensuremath{u}}
\newcommand{\budget}{\ensuremath{\mathbb{B}}}

\newcommand{\A}{\ensuremath{A}}

\newcommand{\mech}{\ensuremath\mathcal{M}}

\newcommand{\ce}{\ensuremath{\textit{CE}}}
\newcommand{\swap}{\ensuremath{h}}
\newcommand{\mne}{\ensuremath{\textit{MNE}}}
\newcommand{\pne}{\ensuremath{\textit{PNE}}}

\usepackage{wrapfig}
\usepackage{enumitem}
\usepackage{graphicx}
\usepackage{titling}
\usepackage{subcaption}
\usepackage{color}
\usepackage[noblocks]{authblk}
\renewcommand\Affilfont{\normalsize}
\usepackage[normalem]{ulem}
\usepackage{nicefrac}
\usepackage{pgfplots}
\usepgfplotslibrary{fillbetween}
\pgfplotsset{compat=1.18} 
\usepackage{multirow}

\usepackage{bm}
\renewcommand{\vec}[1]{\bm{#1}}

\title{\bfseries
Tight Liquid Welfare Guarantees for Auctions with Budgets}

\author[2]{Twan Kroll}
\author[1,2]{Guido Sch\"afer}
\author[3]{Artem Tsikiridis}

\affil[1]{Centrum Wiskunde \& Informatica (CWI), The Netherlands}
\affil[2]{University of Amsterdam, The Netherlands}
\affil[3]{University of Warsaw, Poland}

\date{}

\begin{document}

\maketitle
\thispagestyle{empty}

\begin{abstract}
\noindent
We study the efficiency of auctions with budget-constrained bidders and derive tight price of anarchy (PoA) bounds for coarse correlated equilibria (CCE). To this end, we develop a general framework that reduces proving PoA guarantees to bounding the optimal value of a linear program. The LP is formulated over a small number of per-bidder equilibrium statistics, incorporates constraints implied by the auction setting and equilibrium behavior, and captures a novel smoothness notion tailored to budget-constrained environments; we provide a reduction from classical smoothness to this new notion. Our approach enables simple and tight PoA analyses across a broad range of auction formats. Using our framework, we obtain tight liquid welfare guarantees for simultaneous first-price, second-price, and all-pay auctions, simultaneous first-price auctions in which every bidder scales her values by a single bid multiplier (Deng et al., 2021), discriminatory and uniform price multi-unit auctions, and generalized first-price and second-price position auctions. Beyond budget-constrained settings, we also prove tight bounds for the budget-free uniform price auction and the budget-free generalized second-price auction.\end{abstract}

\thispagestyle{empty}

\setcounter{page}{1}

\section{Introduction}

In many important auctions, bidders operate under predetermined budgets.
A prominent example is online advertising, where advertisers face 
hard budgets set in advance for each ad campaign. 
In this domain, budget constraints have become increasingly important with the advent of autobidding~\cite{aggarwalSurvey}: 
advertisers configure automated bidders to bid on their behalf, and the primary parameter passed to such an agent is the budget they must not exceed.\footnote{According to a 2024 Financial Times report~\cite{ft2024}, major advertising platforms increasingly rely on AI-based bidding systems, while industry data from Topsort~\cite{topsort24} show that over 70\% of advertisers on platforms such as Google depend on automated bidding (“autobidding”) tools to manage their campaigns.} Beyond online advertising, government bonds \cite{auctiontheory} and carbon emission allowances \cite{goldner20} are sold through multi-unit auctions to bidders with hard budget constraints, while financially constrained bidders have long been a major concern in the design of spectrum auctions~\cite{chegale1998, cramton1997fcc}. A common feature of all these markets is that, due to their scale, they are typically implemented through \emph{simple}\footnote{We refer to an auction format as being \emph{simple} when the bidding interface is much more restricted than the bidders' preferences, e.g., a single bid per item, click, or unit, and allocation and payments are determined separately per item or unit.} auction formats, rather than through fully expressive truthful mechanisms. As a result, equilibrium outcomes may be inefficient. 

Measuring the inefficiency of equilibria in auctions through the lens of the Price of Anarchy (PoA) \cite{koutsoupias1999worst} has developed into a rich line of work, with tight PoA bounds known for many important auction formats, including simultaneous auctions~\cite{christodoulou16bayesian, christodoulou16, feldmanfu}, multi-unit auctions~\cite{dekeijzer13, birmpas2019tight}, and position auctions~\cite{caragiannis2015gsp}. A key toolkit in this line of work is the smoothness framework~\cite{Rou15, ST13}, which gives a principled approach for obtaining PoA bounds that extend to correlated equilibria (CE) and the broader class of coarse-correlated equilibria (CCE), and to settings with incomplete information. These results, however, are predominantly for budget-free environments. Budget constraints make the problem challenging from a game-theoretic standpoint in two ways. First, the auction becomes a \emph{generalized game}~\cite{debreu52, arrow54}: whether an action is feasible depends on the payments it induces, and hence on the actions of the other bidders. Second, quasi-linearity is lost: once 
bidders face budget constraints, 
preferences can no longer be expressed as value minus payment (see e.g., the discussion in \cite{fadaei2017truthfulness}). Even the right measure of (in)efficiency is non-trivial: social welfare is no longer an appropriate benchmark 
when values exceed budgets, and the 
standard benchmark becomes the \emph{liquid welfare}~\cite{dobzinski2014efficiency}, which captures the maximum total payments that can be extracted from the bidders.

The PoA of auctions with budget-constrained bidders has been studied for several important formats, including simultaneous first-price auctions~\cite{ST13, azar17, baldeschi2026}, divisible-resource mechanisms~\cite{caragiannis16, caragiannis2018efficiency}, and position auctions~\cite{voudouris2019position}. These results, however, are largely derived in a format-specific manner
and, with the exception of the recent work of \citet{baldeschi2026} on simultaneous first-price auctions, concern only pure, mixed, or Bayes--Nash equilibria, leaving the more robust learning-based solution concepts CE and CCE largely unexplored. 
At the same time, CE and CCE are important solution concepts in these settings. This is especially true when bidding is delegated to automated agents running learning algorithms: if every bidder uses a no-regret (respectively, no-swap-regret) algorithm, the empirical distribution of play converges to the set of CCE (respectively, CE)~\cite{hannan57, blum07} (see also \cite{hartline15}). This motivates the central question of our paper: %
\emph{``What is the price of anarchy of coarse correlated equilibria of simple auctions with budget-constrained bidders?''}

\subsection{Our Contributions}
\begin{table}[t]
\centering
\footnotesize
\setlength{\tabcolsep}{4pt}
\makebox[\linewidth][c]{\begin{tabular}{@{}lllll@{}}
\toprule
Format & Valuations & Upper bound & Lower bound & Tight? \\
\midrule
\multicolumn{5}{@{}l}{\textbf{Simultaneous item auctions}} \\
\cmidrule(r){1-5}
First-price & XOS & $2.188^{\,\ast}$~\cite{baldeschi2026} & $2.188$ (CCE)~\cite{baldeschi2026} & \cmark \\
First-price & subadditive & \textcolor{lightpurple}{$3$} (Thm.~\ref{thm:fpa-subadditive}) & \makecell[tl]{$2$ (budget-free MNE)~\cite{christodoulou16},\\ $2.188$ (CCE XOS)~\cite{baldeschi2026}} & ? \\
First-price (uniform-multiplier bidding) & additive & \textcolor{lightpurple}{$2.188$} (Thm.~\ref{thm:first-price}) & \textcolor{lightpurple}{$2.188$} (MNE, Thm.~\ref{thm:lb-uniform-bids}) & \cmark \\
Second-price$^{\,\circ}$ & monotone & \textcolor{lightpurple}{$2$}$^{\,\dagger}$ (Thm.~\ref{thm:second-price})& $2$ (budget-free PNE)~\cite{christodoulou16bayesian} & \cmark \\
All-pay & subadditive & \textcolor{lightpurple}{$3$} (Thm.~\ref{thm:all-pay})& \textcolor{lightpurple}{$3$} (MNE, Thm.~\ref{thm:lb-apa}) & \cmark \\
\midrule
\multicolumn{5}{@{}l}{\textbf{Multi-unit auctions}} \\
\cmidrule(r){1-5}
Discriminatory price & concave & \textcolor{lightpurple}{$2.188$} (Thm.~\ref{thm:first-price})& \textcolor{lightpurple}{$2.188$}  (CCE, Thm.~\ref{thm:lb-fpa}) & \cmark \\
Uniform price$^{\,\circ}$ & concave & \textcolor{lightpurple}{$2.188$}$^{\,\dagger}$ (Thm.~\ref{thm:upa-individual-certificate})& \textcolor{lightpurple}{$2.188$} (budget-free PNE, Thm.~\ref{thm:UPA-LB})& \cmark \\
\midrule
\multicolumn{5}{@{}l}{\textbf{Position auctions}} \\
\cmidrule(r){1-5}
Generalized first-price & value-per-click & \textcolor{lightpurple}{$2.188$} (Thm.~\ref{thm:first-price})& \textcolor{lightpurple}{$2.188$} (CCE, Thm.~\ref{thm:lb-fpa}) & \cmark \\
Generalized second-price$^{\,\circ}$ & value-per-click & \textcolor{lightpurple}{$2.582$} (Thm.~\ref{thm:gsp})& \textcolor{lightpurple}{$2.582$} (CCE, Thm.~\ref{thm:lb-gsp-budgets}) & \cmark \\
\bottomrule
\end{tabular}}
\caption{The PoA of CCE for budget-constrained bidders. Bounds indicated in \textcolor{lightpurple}{blue} are established in this paper. The constants are rounded and admit the closed forms $2.188 \approx {(2+\mathcal{W}_0(-e^{-2}))}/{(1+\mathcal{W}_0(-e^{-2}))}$ and $2.582 \approx 1+{e}/{(e-1)}$, where $\mathcal{W}_0$ is the principal branch of the Lambert~$W$ function \cite{corless96}. For the formats marked with $^{\circ}$ (second-price-type payment rules), the bounds are for equilibria that satisfy a no-overbidding assumption (Definition~\ref{def:nob}). \\
$^{\ast}$Recovered in this paper via our LP framework. 
$^{\dagger}$Matches the corresponding budget-free bounds (\cite{christodoulou16bayesian} for second-price; Theorem~\ref{thm:upa-individual-certificate} for uniform price).}
\label{tab:results}
\end{table}

We obtain the following results.

\begin{enumerate}[wide, labelindent=0pt, label=\arabic*., itemsep=2pt]
    \item We develop a linear programming framework (Section~\ref{sec:poa-duality}) that reduces proving PoA bounds for CCE of budget-constrained bidders
    to bounding the optimal value of a linear program \eqref{eq:ultimate-LP-simple}. Unlike most previous LP-based approaches in the literature, which concern budget-free settings (see also Appendix~\ref{app:duality} for a comparison of LP-based PoA frameworks), the variables of our LP formulation are not the equilibrium distribution itself, but rather a small number of per-bidder \emph{equilibrium statistics}. Further, the constraints of our LP capture a new notion of \emph{budget-feasible semi-smoothness} (Definition~\ref{def:semi-smoothness}), tailored to environments with budgets, together with a reduction (Theorem~\ref{thm:smoothness-reduction}) showing that classical, budget-free smoothness implies the new notion. Proposition~\ref{prop:welfare-lp-bound} gives such a bound, and Theorem~\ref{thm:ultimate-POA} converts it into a PoA bound, for pay-your-bid as well as second-price-type payment rules. %

\item Our LP framework in Section~\ref{sec:poa-duality} yields tight or near-tight liquid welfare guarantees for a number of auction formats of practical significance. We first consider simultaneous auctions in Sections~\ref{sec:first-price} and~\ref{sec:simultaneous-auctions}. We recover the tight bound of \cite{baldeschi2026} for first-price auctions (FPAs) with fractionally subadditive (XOS) valuations (Theorem~\ref{thm:first-price}), prove the first upper bound for the more general class of subadditive valuations (Theorem~\ref{thm:fpa-subadditive}), and derive new tight bounds for the second-price and all-pay variants (Theorems~\ref{thm:second-price} and~\ref{thm:all-pay}). We also obtain the first bounds for simultaneous FPAs under the restricted \emph{uniform-multiplier bidding} interface studied in \cite{deng21, LMZ24} (through Theorem~\ref{thm:first-price}), which is relevant for autobidding applications\footnote{This interface also coincides with the first-price pacing model of \citet{conitzer2022pacing}, in which budget-constrained bidders scale their valuations by pacing multipliers. There, a \emph{first-price pacing equilibrium} (FPPE) always exists, is unique, and can be computed in polynomial time. The FPPE is, however, a market equilibrium notion; our results instead bound the inefficiency of all CCE of the induced game with utility maximizers.}. In Sections~\ref{sec:first-price} and~\ref{sec:upa}, we turn to the standard multi-unit formats \cite{auctiontheory}, i.e., the discriminatory and uniform price auctions, and settle the price of anarchy of both under budgets (Theorems~\ref{thm:first-price} and~\ref{thm:upa-individual-certificate}). Finally, we analyze the efficiency of position auctions, namely the generalized first-price (GFP) and generalized second-price (GSP) auction, under budget constraints and prove tight bounds for both (Theorems~\ref{thm:first-price} and~\ref{thm:gsp}).

\item We complement our upper bounds with a collection of matching lower bound constructions. These are established for the most restrictive solution concept for which they are possible, i.e., for several formats already at mixed or even pure Nash equilibria, and therefore also for every broader equilibrium concept (see Table~\ref{tab:results}). Our lower bounds are not limited to budget-constrained settings. Notably, under the no-overbidding assumption we consider (see Definition \ref{def:nob} and the discussion in Appendix~\ref{app:sec:NOB-discussion}), our lower bound for the uniform price auction is budget-free and matches our upper bound, thereby settling the inefficiency of the format under uniform bidding even \emph{without} budgets. For the GSP auction, we also settle the budget-free case, where the tight bound is $\frac{1}{1-\mathcal{W}_0(1)} \approx 2.310$ (Theorems~\ref{thm:gsp} and~\ref{thm:lb-gsp-free}).

\end{enumerate}

\subsection{Technical Summary}

\paragraph{The liquid welfare LP.} 
Linear programming duality has previously been used for PoA proofs~\cite{nadav10, kulkarni14, kim18, bilo18, kim25}; however, our formulation in \eqref{eq:ultimate-LP-simple} differs substantially from these approaches.
Rather than using the CCE distribution itself as the variables of the LP, 
as is done in previous approaches, we keep track of a handful of \emph{equilibrium statistics} per bidder: their \emph{liquid value}\footnote{The \emph{liquid value} of a bidder is the minimum of their expected value and their budget (see Section~\ref{sec:preliminaries}).} at equilibrium, their expected payment, their expected willingness to pay (see Definition~\ref{def:wtp}), and the expected value of a mechanism-dependent \emph{charge} evaluated at the optimal allocation.%
The key step is to identify a small set of linear constraints that these statistics satisfy at every CCE, together with a partition of bidders based on whether their expected value at equilibrium reaches their budget.
Altogether, the liquid welfare of any CCE is lower bounded by the optimum of a linear program, and our bound on this optimum (Proposition~\ref{prop:welfare-lp-bound}) gives a PoA bound (Theorem~\ref{thm:ultimate-POA}).

\paragraph{A reduction.} 
The obstacle that separates budget-constrained from budget-free mechanisms is that the deviation at the heart of every smoothness proof may be unaffordable against some behavior of the opponents. Our notion of \emph{budget-feasible semi-smoothness} (Definition~\ref{def:semi-smoothness}) builds affordability into the deviation, and this is precisely what Theorem~\ref{thm:ultimate-POA} requires.  Whenever the budget-free deviations are \emph{conservative}, i.e., they pay in expectation at most the expected value of the allocation they target, they can be converted into budget-feasible deviations (Theorem~\ref{thm:smoothness-reduction}) by mixing them with the opt-out action. The mixing probability is calibrated so that the expected payment remains below the budget, while the guarantee degrades to the \emph{liquid} value; this is the statistic tracked by \eqref{eq:ultimate-LP-simple}. Thus, for each auction format, the remaining task is to prove a single conservative smoothness inequality.

\subsection{Related Work}

The study of auctions with budget-constrained bidders has received considerable attention in algorithmic game theory. Important early works are those of \citet{chegale1998}, who study single-item auctions with financially constrained bidders, and of \citet{borgs2005}, who focus on multi-unit auctions, both with an emphasis on revenue. Since then, an important line of work has focused on the design of truthful auctions for budget-constrained bidders \cite{dobzinski2012multi, fiat2011single, goel2015polyhedral}. The study of the inefficiency of equilibria started later, after \citet{dobzinski2014efficiency} defined the liquid welfare benchmark (see also \cite{luxiao15, luxiao17}) and the conservative smoothness notion\footnote{The guarantee of \citet{ST13} lower bounds the \emph{social} welfare at equilibrium by a fraction of the optimal liquid welfare, and hence does not bound the liquid welfare attained at equilibrium.}, \citet{azar17} studied the ex-post liquid PoA of simultaneous item auctions, whereas \citet{caragiannis16, caragiannis2018efficiency} adopted, as we do, the ex-ante measure to study the proportional (Kelly) mechanism, calling it effective welfare. For position auctions with budgets, equilibrium existence was studied in \cite{ashlagi2010position, diaz2016stability} and the PoA for pure Nash equilibria was shown by \citet{voudouris2019position} to be exactly $2$ for GFP, GSP and VCG.

Most work on the price of anarchy of simple auctions has been carried out for bidders who are not budget-constrained. The standard toolkit for analyzing coarse-correlated equilibria has been the smoothness framework of \citet{Rou15} and its extension to auction mechanisms by \citet{ST13}. In the budget-free setting, most simple auction formats have been studied, including in the incomplete information setting; see e.g., \cite{ST13, syrgkanis-thesis, christodoulou16, feldmanfu, hassidim2011non} for simultaneous first-price auctions, \cite{christodoulou16bayesian, bhawalkar11, feldmanfu, dutting2018valuation} for simultaneous second-price auctions, \cite{dekeijzer13, birmpas2019tight, markakis21} for multi-unit auctions, and \cite{caragiannis2015gsp} for the generalized second price auction. For the single-item first-price auction, \citet{feldman16} study correlated and coarse correlated equilibria, while a sequence of works in the Bayesian setting \cite{ST13, HTW18} led to the tight bound of $\frac{e^2}{e^2-1}$ of \citet{JL23}. We refer the reader to the survey of \citet{roughgarden2017price} for a more detailed overview of this line of research. For a detailed comparison with existing LP-based PoA frameworks we refer to Appendix~\ref{app:duality}.

Our work also relates to the autobidding literature, where automated bidders bid on behalf of advertisers under budget and return-on-investment (ROI) constraints, with objectives ranging from value to utility maximization; see \citet{aggarwalSurvey} for a survey. Tight PoA bounds for value maximizers are known for pure Nash equilibria of the second-price auction~\cite{ABM19}, the VCG mechanism~\cite{deng21}, and the GSP auction~\cite{deng24efficiency}, as well as for simultaneous first-price auctions at pure and mixed Nash equilibria~\cite{LMP23, deng2024, LMZ24}; under the restricted uniform-multiplier bidding interface, the PoA for budgeted value maximizers is unbounded~\cite{LMZ24}, in contrast to our tight bound for budgeted utility maximizers (Theorem~\ref{thm:first-price}). Closest to our work, \citet{baldeschi2026} prove tight liquid welfare guarantees for CCE of simultaneous first-price auctions, which we recover through our LP framework, and \citet{kim25} develops a duality-based analysis of the liquid welfare PoA for the proportional mechanism and the pacing metagame of \citet{feng24}.

\section{Preliminaries}\label{sec:preliminaries}

We use $ [n] = \set{1, \dots, n}$, $\Delta(S)$ for distributions over a set $S$, and $\supp(A)$ for the support of a distribution $A$. Given a vector $\vec{x} = (x_i)_{i \in [n]}$, let $\vec{x}_{-i}$ denote the vector obtained from $\vec{x}$ by removing $x_i$, and let $(x'_i, \vec{x}_{-i})$ denote the vector obtained from $\vec{x}$ by replacing $x_i$ with $x'_i$.
We define $\mathbbm{1}_{\text{[cond]}}$ as the indicator variable evaluating to $1$ iff [cond] is true. 
Finally, $\mathcal{W}_0$ and $\mathcal{W}_{-1}$ denote the two real branches of the Lambert~$W$ function~\cite{corless96}: for $x \in [-e^{-1}, 0)$, the equation $w\, e^{w} = x$ has two possibly coinciding real solutions, $\mathcal{W}_0(x) \in [-1, 0)$ and $\mathcal{W}_{-1}(x) \in (-\infty, -1]$, and for $x \geq 0$ it has the unique real solution $\mathcal{W}_0(x) \geq 0$. At $x=-e^{-1}$, both branches equal $-1$.

\subsection{Auction Model and Equilibria}\label{sec:auction-model}

\paragraph{Model.}
We consider a set $N =  [n]$ of $n \ge 2$ bidders. Each bidder $i$ has an action set $\cA_i$, a finite allocation set $\cX_i$, a valuation function $v_i:\cX_i \to \mathbb{R}_{\geq 0}$,
and a budget $\budget_i\in (0, \infty]$. 
Write $\cA=\prod_{i\in N}\cA_i$ and let $\cX\subseteq\prod_{i\in N}\cX_i$ be the feasible allocation profiles. 
An instance is given by $I=(N,\cA,\cX,\vv,\budget)$, where $\vv := (v_i)_{i \in N}$ and $\budget := (\budget_i)_{i \in N}$ are the valuation and budget profiles. We write $\cI$ for a class of instances, $\cI^\infty \subseteq \cI$ for its \emph{budget-free} instances, i.e., $\budget_i = \infty$ holds for every $i$.%

A mechanism $\mech=(X,P)$ consists of an allocation rule $X:\cA\to\cX$ and a payment rule $P:\cA\to\mathbb R_{\geq0}^n$. Their components are $X_i$ and $P_i$, and bidder $i$'s utility is $\z_i(\va):=v_i(X_i(\va))-P_i(\va)$. Each bidder has an \emph{opt-out} action $\emptyset_i$ with $P_i(\emptyset_i,\va_{-i})=0$ for every opposing profile, and hence $\z_i(\emptyset_i, \va_{-i}) \ge 0$.

\paragraph{Budget feasibility.}
Given an opposing distribution  $\vA_{-i}\in\Delta(\cA_{-i})$, bidder $i$ chooses a distribution (also called \emph{mixed action}) $\A_i\in\Delta(\cA_i)$, sampled independently of the opponents, by solving
\begin{align*}
 \max_{\A_i\in\Delta(\cA_i)}\quad&
 \myexpl{a_i\sim A_i}{\va_{-i}\sim\vA_{-i}}{\z_i(a_i,\va_{-i})}
 \tag{Bidder's Problem}\label{eq:agents-problem}\\
 \text{subject to}\quad&
 \myexpl{a_i\sim A_i}{\va_{-i}\sim\vA_{-i}}{P_i(a_i,\va_{-i})}\leq\budget_i.
\end{align*}
Write $\cC_i(\vA_{-i})$ for the mixed actions $\A_i \in \Delta(\cA_i)$ that are budget-feasible in this problem. 
A joint distribution $\vA\in\Delta(\cA)$ is \emph{budget-feasible} if $\myexpl{\va\sim\vA}{}{P_i(\va)}\leq\budget_i$ holds for every $i$. Because $\cC_i(\vA_{-i})$ depends on the opponents' distribution, this is a generalized game~\cite{debreu52,arrow54}.

\begin{definition}[Coarse correlated equilibrium]\label{def:equilibria}
A budget-feasible distribution $\vA\in\Delta(\cA)$ is a \emph{coarse correlated equilibrium (CCE)} if, for every bidder $i$,
\begin{equation}\label{eq:cce-def}
 \myexpl{\va\sim\vA}{}{\z_i(\va)}
 \geq\myexpl{a_i'\sim A_i'}{\va_{-i}\sim\vA_{-i}}{\z_i(a_i',\va_{-i})}
 \qquad\forall A_i'\in\cC_i(\vA_{-i}).
\end{equation}
\end{definition}
We denote these distributions by $\CCE(I)$. The definitions of correlated equilibria $\CE(I)$, mixed Nash equilibria $\MNE(I)$, and pure Nash equilibria $\PNE(I)$ are given in Appendix~\ref{app:equilibrium-notions}, and $\EQ$ denotes any of these equilibrium notions. We have $\MNE(I)\subseteq\CE(I)\subseteq\CCE(I)$. Under our pure-deviation convention for PNE, expected budgets mean that $\PNE(I)\subseteq\MNE(I)$ need not hold; see Example~\ref{ex:pne-not-mne}. Our upper bounds concern CCE.

\subsection{Auction Formats}\label{sec:auction-formats}

For all the auction formats we consider, tie-breaking rules are fixed as part of the mechanism, and all ties break consistently, where the order may differ between ties at zero and ties at a positive bid. Lists of bids are padded with zeros: whenever we refer to the $j$th highest of fewer than $j$ bids, this bid is zero.

\paragraph{Simultaneous Item Auctions.}
There are $m$ items, indexed by $M=[m]$. Each bidder submits a bid vector $a_i=(a_{ij})_{j\in M}\in\cA_i=\mathbb R_{\geq0}^m$ and receives a bundle $X_i(\va)\subseteq M$, i.e., $\cX_i=2^M$, and $\cX$ consists of the profiles of pairwise disjoint bundles. Each item $j\in M$ is awarded to a bidder with the highest bid on it, which we denote by $\beta_j(\va):=\max_{\ell\in N}a_{\ell j}$, and we denote by $\beta_j(\va_{-i}):=\max_{\ell\ne i}a_{\ell j}$ the highest opposing bid on item $j$. Bidder $i$ pays $P_i(\va)=\sum_{j\in X_i(\va)}a_{ij}$ in a first-price auction, $\sum_{j\in X_i(\va)}\beta_j(\va_{-i})$ in a second-price auction, and $\sum_{j\in M}a_{ij}$ in an all-pay auction. Valuations $v_i\colon 2^M\to\mathbb R_{\geq0}$ are monotone and normalized at the empty set. A valuation $v_i$ is \emph{additive} if $v_i(S)=\sum_{j\in S}v_{ij}$ holds for every $S\subseteq M$ and some $(v_{ij})_{j\in M}\in\mathbb R^m_{\geq0}$. It is \emph{XOS} if, for every bundle $x_i\subseteq M$, there are non-negative additive supports $v_{ij}^{x_i}$ with $v_i(x_i)=\sum_{j\in x_i}v_{ij}^{x_i}$ and $\sum_{j\in S\cap x_i}v_{ij}^{x_i}\leq v_i(S)$ for every $S\subseteq M$, where we extend the support by zero outside $x_i$. It is \emph{subadditive} if $v_i(S\cup T)\leq v_i(S)+v_i(T)$ holds for all $S,T\subseteq M$. Additive valuations are XOS, and XOS valuations are subadditive (see e.g., \cite{lln2001}). For first-price auctions with additive valuations, we also consider \emph{uniform-multiplier bidding}, a restricted bidding interface studied in~\cite{deng21,LMZ24} that is relevant for autobidding applications. Under uniform-multiplier bidding, bidder $i$ submits a \emph{bid multiplier} $\alpha_i\in[0,1]$, i.e., $\cA_i=[0,1]$, and her bid on every item $j$ becomes $a_{ij}=\alpha_i v_{ij}$. The allocation and the payments are those of the first-price auction for the induced bid profile $\va$.

\paragraph{Position Auctions.}
There are $m< n$ slots with click-through rates $\alpha_1\geq\cdots\geq\alpha_m>0$. We assume that these click-through rates are normalized, i.e., $\alpha_1=1$. For the lower bounds in Appendix~\ref{app:lower-bounds} we often work with unnormalized click-through rates for convenience. Bidder $i$ has a value per click $\bar v_i\geq0$, submits a bid $a_i\in\cA_i=\mathbb R_{\geq0}$, and receives at most one slot, i.e., $\cX_i=\set{0,1,\ldots,m}$, where $0$ denotes receiving no slot, and $\cX$ consists of the profiles in which every slot is assigned to at most one bidder. Her value for slot $j$ is $v_i(j)=\alpha_j\bar v_i$, and $v_i(0)=0$. Bidders are ranked by bid, with consistent tie-breaking, and slot $j$ goes to the $j$th ranked bidder. We denote by $\beta_j(\va)$ the $j$th highest bid in $\va$, so that the bidder who receives slot $j$ bids $\beta_j(\va)$, and by $\beta_j(\va_{-i})$ the $j$th highest opposing bid. The bidder $i$ who receives slot $j\geq1$ pays $P_i(\va)=\alpha_j a_i$ under the generalized first-price auction (GFP) and $P_i(\va)=\alpha_j\beta_j(\va_{-i})$, i.e., the click-through rate of her slot times the next bid in the ranking, under the generalized second-price auction (GSP), and a bidder who receives no slot pays nothing. Slots are distinguished by their ranks, even when their click-through rates coincide.

\paragraph{Multi-Unit Auctions.}
There are $k\geq1$ identical units. Under \emph{uniform bidding}~\cite{auctiontheory}, bidder $i$ submits $a_i=(b_i,q_i)\in\cA_i=\mathbb R_{\geq0}\times\set{0,\ldots,k}$, representing $q_i$ copies of the per-unit bid $b_i$, which we call her \emph{marginal bids}. The allocation sets are $\cX_i=\set{0,\ldots,k}$ and $\cX=\sset{\vx\in\prod_i\cX_i}{\sum_i x_i\leq k}$. The $k$ highest marginal bids win, with ties broken consistently. We denote by $\beta_1(\va)\leq\dots\leq\beta_k(\va)$ the $k$ winning marginal bids in non-decreasing order, by $\beta_0(\va)$ the highest losing marginal bid, i.e., the $(k+1)$th highest marginal bid, and by $\beta_1(\va_{-i})\leq\dots\leq\beta_k(\va_{-i})$ the $k$ highest opposing marginal bids. The discriminatory auction (DPA) charges the winning bids, i.e., $P_i(\va)=b_iX_i(\va)$, while the uniform price auction (UPA) charges each winner the highest losing bid per unit, i.e., $P_i(\va)=\beta_0(\va)X_i(\va)$. Valuations $v_i\colon\cX_i\to\mathbb R_{\geq0}$ are monotone and normalized at zero, and we use \emph{concave} valuations, whose marginal values $v_i(x)-v_i(x-1)$ are non-increasing in $x\in[k]$.

\subsection{Welfare and No-Overbidding}\label{sec:welfare-nob}

\paragraph{Liquid welfare.}
For a random allocation $\chi_i\in\Delta(\cX_i)$, define the ex-ante liquid value
\[
 \hat v_i(\chi_i):=\min\setl{\myexpl{x_i\sim\chi_i}{}{v_i(x_i)},\budget_i}.
\]
For $\vchi\in\Delta(\cX)$ with marginals $\chi_i$, let $\LW(I,\vchi):=\sum_{i \in N}\hat v_i(\chi_i)$ and choose an optimal profile $\vchi^*\in\argmax_{\vchi\in\Delta(\cX)}\LW(I,\vchi)$. Write $\OPT(I):=\LW(I,\vchi^*)$. The optimum exists because $\cX$ is finite, and it may require randomization. For a distribution $\vA$ over actions, we use
\[
 \LW_i(I,\vA):=\min\setl{\myexpl{\va\sim\vA}{}{v_i(X_i(\va))},\budget_i},
 \qquad \LW(I,\vA):=\sum_{i\in N}\LW_i(I,\vA).
\]
The cap is applied after taking expectation, following also~\cite{dobzinski2014efficiency,deng2024,baldeschi2026,caragiannis16}. Restricting the benchmark to deterministic allocations need not upper-bound equilibrium liquid welfare; see Example~\ref{ex:randomized-benchmark}. Appendix~\ref{app:liquid-welfare-benchmark} discusses the randomized benchmark and the distinction from ex-post liquid welfare (see Remark~\ref{rem:ex-ante-lw}).

\begin{definition}\label{def:wtp}
Bidder $i$'s \emph{willingness to pay (WTP)} for allocation $x_i \in \cX_i$ when using $a_i \in \cA_i$ is
\[
 W_i(a_i,x_i):=\sup_{\va_{-i}:\,X_i(a_i,\va_{-i})=x_i}P_i(a_i,\va_{-i}),
 \qquad W_i(\va):=W_i(a_i,X_i(\va)).
\]
\end{definition}
Intuitively, $W_i(a_i,x_i)$ is the most that bidder $i$ may pay with the action $a_i$ when she receives $x_i$. This willingness-to-pay notion is due to~\cite{ST13}.

By definition, for every bidder $i$, $P_i(\va) \le W_i(\va)$ holds for every $\va \in \cA$.

\begin{definition}\label{def:pyb}
A payment rule $P$ is \emph{pay-your-bid (PYB)} if, for every bidder $i$, $P_i(\va) = W_i(\va)$ holds for every $\va \in \cA$, and we call a mechanism PYB if its payment rule is.
\end{definition}

The following payment bound follows from budget feasibility and opting out.

\begin{lemma}\label{lem:IR}
For every $\vA\in\EQ(I)$ and bidder $i$,
$
 \myexpl{\va\sim\vA}{}{P_i(\va)}\leq\LW_i(I,\vA).
$
\end{lemma}
\begin{proof}
By budget-feasibility, $\myexpl{\va\sim\vA}{}{P_i(\va)} \le \budget_i$. 
Further, $i$'s utility is non-negative when opting out. So, under $\vA$, $i$'s expected utility is non-negative as well, and thus $\myexpl{\va\sim\vA}{}{P_i(\va)} \le \myexpl{\va\sim\vA}{}{v_i(X_i(\va))}$. 
\end{proof}

\begin{definition}\label{def:nob}
A distribution $\vA$ is
\emph{no-overbidding-feasible (NOB-feasible)} if
\begin{equation}\label{eq:NOB}
 \myexpl{\va\sim\vA}{}{W_i(\va)}\leq\LW_i(I,\vA)
 \qquad\forall i\in N.
\end{equation}
Write $\EQ^{\mathrm{NOB}}(I)$ for the set of NOB-feasible equilibria in $\EQ(I)$.
\end{definition}
This condition is imposed in expectation for each bidder separately. It is an equilibrium refinement: deviations remain unrestricted apart from budget feasibility. PYB equilibria satisfy it automatically by Lemma~\ref{lem:IR}. For other mechanisms, we state our guarantees for NOB-feasible distributions. Note that NOB is a condition on a distribution, not a property of the payment rule. Without NOB, even a single-item second-price auction can have unbounded PoA; see Example~\ref{ex:unboundedPOA}. Further discussion is given in Remark~\ref{rem:nob-comparison}.

\paragraph{Price of anarchy.}
For a class $\cI$ and equilibrium notion $\EQ$, the liquid price of anarchy is
\[
 \EQPOA(\cI):=\sup_{I\in\cI}\sup_{\vA\in\EQ(I)}
 \frac{\OPT(I)}{\LW(I,\vA)}.
\]
When invoking NOB, the inner supremum is over $\EQ^{\mathrm{NOB}}(I)$. Throughout our upper-bound analysis, this restriction is understood for non-PYB mechanisms and it is vacuous for PYB mechanisms.

\section{PoA of Auctions with Budgets via a Liquid Welfare LP}\label{sec:poa-duality}

In this section, we develop the technical machinery for deriving price of anarchy bounds for coarse-correlated equilibria in the presence of budgets.

\subsection{Budget-Feasible Semi-Smoothness and PoA Bounds}\label{sec:poa-bounds}

We derive an extension of the smoothness notions for auctions \cite{ST13,dutting2018valuation,LPL2011} that is suitable for environments with budgets.
The definition relies on a mechanism-dependent \emph{charging function} (\emph{charge} for short) $\charge_i\colon \cX_i \times \cA_{-i} \to \mathbb{R}_{\geq 0}$ for each bidder $i$, whose role will become clear below. 

\begin{definition}\label{def:semi-smoothness}
Fix a mechanism $\mech=(X,P)$, an instance $I$, and an optimal profile $\vchi^*$. For each bidder $i \in N$, fix $\lambda_i>0$, $\mu_i \geq 0$, and a charging function $\charge_i:\cX_i\times\cA_{-i}\to\mathbb R_{\geq0}$. Bidder $i$ is \emph{budget-feasible $(\lambda_i,\mu_i)$-semi-smooth} if, for every $\vA\in\Delta(\cA)$, there is $A_i'\in\cC_i(\vA_{-i})$ such that
\[
 \myexpl{a_i'\sim A_i'}{\va_{-i}\sim\vA_{-i}}{\z_i(a_i',\va_{-i})}
 \geq
 \lambda_i \hat v_i(\chi_i^\ast)
 -\mu_i\myexpl{x_i\sim\chi_i^*}{\va_{-i}\sim\vA_{-i}}{\charge_i(x_i,\va_{-i})}.
\]
\end{definition}
Note that the deviation is independent of realized opposing actions, but may depend on their distribution.
Also, the charge and the coefficients $\lambda_i, \mu_i$ may depend on $I$ and $\vchi^*$, but are fixed before the distribution $\vA$ is considered; when the coefficients are common to all bidders, we write $(\lambda,\mu)$ instead. 
In Section~\ref{sec:reduction-smoothness}, we relate our budget-feasible notion of semi-smoothness to existing smoothness notions and prove reductions from the latter to the former.

\paragraph{The LP constraints.}

Our framework derives upper bounds on the price of anarchy by lower-bounding the liquid welfare at equilibrium through a linear program (LP). The LP is constructed from valid inequalities involving expected statistics of a fixed CCE $\vA$ and an optimal profile $\vchi^*$.

Fix a mechanism $\mech$ and an instance $I$. 
Let $\vA\in \CCE(I)$ be a coarse correlated equilibrium and $\vec{\chi}^*$ an optimal profile. 
Given $\vA\in \CCE(I)$ and $\vec{\chi}^*$, we define the following expected statistics: 
\begin{align*}
  LW_i & := \LW_i(I,\vA) =\min\setl{\myexpl{\va\sim\vA}{}{v_i(X_i(\va))},\budget_i}, \qquad 
  V_i := \hat{v}_i(\chi_i^\ast) = \min\{\myexp{x_i\sim \chi_i^\ast}{}{v_i(x_i)},\budget_i\} \\
  C_i & := \myexp{x_i\sim \chi_i^\ast}{\va_{-i}\sim \vA_{-i}}{\charge_i(x_i,\va_{-i})}, \qquad P_i :=\myexpl{\va\sim\vA}{}{P_i(\va)}, \qquad W_i :=\myexpl{\va\sim\vA}{}{W_i(\va)}.
\end{align*}

We partition the bidders in $N$ into 
\[
 N_1:=\sset{i\in N}{\myexpl{\va\sim\vA}{}{v_i(X_i(\va))}\geq\budget_i} \qquad\text{and}\qquad
 N_2:=N\setminus N_1.
\]

For a bidder $i\in N_1$, we obtain the following valid inequality by definition of $N_1$:
\begin{equation}\label{eq:inequality-N1}
    LW_i =\budget_i\geq\min\{\myexp{x_i\sim \chi_i^\ast}{}{v_i(x_i)},\budget_i\}= V_i.
\end{equation}
Consider a bidder $i \in N_2$ and assume that $i$ is budget-feasible $(\lambda_i, \mu_i)$-semi-smooth. 
Let $A_i'$ be the respective deviation for bidder $i$ from Definition \ref{def:semi-smoothness}. Then, by definition of $N_2$, we obtain that
$$
LW_i=\E_{\va\sim \vA}[v_i(X_i(\va))]=\E_{\va\sim \vA}[\z_i(\va)]+ P_i \geq \myexp{a_i'\sim A_i'}{\va_{-i}\sim \vA_{-i}}{\z_i(a_i',\va_{-i})}+ P_i,
$$
where the inequality holds because $\vA$ is a CCE. 
Using the semi-smoothness inequality from Definition~\ref{def:semi-smoothness}, we obtain the following valid inequality, which we call \emph{deviation inequality}: 
\begin{equation}\label{eq:inequality-N2}
        LW_i \geq \lambda_i V_i - \mu_i \charge_i + P_i.
\end{equation}

Further, by Lemma~\ref{lem:IR}, Definition~\ref{def:wtp} of the willingness to pay, and the NOB constraint, for every bidder $i \in N$, the following constraints are valid:
\begin{align}
    LW_i & \geq P_i \label{equation:inequality-budget-value-payment} 
    \\
    W_i & \geq P_i \label{equation:inequality-willingness}
    \\
    LW_i 
    & \geq W_i .
    \label{equation:NOB-ineq}
\end{align}
In fact, for PYB mechanisms, \eqref{equation:inequality-willingness} becomes an equality; this strengthening is crucial for PYB and will be handled separately below.

\paragraph{The coupling constraint.}

In our framework, we separate the deviation inequality \eqref{eq:inequality-N2} from the accounting of its charges.
To this end, we introduce the following \emph{coupling constraint}: for fixed $\kappa_i,\nu_i\geq0$, 
\begin{equation}\label{eq:coupling-constraint}
 \sum_{i\in N}\frac{\mu_i}{\lambda_i}\charge_i
 \leq\sum_{i\in N}\kappa_i P_i +\sum_{i\in N}\nu_i W_i.
\end{equation}
As for $\lambda_i,\mu_i$, the coefficients $\kappa_i,\nu_i$ may depend on $I$ and $\vchi^*$, but not on $\vA$.
The coupling constraint is an auction-specific inequality to prove for the chosen charges. %

\paragraph{The liquid welfare LP.}

We now introduce the \emph{liquid welfare LP}, which constitutes the engine underlying all our liquid welfare bounds. 
Using the inequalities~\eqref{eq:inequality-N1}, \eqref{eq:inequality-N2}, \eqref{equation:inequality-budget-value-payment}, \eqref{equation:inequality-willingness}, \eqref{equation:NOB-ineq}, and \eqref{eq:coupling-constraint} we obtain the following \emph{liquid welfare LP}:
\begin{equation}\tag{Liquid welfare LP}\label{eq:ultimate-LP-simple}
\begin{aligned}
 \LP:=\min\quad &\sum_{i\in N}LW_i\\
 \text{subject to}\quad
 &LW_i\geq V_i &&i\in N_1\\
 &LW_i  \geq\lambda_iV_i-\mu_iC_i+P_i &&i\in N_2\\
 &LW_i  \geq W_i\geq P_i &&i\in N\\
 &\sum_{i\in N}\frac{\mu_i}{\lambda_i}C_i
       \leq\sum_{i\in N}\kappa_iP_i+\sum_{i\in N}\nu_iW_i\\
 &LW_i, C_i, P_i, W_i \geq0 &&i\in N.
\end{aligned}
\end{equation}
Recall that for PYB mechanisms we additionally impose $W_i=P_i$. 
In the LP, $LW_i,C_i,P_i,W_i$ are variables, for which we use the same notation as for the corresponding statistics. The benchmark values $V_i=\hat v_i(\chi_i^*)$ are determined by the fixed optimal profile $\vchi^*$, so $\sum_i V_i=\OPT(I)$.

\paragraph{The LP bound.}

We first bound the value of the liquid welfare LP. This argument uses only the LP constraints; the equilibrium assumptions are needed later to establish feasibility.

\begin{restatable}[Liquid welfare LP bound]{proposition}{propwelfarelpbound}\label{prop:welfare-lp-bound}
\label{cor:welfare-lp-errors}
Consider the liquid welfare LP. 
Let $\iota,\sigma\in\{0,1\}$ indicate whether $N_1\neq\emptyset$ and whether $W_i=P_i$ is imposed for every $i$, respectively.
Then 
\begin{equation}\label{eq:lp-welfare-certificate}
\csname ltx@label\endcsname{eq:lp-welfare-certificate-pyb}
\OPT(I)\leq\alpha(I)\LP, 
\quad\text{where}\quad
 \alpha(I):=\max_{i\in N}\max\left\{
 \frac1{\lambda_i}+(1-\sigma)\nu_i,\;\kappa_i+\nu_i+\iota
 \right\}.
\end{equation}
\end{restatable}

\begin{proof}
Consider any feasible solution to~\eqref{eq:ultimate-LP-simple}. 
Summing the first constraints of the LP over $i\in N_1$ and the deviation constraints divided by $\lambda_i>0$ over $i\in N_2$, we obtain
\[
 \OPT(I)=
 \sum_{i\in N} 
 V_i
 \leq\sum_{i\in N_1}LW_i
 +\sum_{i\in N_2}\left(\frac{LW_i-P_i}{\lambda_i} +\frac{\mu_i}{\lambda_i}C_i\right).
\]
Using $C_i\geq0$ and the coupling constraint \eqref{eq:coupling-constraint} gives
\begin{equation}\label{eq:normalized-poa-sum}
 \OPT(I)\leq\sum_{i\in N_1}LW_i
 +\sum_{i\in N_2}\frac{LW_i-P_i}{\lambda_i}
 +\sum_{i\in N}\kappa_iP_i+\sum_{i\in N}\nu_iW_i.
\end{equation}
Since $0\leq P_i\leq W_i\leq LW_i$, the contribution of $i\in N_1$ to the right-hand side is at most $(1+\kappa_i+\nu_i)LW_i$. By the definition of $\sigma$, we also have $W_i\leq\sigma P_i+(1-\sigma)LW_i$. Thus, for $i\in N_2$, the contribution to the right-hand side is at most
\[
 \begin{aligned}
 &\left(\frac1{\lambda_i}+(1-\sigma)\nu_i\right)LW_i
 +\left(\kappa_i+\sigma\nu_i-\frac1{\lambda_i}\right)P_i
 \leq\max\left\{\frac1{\lambda_i}+(1-\sigma)\nu_i,\;\kappa_i+\nu_i\right\}LW_i,
 \end{aligned}
\]
where we used $0\leq P_i\leq LW_i$. If $N_1\neq\emptyset$, then $\iota=1$; hence $\alpha(I)$ bounds the coefficients for both $N_1$ and $N_2$. We conclude that $\OPT(I)\leq\alpha(I)\sum_i LW_i$. Minimizing over feasible solutions proves the claim.
\end{proof}

\paragraph{Extension Theorem.}

The equilibrium statistics derived above give a feasible solution. Proposition~\ref{prop:welfare-lp-bound} therefore yields the following theorem.

\begin{restatable}[CCE-POA Extension Theorem]{theorem}{thmultimatePOA}\label{thm:ultimate-POA}
Let $\mech$ be a mechanism and $\cI$ a class of its instances. For each $I\in\cI$, choose an optimal profile $\vchi^*$, charges $(\charge_i)_{i\in N}$, and coefficients $\lambda_i>0$, $\mu_i,\kappa_i,\nu_i\geq0$. Suppose every bidder is budget-feasible $(\lambda_i,\mu_i)$-semi-smooth and the coupling constraint \eqref{eq:coupling-constraint} holds for every CCE of $I$. Let $\iota=0$ if all instances in $\cI$ are budget-free and $\iota=1$ otherwise. 
\begin{itemize}\itemsep0pt
 \item For PYB mechanisms: \qquad\qquad
 \begin{equation}\label{eq:normalized-pyb-poa}
 \CCEPOA(\cI)\leq\sup_{I\in\cI}\max_{i\in N}\max\left\{
 \frac1{\lambda_i},\;\kappa_i+\nu_i+\iota
 \right\}.
 \end{equation}
 
 \item For mechanisms under NOB: 
 \begin{equation}\label{eq:normalized-nob-poa}
 \CCEPOA(\cI)\leq\sup_{I\in\cI}\max_{i\in N}\max\left\{
 \frac1{\lambda_i}+\nu_i,\;\kappa_i+\nu_i+\iota
 \right\}.
 \end{equation}
\end{itemize}
\end{restatable}
\begin{proof}%
Fix $I\in\cI$ and a NOB-feasible CCE $\vA$. As shown above, its statistics satisfy the constraints in \eqref{eq:ultimate-LP-simple}, so $\LP\leq\LW(I,\vA)$. In budget-free instances, $N_1=\emptyset$; otherwise, using $\iota=1$ can only increase the factor in Proposition~\ref{prop:welfare-lp-bound}.

For PYB mechanisms, $W_i=P_i$ and Lemma~\ref{lem:IR} gives $LW_i\geq W_i$ at every CCE, so NOB holds automatically. Applying Proposition~\ref{prop:welfare-lp-bound} with $\sigma=1$ gives $\OPT(I)\leq\alpha(I)\LW(I,\vA)$. Taking the supremum over equilibria and instances proves~\eqref{eq:normalized-pyb-poa}.

For mechanisms under expected NOB, applying Proposition~\ref{prop:welfare-lp-bound} with $\sigma=0$ gives $\OPT(I)\leq\alpha(I)\LW(I,\vA)$. Taking the supremum over NOB-feasible equilibria and instances proves~\eqref{eq:normalized-nob-poa}.
\end{proof}

The following simplified bounds follow directly from Theorem~\ref{thm:ultimate-POA}. 
For common $(\lambda,\mu)$, the coupling $\sum_iC_i\leq\sum_iP_i$ gives $\kappa_i=\mu/\lambda$ and $\nu_i=0$. For PYB mechanisms, Theorem~\ref{thm:ultimate-POA} then yields
\[
 \CCEPOA(\cI)\leq\frac{\max\{1,\mu+\lambda\iota\}}{\lambda}.
\]
The coupling $\sum_iC_i\leq\sum_iW_i$ instead gives $\kappa_i=0$ and $\nu_i=\mu/\lambda$, hence
\[
 \CCEPOA(\cI)\leq\frac{\mu+\max\{1,\lambda\iota\}}{\lambda}.
\]

\subsection{From Conservative Deviations to Budget Feasibility}\label{sec:reduction-smoothness}
Theorem~\ref{thm:ultimate-POA} takes budget-feasible semi-smoothness as its hypothesis, but our per-format proofs in Sections~\ref{sec:first-price} to~\ref{sec:upa} establish classical, budget-free smoothness properties. In this section, we introduce these classical notions (Definition~\ref{def:smoothness}) and prove a reduction from them to budget-feasible semi-smoothness (Theorem~\ref{thm:smoothness-reduction}): whenever the classical smoothness deviations are \emph{conservative}, i.e., pay in expectation at most the expected value of the distribution they target, they can be converted into budget-feasible ones.

\begin{definition}[Smoothness]\label{def:smoothness}
Fix a mechanism $\mech = (X, P)$ and an instance $I$. Consider a charging function $\charge_i:\cX_i\times\cA_{-i} \to \mathbb R_{\geq0}$ for each $i\in N$, and a target distribution $\chi_i \in \Delta(\cX_i)$.
\begin{itemize}
    \item Bidder $i \in N$ is \emph{$(\lambda, \mu)$-smooth w.r.t.~$\chi_i$} if there exists a (possibly mixed) deviation $\A_i' \in \Delta(\cA_i)$ such that, for every $\va_{-i} \in \cA_{-i}$,
    $$\myexp{a_i'\sim A_i'}{}{\z_i(a_i', \va_{-i})} \geq \lambda \myexp{x_i \sim \chi_i}{}{v_i(x_i)} - \mu \myexp{x_i \sim \chi_i}{}{\charge_i(x_i,\va_{-i})}.$$
    Bidder $i$ is \emph{conservatively $(\lambda,\mu)$-smooth w.r.t.~$\chi_i$} if some such deviation $A_i'$ (possibly mixed) also satisfies $\myexp{a_i'\sim A_i'}{}{P_i(a_i',\va_{-i})} \leq \myexp{x_i \sim \chi_i}{}{v_i(x_i)}$ for every $\va_{-i} \in \cA_{-i}$.
    \item Bidder $i \in N$ is \emph{$(\lambda, \mu)$-semi-smooth w.r.t.~$\chi_i$} if for every distribution of action profiles $\vA \in \Delta(\cA)$ there exists a deviation $\A_i' \in \Delta(\cA_i)$, independent of the realizations of $\vA$, such that
        $$\myexp{a_i'\sim A_i'}{\va_{-i}\sim \vA_{-i}}{\z_i(a_i', \va_{-i})} \geq \lambda \myexp{x_i \sim \chi_i}{}{v_i(x_i)} - \mu \myexp{x_i \sim \chi_i}{\va_{-i}\sim \vA_{-i}}{\charge_i(x_i,\va_{-i})}.$$
        Bidder $i$ is \emph{conservatively $(\lambda, \mu)$-semi-smooth w.r.t.~$\chi_i$} if, for every $\vA \in \Delta(\cA)$, some such deviation $A_i'$ also satisfies $\myexp{a_i'\sim A_i'}{\va_{-i}\sim \vA_{-i}}{P_i(a_i',\va_{-i})} \leq \myexp{x_i \sim \chi_i}{}{v_i(x_i)}$.
\end{itemize}
Throughout the paper, we do not refer to $\chi_i$ if a property holds w.r.t.~every $\chi_i \in \Delta(\cX_i)$ and we write w.r.t.~$x_i$ for the point mass on a target allocation $x_i \in \cX_i$.
\end{definition}
The first notion in Definition~\ref{def:smoothness} is the classical smoothness of~\cite{Rou15, ST13}, stated per bidder: the deviation is a single (possibly mixed) action that performs well against \emph{every} action profile of the opponents. The second notion, semi-smoothness, allows the deviation to depend on the distribution of the opponents' actions, similarly to the relaxed smoothness of~\cite{dutting2018valuation}. This added flexibility is what allows us to capture proofs in which the deviation is sampled from the opponents' bid distribution.\footnote{An example is the argument for subadditive valuations of~\cite{feldmanfu}, which we extend to the budgeted setting in Lemma~\ref{lem:SAPA-smoothness}.}

Compared to the literature, both notions differ in two respects, and each difference is needed for the reduction below. First, smoothness for bidder $i \in N$ is stated with respect to an arbitrary target distribution $\chi_i \in \Delta(\cX_i)$ rather than with respect to the optimal one: our reduction constructs the budget-feasible deviation from a deviation for the marginal $\chi^{*}_i$ of $\vchi^*$, which can be obtained by sampling a target allocation from $\vchi^\ast$ and playing the deviation for that allocation. Second, conservativeness requires that the deviation pays in expectation at most the expected value of the \emph{targeted} distribution, in contrast to the conservative smoothness of~\cite{ST13}, where payments are only capped by the bidder's maximum value. This cap is what makes the sampled deviation affordable after rescaling. Theorem~\ref{thm:smoothness-reduction} below then converts conservative semi-smoothness into the budget-feasible semi-smoothness required by Theorem~\ref{thm:ultimate-POA}.

\begin{theorem}\label{thm:smoothness-reduction}
Fix a mechanism $\mech$, an instance $I$, an optimal profile $\vchi^*$ and a charging function $\charge_i$ for every bidder $i \in N$. Every bidder $i \in N$ that is conservatively $(\lambda,\mu)$-semi-smooth w.r.t.~$\chi_i^*$ for the charge $\charge_i$ is budget-feasible $(\lambda,\mu)$-semi-smooth for the charge $\charge_i$.
\end{theorem}
\begin{proof}
Fix a bidder $i \in N$ that is conservatively $(\lambda,\mu)$-semi-smooth w.r.t.~$\chi_i^*$ for the charge $\charge_i$, fix $\vA \in \Delta(\cA)$, and let $A_i'$ be the deviation of Definition~\ref{def:smoothness} for $\vA$. If $\hat v_i(\chi_i^*) = 0$, then the opt-out action $\emptyset_i$, which is in $\cC_i(\vA_{-i})$, satisfies the inequality of Definition~\ref{def:semi-smoothness}, as its utility is non-negative while the right-hand side is at most $0$. So suppose that $\hat v_i(\chi_i^*) > 0$, and let $\hat A_i'$ be the deviation that plays $A_i'$ with probability $q_i := \min\bigl\{1, \frac{\hat v_i(\chi_i^*)}{\myexp{x_i\sim \chi_i^*}{}{v_i(x_i)}}\bigr\}$ and opts out otherwise. Then $\hat A_i' \in \cC_i(\vA_{-i})$, as
\[
\myexp{\hat a_i'\sim \hat A_i'}{\va_{-i}\sim \vA_{-i}}{P_i(\hat a_i',\va_{-i})} = q_i\myexp{a_i'\sim  A_i'}{\va_{-i}\sim \vA_{-i}}{P_i(a_i',\va_{-i})} \leq q_i\myexp{x_i\sim \chi_i^*}{}{v_i(x_i)} = \hat v_i(\chi_i^*) \leq \budget_i ,
\]
where the first inequality follows by conservativeness, and the second equality by the definition of $q_i$, since $\hat v_i(\chi_i^*) \leq \myexp{x_i\sim \chi_i^*}{}{v_i(x_i)}$. For the expected utility, we obtain
\begin{align*}
\myexp{\hat a_i'\sim \hat A_i'}{\va_{-i}\sim \vA_{-i}}{\z_i(\hat a_i',\va_{-i})}&\geq q_i\myexp{a_i'\sim A_i'}{\va_{-i}\sim \vA_{-i}}{\z_i(a_i',\va_{-i})}\\
&\geq q_i\left(\lambda \myexp{x_i\sim \chi_i^*}{}{v_i(x_i)} - \mu \myexp{x_i\sim \chi_i^*}{\va_{-i}\sim \vA_{-i}}{\charge_i(x_i,\va_{-i})}\right)\\
&=\lambda \hat v_i(\chi_i^*)-q_i\mu \myexp{x_i\sim \chi_i^*}{\va_{-i}\sim \vA_{-i}}{\charge_i(x_i,\va_{-i})}\\
&\geq \lambda\hat v_i(\chi_i^*)-\mu\myexp{x_i\sim \chi_i^*}{\va_{-i}\sim \vA_{-i}}{\charge_i(x_i,\va_{-i})},
\end{align*}
where the first inequality holds as opting out gives non-negative utility, and the second inequality is the semi-smoothness inequality of Definition~\ref{def:smoothness}. The equality then follows by the definition of $q_i$, and the last inequality as $q_i \leq 1$ and charges are non-negative. Hence $\hat A_i'$ satisfies Definition~\ref{def:semi-smoothness}.
\end{proof}

Conservative $(\lambda,\mu)$-smoothness w.r.t.~$\chi_i$ implies conservative $(\lambda,\mu)$-semi-smoothness w.r.t.~$\chi_i$, by taking expectations over $\va_{-i} \sim \vA_{-i}$. Hence proving either property w.r.t.~$\chi_i^*$ suffices to apply Theorem~\ref{thm:ultimate-POA}. In Sections~\ref{sec:first-price} to~\ref{sec:gsp}, we prove these properties w.r.t.~every allocation $x_i \in \cX_i$ instead, which suffices as well.

\begin{corollary}\label{cor:smoothness-reduction}
Fix a mechanism $\mech$, an instance $I$, an optimal profile $\vchi^*$ and a charging function $\charge_i$ for every bidder $i \in N$. Every bidder $i \in N$ that is conservatively $(\lambda,\mu)$-smooth or conservatively $(\lambda,\mu)$-semi-smooth w.r.t.~every $x_i \in \cX_i$ for the charge $\charge_i$ is budget-feasible $(\lambda,\mu)$-semi-smooth for the charge $\charge_i$.
\end{corollary}
\begin{proof}
Fix such a bidder $i \in N$ and $\vA \in \Delta(\cA)$, and for every $x_i \in \cX_i$ let $A_i^{x_i}$ be her deviation of Definition~\ref{def:smoothness} w.r.t.~$x_i$, for $\vA$ in the semi-smooth case. Consider the deviation that draws $x_i \sim \chi_i^*$ and plays $A_i^{x_i}$. It is independent of the realizations of $\vA$, and taking expectations over $x_i \sim \chi_i^*$, and over $\va_{-i} \sim \vA_{-i}$ in the smooth case, in the two inequalities of Definition~\ref{def:smoothness} for $x_i$ shows that bidder $i$ is conservatively $(\lambda,\mu)$-semi-smooth w.r.t.~$\chi_i^*$ for the charge $\charge_i$. The claim follows by Theorem~\ref{thm:smoothness-reduction}.
\end{proof}
\begin{remark}\label{rem:wlog-positive-value}
    When proving (semi-)smoothness with respect to a distribution $\chi_i$, we assume without loss of generality that $\myexp{x_i\sim\chi_i}{}{v_i(x_i)} > 0$. Indeed, suppose $\myexp{x_i\sim\chi_i}{}{v_i(x_i)} = 0$ and consider the deviation $\A_i' = \emptyset_i$. The left-hand side of the smoothness inequality in Definition~\ref{def:smoothness} then satisfies $\myexp{}{}{\z_i(\emptyset_i, \va_{-i})} = \myexp{}{}{v_i(X_i(\emptyset_i, \va_{-i}))} \ge 0$, since the opt-out action induces no payment and valuations are non-negative, whereas the right-hand side equals $\lambda\, \myexp{x_i\sim\chi_i}{}{v_i(x_i)} - \mu\, \myexp{x_i\sim\chi_i}{}{\charge_i(x_i, \va_{-i})} = -\mu\, \myexp{x_i\sim\chi_i}{}{\charge_i(x_i, \va_{-i})} \le 0$, since charges are non-negative. The conservativeness requirement $\myexp{x_i\sim\chi_i}{}{v_i(x_i)} \ge P_i(\emptyset_i, \va_{-i})$ also holds, with both sides equal to $0$.
\end{remark}

\Needspace{5\baselineskip}
\section{Auctions with First-Price Payments}\label{sec:first-price}

In this section, we consider the pay-your-bid formats of Section~\ref{sec:auction-formats}: simultaneous first-price auctions with XOS valuations, simultaneous first-price auctions with additive valuations under uniform-multiplier bidding, discriminatory auctions with concave valuations, and GFP auctions. For the remainder of this section, we refer to these four settings as the first-price formats. Our main result is the following theorem.

\begin{theorem}\label{thm:first-price}
The PoA of CCE of each of the four first-price formats is $\frac{e}{e-1} \approx 1.582$ for budget-free instances and
$
\frac{2 + \mathcal{W}_0(-e^{-2})}{1 + \mathcal{W}_0(-e^{-2})} \approx 2.188
$
with budgets, where $\mathcal{W}_0$ is the principal branch of the Lambert $W$ function (see Section~\ref{sec:preliminaries}).
\end{theorem}

The first-price formats are PYB (see Definition~\ref{def:pyb}), so their equilibria are NOB-feasible by Lemma~\ref{lem:IR}, and Theorem~\ref{thm:first-price} needs no no-overbidding assumption. Both bounds are tight. In Appendix~\ref{sec:first-price-lb}, we prove matching lower bounds with budgets: Theorem~\ref{thm:lb-fpa} for first-price auctions with a single item, which covers simultaneous first-price auctions with XOS valuations, discriminatory auctions and GFP auctions, and Theorem~\ref{thm:lb-uniform-bids} for simultaneous first-price auctions under uniform-multiplier bidding, with budgets as well as for budget-free instances, already for mixed Nash equilibria. For budget-free instances, the bound $\frac{e}{e-1}$ is tight for CCE already for a single item, by an example of Syrgkanis~\cite{syrgkanis-thesis}.

The budget-free bound $\frac{e}{e-1}$ for simultaneous first-price auctions with XOS valuations is due to~\cite{ST13}, to~\cite{dekeijzer13} for discriminatory auctions, and to~\cite{hartline14} for Bayes--Nash equilibria of GFP auctions. With budgets, Theorem~\ref{thm:first-price} recovers, for simultaneous first-price auctions with XOS valuations, the bound of~\cite{baldeschi2026}, and extends it to the three other formats. Under uniform-multiplier bidding, the prior bounds concern value maximizers, i.e., bidders who maximize their value subject to a return-on-spend constraint. By~\cite{deng21}, every Nash equilibrium of budget-free value maximizers is efficient, and by~\cite{LMZ24}, with budgets the PoA grows linearly in the number of bidders, in contrast to the bound $2.188$ of Theorem~\ref{thm:first-price} for utility maximizers with budgets in expectation.

To prove Theorem~\ref{thm:first-price}, we first show in Section~\ref{sec:first-price-representation} that the four formats share a common representation, in which one property of valuations, which we call $\phi$-XOS, covers the valuations of all four formats. In this representation, we define the charge of a target allocation as the opposing bids that it displaces. In Section~\ref{sec:first-price-deviation}, we define an appropriate randomized deviation, and we show a conservative semi-smoothness guarantee (see Lemma~\ref{lem:first-price-smoothness}). Then, in Section~\ref{sec:first-price-coupling}, we verify the coupling constraint~\eqref{eq:coupling-constraint} for these charges in every format, and Theorem~\ref{thm:ultimate-POA} gives the bounds of Theorem~\ref{thm:first-price}. The matching lower bounds can be found in Appendix~\ref{sec:first-price-lb}.

\subsection{A Common Perspective under XOS Valuations}\label{sec:first-price-representation}

In each first-price format, we represent the action of bidder $i \in N$ by a non-negative vector $a_i$ and every allocation $x_i \in \cX_i$ by a vector $\phi_i(x_i)$ of the same dimension with entries in $[0,1]$, so that bidder $i$ pays her bid for the allocation she receives, i.e., $P_i(\va) = a_i \cdot \phi_i(X_i(\va))$.
\begin{itemize}
\item In simultaneous first-price auctions, $a_i \in \mathbb R^m_{\geq 0}$ is the bid vector, and $\phi_i(x_i) \in \set{0,1}^m$ is the indicator vector of the bundle $x_i$. Under uniform-multiplier bidding, the bid multiplier $\alpha_i$ is represented by the bid vector $a_i = (\alpha_i v_{ij})_{j \in M}$.
\item In discriminatory auctions, the action $(b_i, q_i)$ is represented by the vector $a_i \in \mathbb R^k_{\geq 0}$ with $q_i$ entries equal to $b_i$ followed by zeros, and $\phi_i(x_i) \in \set{0,1}^k$ has $x_i$ entries equal to one followed by zeros.
\item In GFP auctions, $a_i \geq 0$ is the bid, and $\phi_i(x_i) = \alpha_{x_i}$ is the click-through rate of slot $x_i$, with $\phi_i(0) = 0$.
\end{itemize}
In all four formats, the allocation rule maximizes the total bid, i.e., $X(\va) \in \argmax_{\vx \in \cX} \sum_{i \in N} a_i \cdot \phi_i(x_i)$ with ties broken consistently, and $t a_i$ is an action for every action $a_i$ and every $t \in [0,1]$. We write $\phi_i(y_i) \geq \phi_i(x_i)$ for the componentwise inequality, which means that $y_i$ contains the bundle $x_i$, supplies at least $x_i$ units, or is a slot with at least the click-through rate of slot $x_i$.

For valuations, we use the following property.

\begin{definition}\label{def:phi-xos}
A valuation $v_i$ is \emph{$\phi$-XOS} if for every $x_i \in \cX_i$ there exists a \emph{supporting bid} $a_i^{x_i} \in \cA_i$ such that $a_i^{x_i} \cdot \phi_i(x_i) = v_i(x_i)$ and $a_i^{x_i} \cdot \phi_i(y_i) \leq v_i(y_i)$ for all $y_i \in \cX_i$.
\end{definition}
Intuitively, the supporting bid $a_i^{x_i}$ is a bid under which bidder $i$ pays her value $v_i(x_i)$ if she receives $x_i$, and never more than her value for the allocation she receives. For simultaneous first-price auctions, this is the XOS condition of Section~\ref{sec:auction-formats}, with $a_i^{x_i} = (v_{ij}^{x_i})_{j \in M}$, and under uniform-multiplier bidding the bid multiplier $1$ supports every bundle. For multi-unit auctions, this captures a valuation class which contains the concave valuations, with $a_i^{x_i}$ consisting of $x_i$ copies of the per-unit bid $\frac{v_i(x_i)}{x_i}$. For position auctions, the scalar bid $\bar v_i$ supports every slot.

For the charge, we use the Clarke payment~\cite{clarke71}, which we express using the $\phi$-notation from above.\footnote{The same comparison of pay-your-bid payments with Clarke payments is used in~\cite{DKT21}.} For every $i \in N$, $x_i \in \cX_i$ and $\va_{-i} \in \cA_{-i}$, let
\begin{equation}\label{eq:first-price-charge}
\charge_i(x_i, \va_{-i}) := \max_{y\in \cX}\sum_{\ell\in N\setminus\{i\}}a_\ell \cdot \phi_\ell(y_\ell)-\max_{y\in \cX:\phi_i(y_i)\geq \phi_i(x_i)}\sum_{\ell\in N\setminus\{i\}}a_\ell\cdot  \phi_\ell(y_{\ell}) .
\end{equation}
Intuitively, $\charge_i(x_i, \va_{-i})$ is the total bid that the other bidders lose when bidder $i$ is guaranteed an allocation containing $x_i$, i.e., the opposing bids that $x_i$ displaces.

\subsection{A Randomized Deviation}\label{sec:first-price-deviation}

We are now ready to prove the smoothness lemma.

\begin{lemma}\label{lem:first-price-smoothness}
    In every first-price format, every bidder $i \in N$ is conservatively $(\mu(1-e^{-1/\mu}),\mu)$-semi-smooth w.r.t.~every $x_i \in \cX_i$ for the charge~\eqref{eq:first-price-charge}, for every $\mu > 0$.\end{lemma}
\begin{proof}
    Fix $\mu > 0$, a bidder $i \in N$, an allocation $x_i\in \cX_i$ and a distribution $\vA \in \Delta(\cA)$, and consider the deviation $A_i'$ where bidder $i$ bids $t\cdot a_i^{x_i}$, where $t$ is sampled from $[0, 1-e^{-1/\mu}]$ with probability density function $f(t)=\frac{\mu}{1-t}$, which integrates to one as $\int_0^{1-e^{-1/\mu}} \frac{\mu}{1-t}\, dt = -\mu \ln e^{-1/\mu} = 1$. For conservativeness, we let $\hat A_i$ denote the deviation where bidder $i$ samples $a_i'=t\cdot a_i^{x_i}$ from $A_i'$ and bids it with probability
    $$q_i(t):=\min\left\{1,\frac{v_i(x_i)}{\myexp{\va_{-i}\sim \vA_{-i}}{}{P_i(t\cdot a_i^{x_i},\va_{-i})}}\right\},$$
    where $q_i(t) := 1$ if the expected payment is zero, and opts out otherwise. Then
    $$\myexp{\hat a_i\sim \hat A_i}{\va_{-i}\sim \vA_{-i}}{P_i(\hat a_i,\va_{-i})}=\int_0^{1-e^{-1/\mu}} q_i(t)\,\myexp{\va_{-i}\sim \vA_{-i}}{}{P_i(t\cdot a_i^{x_i},\va_{-i})}\, f(t)\, dt\leq v_i(x_i),$$
    showing that $\hat A_i$ is conservative.

    By Definition~\ref{def:phi-xos}, we have $v_i(X_i(t\cdot a_i^{x_i},\va_{-i}))\geq a_i^{x_i}\cdot \phi_i(X_i(t\cdot a_i^{x_i},\va_{-i}))$. Given a fixed $t$, let $\hat A_i(t)$ denote the distribution of bidding $t\cdot a_i^{x_i}$ with probability $q_i(t)$ and opting out otherwise. Conditional on $t$, we see that
    $$\begin{aligned}\myexp{\hat a_i\sim \hat A_i(t)}{\va_{-i}\sim \vA_{-i}}{\z_i(\hat a_i,\va_{-i})}&\geq q_i(t)\myexp{\va_{-i}\sim \vA_{-i}}{}{v_i(X_i(t\cdot a_i^{x_i},\va_{-i}))-P_i(t\cdot a_i^{x_i},\va_{-i})}\\
    &\geq q_i(t)\myexp{\va_{-i}\sim \vA_{-i}}{}{a_i^{x_i}\cdot \phi_i(X_i(t\cdot a_i^{x_i},\va_{-i}))-P_i(t\cdot a_i^{x_i},\va_{-i})}\\
    &=q_i(t)(1-t)\myexp{\va_{-i}\sim \vA_{-i}}{}{a_i^{x_i}\cdot \phi_i(X_i(t\cdot a_i^{x_i},\va_{-i}))}\\
    &\geq (1-t)\min\left\{\myexp{\va_{-i}\sim \vA_{-i}}{}{a_i^{x_i}\cdot \phi_i(X_i(t\cdot a_i^{x_i},\va_{-i}))},v_i(x_i)\right\}\\
    &\geq (1-t)\myexp{\va_{-i}\sim \vA_{-i}}{}{\min\left\{a_i^{x_i}\cdot \phi_i(X_i(t\cdot a_i^{x_i},\va_{-i})),v_i(x_i)\right\}}.\end{aligned}$$
    Here, the first inequality holds as opting out gives non-negative utility, and the second inequality follows by Definition~\ref{def:phi-xos}. The equality then follows as $P_i(t\cdot a_i^{x_i},\va_{-i}) = t\, a_i^{x_i}\cdot \phi_i(X_i(t\cdot a_i^{x_i},\va_{-i}))$, and the third inequality by the definition of $q_i(t)$ and $t \leq 1-e^{-1/\mu} < 1$. Finally, the last inequality follows by Jensen's inequality, as the minimum is concave. Integrating over $t$, where the density $f$ cancels the factor $1 - t$, we see that
    $$\begin{aligned}\myexp{\hat a_i\sim \hat A_i}{\va_{-i}\sim \vA_{-i}}{\z_i(\hat a_i,\va_{-i})}&\geq \mu\int_0^{1-e^{-1/\mu}}\myexp{\va_{-i}\sim \vA_{-i}}{}{\min\left\{a_i^{x_i}\cdot \phi_i(X_i(t\cdot a_i^{x_i},\va_{-i})),v_i(x_i)\right\}}dt\\
    &=\mu(1-e^{-1/\mu})v_i(x_i)\\
    &\quad-\mu\myexpl{\va_{-i}\sim \vA_{-i}}{}{\int_0^{1-e^{-1/\mu}}\left(v_i(x_i)-a_i^{x_i}\cdot \phi_i(X_i(t\cdot a_i^{x_i},\va_{-i}))\right)_+dt}.\end{aligned}$$
    We now bound the second term pointwise in $\va_{-i}$, and call it $Z$. The value $a_i^{x_i}\cdot \phi_i(X_i(s\cdot a_i^{x_i},\va_{-i}))$ is non-decreasing in $s$, since for $s < s'$, applying the allocation rule to the profiles $(s\cdot a_i^{x_i},\va_{-i})$ and $(s'\cdot a_i^{x_i},\va_{-i})$ and adding the two resulting inequalities gives $(s' - s)\bigl(a_i^{x_i}\cdot \phi_i(X_i(s'\cdot a_i^{x_i},\va_{-i})) - a_i^{x_i}\cdot \phi_i(X_i(s\cdot a_i^{x_i},\va_{-i}))\bigr) \geq 0$. Hence the integrand of $Z$ is non-increasing in $t$, and rewriting gives
    \begin{align}\label{eq:Z-rewrite}Z&=\int_0^{1-e^{-1/\mu}}\left(v_i(x_i)-a_i^{x_i}\cdot \phi_i(X_i(t\cdot a_i^{x_i},\va_{-i}))\right)_+dt\nonumber\\
    &=\max_{0\leq t\leq 1-e^{-1/\mu}}\int_0^t \bigl(v_i(x_i)-a_i^{x_i}\cdot \phi_i(X_i(s\cdot a_i^{x_i},\va_{-i}))\bigr)\, ds\nonumber\\
    &=\max_{0\leq t\leq 1-e^{-1/\mu}}\left(tv_i(x_i)-\int_0^ta_i^{x_i}\cdot \phi_i(X_i(s\cdot a_i^{x_i},\va_{-i}))ds\right).
    \end{align}
    For the integral term, the allocation rule gives
    $$sa_i^{x_i}\cdot \phi_i(X_i(s\cdot a_i^{x_i},\va_{-i}))=\max_{y\in \cX}\left(sa_i^{x_i}\cdot \phi_i(y_i)+\sum_{\ell\neq i}a_{\ell}\cdot \phi_{\ell}(y_{\ell})\right)- \sum_{\ell\neq i}a_{\ell}\cdot \phi_\ell(X_{\ell}(s\cdot a_i^{x_i},\va_{-i})).$$
    The maximum on the right-hand side is a piecewise linear function of $s$, and $\phi_\ell(X_\ell(s\cdot a_i^{x_i},\va_{-i}))$ is piecewise constant in $s$, so differentiating gives, for all but finitely many $s$,
    $$a_i^{x_i}\cdot \phi_i(X_i(s\cdot a_i^{x_i},\va_{-i}))=\frac{d}{ds}\max_{y\in \cX}\left(sa_i^{x_i}\cdot \phi_i(y_i)+\sum_{\ell\neq i}a_{\ell}\cdot \phi_{\ell}(y_{\ell})\right).$$
    Integrating from $s=0$ to $s=t$ and continuing from~\eqref{eq:Z-rewrite} then gives, with all maxima over $t$ taken over $0 \leq t \leq 1-e^{-1/\mu}$,
    $$\begin{aligned}Z&=\max_{t}\left(tv_i(x_i)+\max_{y\in \cX}\sum_{\ell\neq i}a_{\ell}\cdot \phi_{\ell}(y_{\ell})-\max_{y\in \cX}\left(ta_i^{x_i}\cdot \phi_i(y_i)+\sum_{\ell\neq i}a_{\ell}\cdot \phi_{\ell}(y_{\ell})\right)\right)\\
    &\leq \max_{t}\left(tv_i(x_i)+\max_{y\in \cX}\sum_{\ell\neq i}a_{\ell}\cdot \phi_{\ell}(y_{\ell})-\max_{y\in \cX:\phi_i(y_i)\geq \phi_i(x_i)}\left(ta_i^{x_i}\cdot \phi_i(y_i)+\sum_{\ell\neq i}a_{\ell}\cdot \phi_{\ell}(y_{\ell})\right)\right)\\
    &\leq \max_{t}\left(tv_i(x_i)+\max_{y\in \cX}\sum_{\ell\neq i}a_{\ell}\cdot \phi_{\ell}(y_{\ell})-\max_{y\in \cX:\phi_i(y_i)\geq \phi_i(x_i)}\left(ta_i^{x_i}\cdot \phi_i(x_i)+\sum_{\ell\neq i}a_{\ell}\cdot \phi_{\ell}(y_{\ell})\right)\right)\\
    &= \max_{t}\left(tv_i(x_i)+\max_{y\in \cX}\sum_{\ell\neq i}a_{\ell}\cdot \phi_{\ell}(y_{\ell})-\max_{y\in \cX:\phi_i(y_i)\geq \phi_i(x_i)}\left(t v_i(x_i)+\sum_{\ell\neq i}a_{\ell}\cdot \phi_{\ell}(y_{\ell})\right)\right)\\
    &=\max_{y\in \cX}\sum_{\ell\neq i}a_{\ell}\cdot \phi_{\ell}(y_{\ell})-\max_{y\in \cX:\phi_i(y_i)\geq \phi_i(x_i)}\sum_{\ell\neq i}a_{\ell}\cdot \phi_{\ell}(y_{\ell}) = \charge_i(x_i,\va_{-i}),
    \end{aligned}$$
    where the first inequality follows by restricting the last maximum to the allocation profiles $y$ with $\phi_i(y_i)\geq \phi_i(x_i)$, the second inequality as $a_i^{x_i}\cdot \phi_i(y_i)\geq a_i^{x_i}\cdot \phi_i(x_i)$ holds for these profiles, and the second equality since $v_i$ is $\phi$-XOS. Together with the bound on the expected utility above, this gives the semi-smoothness inequality with $\lambda = \mu(1-e^{-1/\mu})$ and the charge~\eqref{eq:first-price-charge}.
\end{proof}

\subsection{The Coupling Constraint}\label{sec:first-price-coupling}

We now verify the coupling constraint~\eqref{eq:coupling-constraint} for the charges defined in~\eqref{eq:first-price-charge}.

\begin{lemma}\label{lem:first-price-coupling}
In every first-price format, for every profile $\va \in \cA$ and every allocation profile $\vx \in \cX$,
$$\sum_{i\in N}\charge_i(x_i,\va_{-i})\leq \sum_{i\in N}P_i(\va).$$
\end{lemma}
\begin{proof}
Fix $\va \in \cA$ and $\vx \in \cX$. We distinguish three cases, based on the format.

\medskip
\noindent \textbf{Case 1:} simultaneous first-price auctions, where under uniform-multiplier bidding $\va$ is the induced bid profile. Let $i \in N$. Without bidder $i$, every item $j \in M$ goes to the highest opposing bid $\beta_j(\va_{-i})$, and if bidder $i$ receives a bundle containing $x_i$, the other bidders receive at most the items outside $x_i$. Hence the two maxima in~\eqref{eq:first-price-charge} equal $\sum_{j \in M} \beta_j(\va_{-i})$ and $\sum_{j \in M \setminus x_i} \beta_j(\va_{-i})$, and
\[
\charge_i(x_i, \va_{-i}) = \sum_{j \in x_i} \beta_j(\va_{-i}) .
\]
Every item $j \in M$ is awarded to a bidder with the highest bid $\beta_j(\va)$ on it, who pays this bid, and the bundles $x_1, \dots, x_n$ are pairwise disjoint. Hence
\[
\sum_{i \in N} \charge_i(x_i, \va_{-i}) = \sum_{i \in N} \sum_{j \in x_i} \beta_j(\va_{-i}) \leq \sum_{i \in N} \sum_{j \in x_i} \beta_j(\va) \leq \sum_{j \in M} \beta_j(\va) = \sum_{i \in N} P_i(\va) ,
\]
where the first inequality follows as $\beta_j(\va_{-i}) \leq \beta_j(\va)$, and the second as the bundles are pairwise disjoint and bids are non-negative.

\medskip
\noindent \textbf{Case 2:} discriminatory auctions. Let $i \in N$. Without bidder $i$, the $k$ highest opposing marginal bids $\beta_1(\va_{-i}) \leq \dots \leq \beta_k(\va_{-i})$ win, and if bidder $i$ receives at least $x_i$ units, the other bidders win at most $k - x_i$ units. Hence the two maxima in~\eqref{eq:first-price-charge} equal $\sum_{j=1}^{k} \beta_j(\va_{-i})$ and $\sum_{j = x_i + 1}^{k} \beta_j(\va_{-i})$, and
\[
\charge_i(x_i, \va_{-i}) = \sum_{j=1}^{x_i} \beta_j(\va_{-i}) .
\]
The winning marginal bids $\beta_1(\va) \leq \dots \leq \beta_k(\va)$ are paid by the bidders who submit them, so $\sum_{i \in N} P_i(\va) = \sum_{j=1}^{k} \beta_j(\va)$, and removing the marginal bids of bidder $i$ does not increase any of them, so $\beta_j(\va_{-i}) \leq \beta_j(\va)$ holds for every $j \in [k]$. Moreover, since $\sum_{\ell \in N} x_\ell \leq k$, the sets $S_i := \set{\sum_{\ell < i} x_\ell + 1, \dots, \sum_{\ell \leq i} x_\ell}$, $i \in N$, are pairwise disjoint subsets of $[k]$ with $|S_i| = x_i$, and the $j$th smallest element of $S_i$ is at least $j$. Hence
\[
\sum_{i \in N} \charge_i(x_i, \va_{-i}) = \sum_{i \in N} \sum_{j=1}^{x_i} \beta_j(\va_{-i}) \leq \sum_{i \in N} \sum_{j=1}^{x_i} \beta_j(\va) \leq \sum_{i \in N} \sum_{j \in S_i} \beta_j(\va) \leq \sum_{j=1}^{k} \beta_j(\va) = \sum_{i \in N} P_i(\va) ,
\]
where the second inequality follows as $\beta_j(\va)$ is non-decreasing in $j$, and the third as the sets $S_i$ are pairwise disjoint and bids are non-negative.

\medskip
\noindent \textbf{Case 3:} GFP auctions. Let $i \in N$, and let $\alpha_{m+1} := 0$. If $x_i = 0$, the two maxima in~\eqref{eq:first-price-charge} coincide, so $\charge_i(0, \va_{-i}) = 0$. If $x_i \geq 1$, the first maximum equals $\sum_{j=1}^{m} \alpha_j \beta_j(\va_{-i})$, as the total bid is maximized by assigning the slots to the other bidders in the order of their bids, and the second maximum is at least the total bid of the assignment that gives slot $x_i$ to bidder $i$ and the remaining slots to the other bidders in the order of their bids. Hence
\begin{align*}
\charge_i(x_i, \va_{-i}) &\leq \sum_{j=1}^{m} \alpha_j \beta_j(\va_{-i}) - \sum_{j=1}^{x_i - 1} \alpha_j \beta_j(\va_{-i}) - \sum_{j=x_i}^{m} \alpha_{j+1} \beta_j(\va_{-i}) \\
&= \sum_{j=x_i}^{m} (\alpha_j - \alpha_{j+1})\, \beta_j(\va_{-i}) \leq \alpha_{x_i} \beta_{x_i}(\va_{-i}) ,
\end{align*}
where the last inequality follows as $\beta_j(\va_{-i}) \leq \beta_{x_i}(\va_{-i})$ holds for $j \geq x_i$ and the differences $\alpha_j - \alpha_{j+1}$ are non-negative and sum to $\alpha_{x_i}$. The bidder who receives slot $j \in [m]$ bids $\beta_j(\va)$ and pays $\alpha_j \beta_j(\va)$, so $\sum_{i \in N} P_i(\va) = \sum_{j=1}^{m} \alpha_j \beta_j(\va)$, and removing the bid of bidder $i$ does not increase the $j$th highest bid, so $\beta_j(\va_{-i}) \leq \beta_j(\va)$ holds for every $j \in [m]$. Since the slots $x_i \geq 1$ are pairwise distinct,
\[
\sum_{i \in N} \charge_i(x_i, \va_{-i}) \leq \sum_{i \in N :\, x_i \geq 1} \alpha_{x_i} \beta_{x_i}(\va_{-i}) \leq \sum_{i \in N :\, x_i \geq 1} \alpha_{x_i} \beta_{x_i}(\va) \leq \sum_{j=1}^{m} \alpha_j \beta_j(\va) = \sum_{i \in N} P_i(\va) .
\]

\medskip
In all three cases, $\sum_{i \in N} \charge_i(x_i, \va_{-i}) \leq \sum_{i \in N} P_i(\va)$.
\end{proof}

We can now prove Theorem~\ref{thm:first-price}.
\begin{proof}[Proof of Theorem~\ref{thm:first-price}]
Let $\cI$ be the class of all instances of one of the four first-price formats, and fix $I \in \cI$, an optimal profile $\vchi^\ast$ and $\mu > 0$. By Lemma~\ref{lem:first-price-smoothness} and Corollary~\ref{cor:smoothness-reduction}, every bidder $i \in N$ is budget-feasible $(\mu(1-e^{-1/\mu}), \mu)$-semi-smooth for the charge~\eqref{eq:first-price-charge}. For every CCE $\vA$ of $I$, Lemma~\ref{lem:first-price-coupling} applied to every $\vx$ in the support of $\vchi^\ast$ and every $\va$ in the support of $\vA$ gives, after taking expectations, the coupling constraint~\eqref{eq:coupling-constraint} with $\lambda_i = \mu(1-e^{-1/\mu})$, $\mu_i = \mu$, $\kappa_i = \frac{\mu_i}{\lambda_i} = \frac{1}{1-e^{-1/\mu}}$ and $\nu_i = 0$ for every $i \in N$. Since the first-price formats are PYB, Theorem~\ref{thm:ultimate-POA} gives, by~\eqref{eq:normalized-pyb-poa},
$$\CCEPOA(\cI)\leq \max\left\{\frac{1}{\mu(1-e^{-1/\mu})},\ \frac{1}{1-e^{-1/\mu}} + \iota\right\} \qquad \text{for every } \mu > 0,$$
where $\iota = 0$ if $\cI$ is budget-free and $\iota = 1$ otherwise. For budget-free instances, we choose $\mu = 1$, for which both terms in the maximum equal $\frac{1}{1 - e^{-1}} = \frac{e}{e-1}$. With budgets, solving $\frac{1}{\mu(1-e^{-1/\mu})} = \frac{1}{1-e^{-1/\mu}} + 1$, i.e., $\frac{1}{\mu} = 2 - e^{-1/\mu}$, gives $\mu = \frac{1}{2 + \mathcal{W}_0(-e^{-2})} \approx 0.543$, for which $e^{-1/\mu} = -\mathcal{W}_0(-e^{-2})$ by the definition of $\mathcal{W}_0$. Plugging this in gives
$$\CCEPOA(\cI)\leq \frac{1}{1+\mathcal{W}_0(-e^{-2})} + 1 = \frac{2 + \mathcal{W}_0(-e^{-2})}{1 + \mathcal{W}_0(-e^{-2})} \approx 2.188 .$$
Together with the matching lower bounds of Appendix~\ref{sec:first-price-lb} and the budget-free example of Syrgkanis~\cite{syrgkanis-thesis}, this proves the theorem.
\end{proof}

\section{Simultaneous Item Auctions beyond XOS Valuations}\label{sec:simultaneous-auctions}

In this section, we consider simultaneous item auctions beyond XOS valuations: first-price auctions with subadditive valuations, all-pay auctions with subadditive valuations, and second-price auctions with general monotone valuations. Throughout this section, we use the charge
\begin{equation}\label{eq:item-charge}
\charge_i(x_i, \va_{-i}) := \sum_{j \in x_i} \beta_j(\va_{-i}) \qquad \text{for every } i \in N,\ x_i \in \cX_i \text{ and } \va_{-i} \in \cA_{-i} ,
\end{equation}
i.e., the sum of the highest opposing bids on the items of $x_i$. For each of the three formats, we show that every bidder is conservatively semi-smooth w.r.t.~every bundle for this charge, and Corollary~\ref{cor:smoothness-reduction} and Theorem~\ref{thm:ultimate-POA} then give the PoA bound. The matching lower bounds for the all-pay settings are given in Appendix~\ref{app:lower-bounds}.
\subsection{First-Price and All-Pay Auctions with Subadditive Valuations}\label{sec:all-pay}

\begin{lemma}\label{lem:SAPA-smoothness}
    In simultaneous first-price or all-pay auctions with subadditive valuations, every bidder $i \in N$ is conservatively $(\frac12-\delta,1)$-semi-smooth w.r.t.~every $x_i \in \cX_i$ for the charge~\eqref{eq:item-charge}, for every $0<\delta<\frac12$.
\end{lemma}
\begin{proof}
Fix a bidder $i \in N$, a target $x_i\subseteq M$, a distribution $\vA \in \Delta(\cA)$ and $\epsilon > 0$, and let $O(a_i,\vec p) := \{j\in x_i\mid a_{ij}\geq p_j\}$ for $a_i \in \cA_i$ and $\vec p=(p_j)_{j\in x_i}$ with $p_j\geq 0$. By Remark~\ref{rem:wlog-positive-value}, we assume that $v_i(x_i) > 0$, so $x_i \neq \emptyset$. Let $\vec p(\vA_{-i})$ denote the distribution of the highest opposing bids $(\beta_j(\va_{-i}))_{j \in x_i}$ on the items of $x_i$, where $\va_{-i} \sim \vA_{-i}$, and let $\vec p_\epsilon(\vA_{-i})$ denote the same distribution with $\epsilon$ added to every coordinate. Consider the deviation $\hat A_i'$ that samples $\hat a_i'\sim \vec p_{\epsilon}(\vA_{-i})$, bids $\hat a_i'$ on the items of $x_i$ if $v_i(x_i)\geq \sum_{j\in x_i}\hat a_{ij}'$ and $0$ on them otherwise, and bids $0$ on every item outside $x_i$. It is conservative, since the payment is at most the sum of the bids, which is at most $v_i(x_i)$ by construction. 
For any winning bid vector $\vec p$ and any action $\tilde a_i$ we see that $O(\tilde a_i,\vec p)\cup O(\vec p,\tilde a_i)=x_i$ and using this gives: 
\begin{equation}\begin{aligned}\label{ineq:proof-subadditive} \E_{\substack{a_i'\sim \vec p(\vA_{-i})\\\va_{-i}\sim \vA_{-i}}}[v_i(O(a_i',\vec p(\va_{-i})))]
&=\frac12\E_{\substack{a_i'\sim \vec p(\vA_{-i})\\\va_{-i}\sim \vA_{-i}}}\left[v_i(O(a_i',\vec p(\va_{-i})))+v_i(O(\vec p(\va_{-i}),a_i'))\right]\\
&\geq \frac12 v_i(x_i)\end{aligned}\end{equation} 
The first equality follows because the realizations $a_i'\sim \vec p(\vA_{-i})$ and $\vec p(\va_{-i})$ where $\va_{-i}\sim \vA_{-i}$ have identical distributions. The last inequality follows by subadditivity of $v_i$. Now, fix a realization $(a_{ij}'+\epsilon)_{j\in x_i}$ from $\vec p_{\epsilon}(\vA_{-i})$. Suppose that $v_i(x_i)\geq \sum_{j\in x_i}(a_{ij}'+\epsilon)$. Then, since $O(\cdot,\vec p(\va_{-i}))$ is monotone and the payment is at most $\sum_{j\in x_i}(a_{ij}'+\epsilon)$, we see that
\begin{equation}\label{eq:utility-subadditive-LB}u_i((a_{ij}'+\epsilon)_{j\in x_i},\va_{-i})\geq v_i(O(a_i',\vec p(\va_{-i})))-\sum_{j\in x_i}(a_{ij}'+\epsilon).\end{equation}
Otherwise, $\sum_{j\in x_i}(a_{ij}'+\epsilon)>v_i(x_i)$, the right-hand side of~\eqref{eq:utility-subadditive-LB} is negative, and the same inequality with $u_i(0,\va_{-i})$ instead holds, as the deviation then bids $0$. Hence
$$\begin{aligned}\myexp{\hat a_i'\sim \hat A_i'}{\va_{-i}\sim \vA_{-i}}{u_i(\hat a_i',\va_{-i})}&\geq \myexpl{a_i'\sim \vec p(\vA_{-i})}{\va_{-i}\sim \vA_{-i}}{v_i(O(a_i',\vec p(\va_{-i})))-\sum_{j\in x_i}(a_{ij}'+\epsilon)}\\
&\geq \frac12 v_i(x_i)-\sum_{j\in x_i}\myexp{\va_{-i}\sim \vA_{-i}}{}{\beta_j(\va_{-i})}-|x_i|\epsilon.
\end{aligned}$$
Choosing $\epsilon := \frac{\delta v_i(x_i)}{|x_i|}$ gives $(\frac12-\delta,1)$-semi-smoothness.
\end{proof}

\begin{theorem}\label{thm:all-pay}
The PoA of CCE of simultaneous all-pay auctions with subadditive valuations is $2$ for budget-free instances and $3$ with budgets.
\end{theorem}
\begin{proof}
Let $\cI$ be the class of all instances of simultaneous all-pay auctions with subadditive valuations, and fix $I \in \cI$, an optimal profile $\vchi^\ast$ and $0 < \delta < \frac12$. By Lemma~\ref{lem:SAPA-smoothness} and Corollary~\ref{cor:smoothness-reduction}, every bidder $i \in N$ is budget-feasible $(\frac12 - \delta, 1)$-semi-smooth for the charge~\eqref{eq:item-charge}. For every profile $\va \in \cA$ and every allocation profile $\vx \in \cX$, the bundles $x_1, \dots, x_n$ are pairwise disjoint and $\beta_j(\va_{-i}) \leq \beta_j(\va)$ holds for every item $j \in M$ and every bidder $i \in N$, so
$$\sum_{i \in N} \charge_i(x_i, \va_{-i}) \leq \sum_{j \in M} \beta_j(\va) \leq \sum_{i \in N} \sum_{j \in M} a_{ij} = \sum_{i \in N} P_i(\va) ,$$
where the second inequality holds as every highest bid $\beta_j(\va)$ is the bid of some bidder on item $j$. Taking expectations over $\vx \sim \vchi^\ast$ and $\va \sim \vA$ for a CCE $\vA$ of $I$ gives the coupling constraint~\eqref{eq:coupling-constraint} with $\lambda_i = \frac12 - \delta$, $\mu_i = 1$, $\kappa_i = \frac{\mu_i}{\lambda_i}$ and $\nu_i = 0$ for every $i \in N$. Since all-pay auctions are PYB (see Definition~\ref{def:pyb}) and $\delta > 0$ is arbitrary, Theorem~\ref{thm:ultimate-POA} gives, by~\eqref{eq:normalized-pyb-poa}, $\CCEPOA(\cI) \leq \max\{2, 2 + \iota\}$, where $\iota = 0$ if $\cI$ is budget-free and $\iota = 1$ otherwise. Together with the lower bounds of Theorem~\ref{thm:lb-apa}, this proves the theorem.
\end{proof}

Both bounds are tight already for a single item, as we show in Appendix~\ref{sec:lb-apa} (Theorem~\ref{thm:lb-apa}). For simultaneous first-price auctions the following upper bounds follow since the smoothness parameters are the same as for the all-pay setting, and the coupling constraint~\eqref{eq:coupling-constraint} holds as in the proof of Theorem~\ref{thm:all-pay}, with $\sum_{j \in M} \beta_j(\va) = \sum_{i \in N} P_i(\va)$.

\begin{theorem}\label{thm:fpa-subadditive}
    The PoA of CCE of simultaneous first-price auctions with subadditive valuations is $2$ for budget-free instances and at most $3$ with budgets.
\end{theorem}

Tightness in the budget-free setting follows from the lower bound from \cite{christodoulou16}.

\subsection{Second-Price Auctions}
In contrast to the previous auctions we study in this section where we considered instances with subadditive valuations, as we show below, the smoothness argument for simultaneous second-price auctions extends  to \emph{all} monotone valuations.
\begin{lemma}\label{lem:SPA-smoothness}
    In simultaneous second-price auctions with monotone valuations, every bidder $i \in N$ is conservatively $(1-\epsilon,1)$-semi-smooth w.r.t.~every $x_i \in \cX_i$ for the charge~\eqref{eq:item-charge}, for every $0<\epsilon<1$.
\end{lemma}
\begin{proof}
        Fix a bidder $i \in N$, a bundle $x_i\subseteq M$, $0 < \epsilon < 1$ and a distribution $\vA \in \Delta(\cA)$. If $\myexp{\va_{-i}\sim \vA_{-i}}{}{\beta_j(\va_{-i})} = \infty$ holds for some $j \in x_i$, the right-hand side of the semi-smoothness inequality in Definition~\ref{def:smoothness} is $-\infty$, and the opt-out deviation $\emptyset_i$ satisfies both inequalities. Hence we assume that these expectations are finite. 
    
    For $j\in x_i$, let $a_j':=\frac{1}{\epsilon}|x_i|\myexp{\va_{-i}\sim \vA_{-i}}{}{\beta_j(\va_{-i})}$ if $\myexp{\va_{-i}\sim \vA_{-i}}{}{\beta_j(\va_{-i})}>0$, and $a_j':=1$ otherwise. In either case,
$\mathbb P_{\va_{-i}\sim \vA_{-i}}(\beta_j(\va_{-i})\geq a_j')\leq \frac{\epsilon}{|x_i|},$
by Markov's inequality in the former case, and since $\beta_j(\va_{-i})=0$ almost surely in the latter.
    Consider the deviation $a_i'$ in which bidder $i$ bids $a_j'$ on every item $j\in x_i$ and $0$ on every other item. By the union bound, the probability that bidder $i$ wins all items in $x_i$ is
    $$\mathbb P_{\va_{-i}\sim \vA_{-i}}(X_i(a_i',\va_{-i})\supseteq x_i)\geq 1-\mathbb P_{\va_{-i}\sim \vA_{-i}}(\exists j\in x_i :\beta_j(\va_{-i})\geq a_j')\geq 1-\epsilon.$$
    Now, let $\hat a_i'$ be the deviation that bids $a_i'$ whenever $v_i(x_i)\geq \myexp{\va_{-i}\sim \vA_{-i}}{}{P_i(a_i',\va_{-i})}$ and $\emptyset_i$ otherwise. The deviation $\hat a_i'$ is conservative by definition. We distinguish two cases. First, suppose that $\hat a_i'=\emptyset_i$. Under $a_i'$, bidder $i$ pays $\beta_j(\va_{-i})$ for every item $j \in x_i$ that she wins, while every item that she wins outside $x_i$ is won with the bid $0$ and hence has $\beta_j(\va_{-i})=0$. Therefore,
    $$\sum_{j\in x_i}\myexp{\va_{-i}\sim \vA_{-i}}{}{\beta_j(\va_{-i})}\geq \myexp{\va_{-i}\sim \vA_{-i}}{}{P_i(a_i',\va_{-i})}> v_i(x_i),$$
    so that $(1-\epsilon)v_i(x_i)-\sum_{j\in x_i}\myexp{\va_{-i}\sim \vA_{-i}}{}{\beta_j(\va_{-i})}<0\leq \myexp{\va_{-i}\sim \vA_{-i}}{}{\z_i(\emptyset_i,\va_{-i})}$, and the semi-smoothness inequality holds trivially. Otherwise, $\hat a_i'=a_i'$ and
    $$\begin{aligned}\myexp{\va_{-i}\sim \vA_{-i}}{}{\z_i(a_i',\va_{-i})}&=\myexpl{\va_{-i}\sim \vA_{-i}}{}{v_i(X_i(a_i',\va_{-i}))-\sum_{j\in X_i(a_i',\va_{-i})}\beta_j(\va_{-i})}\\&\geq (1-\epsilon)v_i(x_i)-\sum_{j\in x_i}\myexp{\va_{-i}\sim \vA_{-i}}{}{\beta_j(\va_{-i})},
    \end{aligned}$$
    where the inequality holds because $X_i(a_i',\va_{-i})\supseteq x_i$ with probability at least $1-\epsilon$, $v_i$ is monotone and non-negative, and $\beta_j(\va_{-i})=0$ for every item won outside $x_i$. Since $x_i$ was an arbitrary target bundle, bidder $i$ is conservatively $(1-\epsilon,1)$-semi-smooth w.r.t.~every $x_i \in \cX_i$.
\end{proof}

\begin{theorem}\label{thm:second-price}
The PoA of NOB-feasible CCE of simultaneous second-price auctions with monotone valuations is $2$, with budgets as well as for budget-free instances.
\end{theorem}
\begin{proof}
Let $\cI$ be the class of all instances of simultaneous second-price auctions with monotone valuations, and fix $I \in \cI$, an optimal profile $\vchi^\ast$ and $0 < \epsilon < 1$. By Lemma~\ref{lem:SPA-smoothness} and Corollary~\ref{cor:smoothness-reduction}, every bidder $i \in N$ is budget-feasible $(1 - \epsilon, 1)$-semi-smooth for the charge~\eqref{eq:item-charge}. For every profile $\va \in \cA$ and every allocation profile $\vx \in \cX$, the bundles $x_1, \dots, x_n$ are pairwise disjoint and $\beta_j(\va_{-i}) \leq \beta_j(\va)$ holds for every item $j \in M$ and every bidder $i \in N$, so
$$\sum_{i \in N} \charge_i(x_i, \va_{-i}) \leq \sum_{j \in M} \beta_j(\va) = \sum_{i \in N} W_i(\va) ,$$
where the equality holds as the willingness to pay of bidder $i$ is the sum of her winning bids, $W_i(\va) = \sum_{j \in X_i(\va)} a_{ij}$, by Definition~\ref{def:wtp}. Taking expectations over $\vx \sim \vchi^\ast$ and $\va \sim \vA$ for a CCE $\vA$ of $I$ gives the coupling constraint~\eqref{eq:coupling-constraint} with $\lambda_i = 1 - \epsilon$, $\mu_i = 1$, $\kappa_i = 0$ and $\nu_i = \frac{1}{1 - \epsilon}$ for every $i \in N$. Since $\epsilon > 0$ is arbitrary, Theorem~\ref{thm:ultimate-POA} gives, by~\eqref{eq:normalized-nob-poa}, $\CCEPOA(\cI) \leq \max\{1 + 1,\ 1 + \iota\} = 2$, where $\iota = 0$ if $\cI$ is budget-free and $\iota = 1$ otherwise. The lower bound follows from \cite{christodoulou16bayesian}.
\end{proof}

\section{Generalized Second-Price Auctions}\label{sec:gsp}\label{sec:position-auctions}

In this section, we consider GSP auctions, defined in Section~\ref{sec:auction-formats}. Our main result is the following theorem.

\begin{theorem}\label{thm:gsp}
The PoA of NOB-feasible CCE of GSP auctions is $\frac{1}{1 - \mathcal{W}_0(1)} \approx 2.310$ for budget-free instances and $1 + \frac{e}{e-1} \approx 2.582$ with budgets, where $\mathcal{W}_0(1) \approx 0.567$ is the solution of $w\,e^{w} = 1$.
\end{theorem}

Since the GSP auction is not PYB (see Definition~\ref{def:pyb}), the PoA in Theorem~\ref{thm:gsp} is taken over NOB-feasible CCE (see Definition~\ref{def:nob}). Both bounds are tight for CCE, as we show in Appendix~\ref{sec:gsp-lb} (Theorems~\ref{thm:lb-gsp-free} and~\ref{thm:lb-gsp-budgets}). As for the first-price formats, budgets increase the PoA.

For budget-free instances, the bound $2.310$ coincides with the bound of~\cite{caragiannis2015gsp} for coarse correlated equilibria of bidders who never bid above their value per click. Under our no-overbidding condition, which is in expectation (see Definition~\ref{def:nob}), Theorem~\ref{thm:lb-gsp-free} shows that this bound is tight, whereas the best known lower bound for bidders who never bid above their value per click is $1.259$, for pure Nash equilibria~\cite{caragiannis2015gsp}.

Recall from Section~\ref{sec:auction-formats} that the bidder who receives slot $j \in [m]$ bids $\beta_j(\va)$ and pays $\alpha_j \beta_{j+1}(\va)$, i.e., bidder $i \in N$ pays $P_i(\va) = \alpha_{X_i(\va)} \beta_{X_i(\va)}(\va_{-i})$ if $X_i(\va) \geq 1$. By Definition~\ref{def:wtp}, her willingness to pay is $W_i(\va) = \alpha_{X_i(\va)} a_i$, since, for a fixed slot, the opposing bids can raise the next bid in the ranking arbitrarily close to $a_i$, but not above it. Both quantities are zero if $X_i(\va) = 0$. Hence, for every profile $\va \in \cA$,
\begin{equation}\label{eq:gsp-payments}
\sum_{i \in N} W_i(\va) = \sum_{j=1}^{m} \alpha_j \beta_j(\va)
\qquad\text{and}\qquad
\sum_{i \in N} P_i(\va) = \sum_{j=1}^{m} \alpha_j \beta_{j+1}(\va) .
\end{equation}

For $\mu > 0$, let $\lambda := \mu (1 - e^{-1/\mu})$, so that $0 < \lambda \leq \min\set{1, \mu}$. For every bidder $i \in N$, slot $x_i \in \cX_i$ and profile $\va_{-i} \in \cA_{-i}$, we define the charge
\begin{equation}\label{eq:gsp-charge}
\charge_i(x_i, \va_{-i}) := \mathbbm{1}_{[x_i = 1]}\, \alpha_1 \beta_1(\va_{-i}) + \frac{\mu}{\lambda}\, \mathbbm{1}_{[x_i > 1]}\, \alpha_{x_i} \beta_{x_i}(\va_{-i}).
\end{equation}
Note that $\charge_i(0,\va_{-i})=0$. Using this charge, we are ready to state smoothness for bidders in a GSP auction.

\begin{lemma}\label{lem:gsp-smoothness}
In GSP auctions, every bidder $i \in N$ is conservatively $(\lambda, \lambda)$-semi-smooth w.r.t.~every $x_i \in \cX_i$ for the charge~\eqref{eq:gsp-charge}, for every $\mu > 0$ and $\lambda = \mu (1 - e^{-1/\mu})$.
\end{lemma}
\begin{proof}
Fix $\mu > 0$, a bidder $i \in N$, a slot $x_i \in \cX_i$ and a distribution $\vA \in \Delta(\cA)$. We distinguish three cases, based on the value of $x_i$.

\medskip
\noindent \textbf{Case 1:} $x_i = 0$. The opt-out action $\emptyset_i$ pays nothing and has non-negative utility, while $v_i(0) = 0 = \charge_i(0, \va_{-i})$ for every $\va_{-i} \in \cA_{-i}$, so both inequalities of Definition~\ref{def:smoothness} hold.

\medskip
\noindent \textbf{Case 2:} $x_i = 1$. Consider the deviation in which bidder $i$ bids her value per click $\bar v_i$, and fix $\va_{-i} \in \cA_{-i}$. If she receives a slot $j \geq 1$, she pays $\alpha_j \beta_j(\va_{-i}) \leq \alpha_j \bar v_i \leq \alpha_1 \bar v_i = v_i(1)$, as the next bid in the ranking is at most her own, and otherwise she pays nothing. Hence the deviation is conservative and $\z_i(\bar v_i, \va_{-i}) \geq 0$. If $j = 1$, then $\z_i(\bar v_i, \va_{-i}) = v_i(1) - \alpha_1 \beta_1(\va_{-i})$. Otherwise, some opposing bid is at least $\bar v_i$, so $\alpha_1 \beta_1(\va_{-i}) \geq v_i(1)$ and $\z_i(\bar v_i, \va_{-i}) \geq 0 \geq v_i(1) - \alpha_1 \beta_1(\va_{-i})$. In both cases, since $\z_i(\bar v_i, \va_{-i}) \geq 0$ and $\lambda \leq 1$,
\[
\z_i(\bar v_i, \va_{-i}) \geq \lambda \bigl( v_i(1) - \alpha_1 \beta_1(\va_{-i}) \bigr) = \lambda v_i(1) - \lambda\, \charge_i(1, \va_{-i}) ,
\]
and taking expectations over $\va_{-i} \sim \vA_{-i}$ gives the semi-smoothness inequality.

\medskip
\noindent \textbf{Case 3:} $x_i > 1$. Consider the deviation $\hat A_i$ of Lemma~\ref{lem:first-price-smoothness} for the same $\mu$, the GFP auction on the same instance, the target $x_i$ and the distribution $\vA$. For every bid $b \geq 0$ and every $\va_{-i} \in \cA_{-i}$, the GFP and the GSP auction rank the bidders in the same way and assign the same slot $j$ to bidder $i$, and her GSP payment $\alpha_j \beta_j(\va_{-i})$ is at most her GFP payment $\alpha_j b$. Hence the expected payment of $\hat A_i$ in the GSP auction is at most its expected payment in the GFP auction, which is at most $v_i(x_i)$, and its expected utility in the GSP auction is at least its expected utility in the GFP auction, which is at least $\lambda v_i(x_i)$ minus $\mu$ times the expected charge~\eqref{eq:first-price-charge} of the GFP auction. It remains to show that this charge is at most $\alpha_{x_i} \beta_{x_i}(\va_{-i})$ for every $\va_{-i} \in \cA_{-i}$, as then
\begin{align*}
\myexp{\hat a_i \sim \hat A_i}{\va_{-i} \sim \vA_{-i}}{\z_i(\hat a_i, \va_{-i})} &\geq \lambda v_i(x_i) - \mu \myexp{\va_{-i} \sim \vA_{-i}}{}{\alpha_{x_i} \beta_{x_i}(\va_{-i})} \\
&= \lambda v_i(x_i) - \lambda \myexp{\va_{-i} \sim \vA_{-i}}{}{\charge_i(x_i, \va_{-i})} .
\end{align*}
In the GFP auction, the first maximum in~\eqref{eq:first-price-charge} equals $\sum_{j=1}^{m} \alpha_j \beta_j(\va_{-i})$, as the other bidders then receive the slots in the order of their bids. The second maximum is at least the total bid of the assignment that gives slot $x_i$ to bidder $i$ and the remaining slots to the other bidders in the order of their bids, which is $\sum_{j < x_i} \alpha_j \beta_j(\va_{-i}) + \sum_{j \geq x_i} \alpha_{j+1} \beta_j(\va_{-i})$ with $\alpha_{m+1} := 0$. Hence the charge is at most
\[
\sum_{j = x_i}^{m} (\alpha_j - \alpha_{j+1})\, \beta_j(\va_{-i}) \leq \alpha_{x_i} \beta_{x_i}(\va_{-i}) ,
\]
where the inequality follows as $\beta_j(\va_{-i}) \leq \beta_{x_i}(\va_{-i})$ holds for $j \geq x_i$.

\medskip
In all three cases, bidder $i$ is conservatively $(\lambda, \lambda)$-semi-smooth w.r.t.~$x_i$ for the charge~\eqref{eq:gsp-charge}.
\end{proof}

We now verify the coupling constraint~\eqref{eq:coupling-constraint} for the charges defined in~\eqref{eq:gsp-charge}.

\begin{lemma}\label{lem:gsp-coupling}
In GSP auctions, for every $\mu > 0$ with $\lambda = \mu (1 - e^{-1/\mu})$, every profile $\va \in \cA$ and every allocation profile $\vx \in \cX$,
\[
\sum_{i \in N} \charge_i(x_i, \va_{-i}) \leq \sum_{i \in N} W_i(\va) + \frac{\mu - \lambda}{\lambda} \sum_{i \in N} P_i(\va) .
\]
\end{lemma}
\begin{proof}
Fix $\mu > 0$, $\va \in \cA$ and $\vx \in \cX$. Removing the bid of bidder $i$ does not increase the $j$th highest bid, so $\beta_j(\va_{-i}) \leq \beta_j(\va)$ holds for every $i \in N$ and $j \in [m]$. Since every slot is assigned to at most one bidder in $\vx$, this gives
\begin{align*}
\sum_{i \in N :\, x_i \geq 1} \alpha_{x_i} \beta_{x_i}(\va_{-i}) &\leq \sum_{j=1}^{m} \alpha_j \beta_j(\va) = \sum_{i \in N} W_i(\va) \qquad\text{and}\\
\sum_{i \in N :\, x_i > 1} \alpha_{x_i} \beta_{x_i}(\va_{-i}) &\leq \sum_{j=2}^{m} \alpha_j \beta_j(\va) \leq \sum_{j=1}^{m-1} \alpha_j \beta_{j+1}(\va) \leq \sum_{i \in N} P_i(\va) ,
\end{align*}
where the equality and the last inequality follow by~\eqref{eq:gsp-payments}, and the second inequality of the second line as $\alpha_{j+1} \leq \alpha_j$. By~\eqref{eq:gsp-charge}, the sum of the charges equals the left-hand side of the first line plus $\frac{\mu - \lambda}{\lambda}$ times the left-hand side of the second line, and $\frac{\mu - \lambda}{\lambda} \geq 0$, so the claim follows.
\end{proof}

\begin{proof}[Proof of Theorem~\ref{thm:gsp}]
Let $\cI$ be the class of all instances of GSP auctions, and fix $I \in \cI$, an optimal profile $\vchi^\ast$ and $\mu > 0$, and let $\lambda := \mu (1 - e^{-1/\mu})$. By Lemma~\ref{lem:gsp-smoothness} and Corollary~\ref{cor:smoothness-reduction}, every bidder $i \in N$ is budget-feasible $(\lambda, \lambda)$-semi-smooth for the charge~\eqref{eq:gsp-charge}. For every NOB-feasible CCE $\vA$ of $I$, Lemma~\ref{lem:gsp-coupling} applied to every $\vx$ in the support of $\vchi^\ast$ and every $\va$ in the support of $\vA$ gives, after taking expectations, the coupling constraint~\eqref{eq:coupling-constraint} with $\lambda_i = \mu_i = \lambda$, $\kappa_i = \frac{\mu - \lambda}{\lambda} = \frac{e^{-1/\mu}}{1 - e^{-1/\mu}}$ and $\nu_i = 1$ for every $i \in N$. Since GSP auctions are not PYB, Theorem~\ref{thm:ultimate-POA} gives, by~\eqref{eq:normalized-nob-poa},
\[
\CCEPOA(\cI) \leq \max\Bigl\{ \frac{1}{\lambda} + 1,\ \frac{\mu - \lambda}{\lambda} + 1 + \iota \Bigr\} = 1 + \frac{\max\set{1,\ \mu e^{-1/\mu} + \lambda \iota}}{\mu (1 - e^{-1/\mu})} \qquad \text{for every } \mu > 0,
\]
where $\iota = 0$ if $\cI$ is budget-free and $\iota = 1$ otherwise, and the equality uses $\mu - \lambda = \mu e^{-1/\mu}$. For budget-free instances, we choose $\mu$ with $\mu e^{-1/\mu} = 1$, i.e., $\frac{1}{\mu} e^{1/\mu} = 1$, so that $\frac{1}{\mu} = \mathcal{W}_0(1)$. Then $\lambda = \mu - \mu e^{-1/\mu} = \mu - 1$, and the bound becomes $1 + \frac{1}{\mu - 1} = \frac{1}{1 - \mathcal{W}_0(1)} \approx 2.310$. With budgets, $\mu e^{-1/\mu} + \lambda = \mu$, and we choose $\mu = 1$, for which the bound becomes $1 + \frac{1}{1 - e^{-1}} = 1 + \frac{e}{e - 1} \approx 2.582$. Together with the matching lower bounds of Theorems~\ref{thm:lb-gsp-free} and~\ref{thm:lb-gsp-budgets}, this proves the theorem.
\end{proof}

\section{Uniform Price Auction}\label{sec:upa}

In this section, we focus on the PoA of CCE in the uniform price auction under uniform bidding, defined in Section~\ref{sec:auction-formats}, with budgets. Our main result is the following theorem.

\begin{theorem}\label{thm:upa-individual-certificate}
The PoA of NOB-feasible CCE of the uniform price auction under uniform bidding is
\[
\frac{2 + \mathcal{W}_0(-e^{-2})}{1 + \mathcal{W}_0(-e^{-2})} \approx 2.188 ,
\]
with budgets as well as for budget-free instances, where $\mathcal{W}_0$ is the principal branch of the Lambert $W$ function (see Section~\ref{sec:preliminaries}).
\end{theorem}

Since the uniform price auction is not PYB (see Definition~\ref{def:pyb}), the PoA in Theorem~\ref{thm:upa-individual-certificate} is taken over NOB-feasible CCE (see Definition~\ref{def:nob}). The bound is tight already for pure Nash equilibria of budget-free instances under uniform bidding, as we show in Appendix~\ref{sec:upa-lb} (Theorem~\ref{thm:UPA-LB}). Hence, unlike the first-price formats, budgets do not increase the PoA of the uniform price auction.

The previous upper bound for the uniform price auction is from~\cite{dekeijzer13}: for budget-free bidders with submodular, i.e., concave, valuations, the PoA of CCE under no-overbidding is at most $-\mathcal{W}_{-1}(-e^{-2}) \approx 3.146$, and it is at least $2 - \frac{1}{k}$ for pure Nash equilibria. Both bounds hold under uniform bidding as well as under the standard bidding interface, in which every bidder submits $k$ non-increasing marginal bids, one for each unit, so that uniform bidding is the special case in which all positive marginal bids of a bidder are equal.\footnote{For pure Nash equilibria under the standard bidding interface, the PoA is $2.188$~\cite{birmpas2019tight}, and Theorem~\ref{thm:UPA-LB} adapts that construction to uniform bidding.} For CCE under uniform bidding, Theorem~\ref{thm:upa-individual-certificate} gives the upper bound $2.188$, also in the presence of budgets.

To prove Theorem~\ref{thm:upa-individual-certificate}, we first show in Section~\ref{sec:upa-two-point} that there is an optimal profile under which every bidder receives one of two consecutive numbers of units. In Section~\ref{sec:upa-deviation}, we define a randomized deviation for each bidder and show that, with this deviation, the bidder is conservatively smooth w.r.t.~her marginal of this profile (Lemma~\ref{lem:upa-key}). By Theorem~\ref{thm:smoothness-reduction}, this gives budget-feasible semi-smoothness (Definition~\ref{def:semi-smoothness}). In Section~\ref{sec:upa-lp}, we then verify the coupling constraint~\eqref{eq:coupling-constraint} for its charges, with coefficients that depend on the bidder. Finally, in Section~\ref{subsec:upa-proof}, we use Theorem~\ref{thm:ultimate-POA} to prove Theorem~\ref{thm:upa-individual-certificate}. The matching lower bound construction can be found in Appendix~\ref{sec:upa-lb}.

Before presenting the deviation, we recall the notation of Section~\ref{sec:auction-formats} and prove an auxiliary lemma, which we use throughout this section. For a profile $\va$, $\beta_0(\va)$ is the highest losing marginal bid, i.e., the uniform price, and for every bidder $i \in N$, $\beta_1(\va_{-i})\leq\dots\leq\beta_k(\va_{-i})$ are the $k$ highest opposing marginal bids, where we set $x_i\beta_{x_i}(\va_{-i}) := 0$ for $x_i = 0$. The payment of bidder $i$ is $P_i(\va)=\beta_0(\va)X_i(\va)$, and according to Definition~\ref{def:wtp}, her willingness to pay is
\[
W_i(\va)=W_i(a_i, X_i(\va)) = \sup_{\bar{\va}_{-i}:\,X_i(a_i, \bar{\va}_{-i})=X_{i}(\va)}\beta_0(a_i, \bar{\va}_{-i})X_i(\va)=b_iX_i(\va),
\]
where the last equality holds since, for a fixed positive allocation, the opposing bids can raise $\beta_0(a_i,\bar{\va}_{-i})$ arbitrarily close to $b_i$, but not above it. Both quantities are zero when $X_i(\va)=0$. Finally, the values $\beta_j(\va_{-i})$ determine how many units bidder $i$ wins. For every $\va_{-i}$ and $j=1,\dots,k-1$, an action $a_i=(b_i, q_i)$ of bidder $i$ with $q_i\geq j$ and $b_i \in (\beta_{j}(\va_{-i}), \beta_{j+1}(\va_{-i}))$ guarantees $X_i(a_i,\va_{-i}) = j$. Similarly, an action with $q_i=k$ and $b_i > \beta_{k}(\va_{-i})$ guarantees $X_i(a_i,\va_{-i}) = k$.

For $x_i \in [k]$, the term $x_i\beta_{x_i}(\va_{-i})$ is the payment of bidder $i$ when she submits $x_i$ copies of a per-unit bid slightly above $\beta_{x_i}(\va_{-i})$. This action wins $x_i$ units at the uniform price $\beta_{x_i}(\va_{-i})$.

\begin{lemma}\label{lem:upa-dev}
Let $i \in N$ and $x_i \in \cX_i$, and let $a_i^{x_i} := \bigl(\frac{v_i(x_i)}{x_i}, x_i\bigr)$ if $v_i(x_i) > 0$ and $a_i^{x_i} := \emptyset_i$ otherwise. Then, for every $\va_{-i} \in \cA_{-i}$,
\[
P_i(a_i^{x_i}, \va_{-i}) \leq v_i(x_i) \quad\text{and}\quad \z_i(a_i^{x_i}, \va_{-i}) \geq v_i(x_i) - x_i\beta_{x_i}(\va_{-i}).
\]
\end{lemma}
\begin{proof}
Fix $\va_{-i}$. If $v_i(x_i) = 0$, then $P_i(\emptyset_i, \va_{-i}) = 0$ and $\z_i(\emptyset_i, \va_{-i}) \geq 0 \geq -x_i\beta_{x_i}(\va_{-i})$, so both inequalities hold. Suppose that $v_i(x_i) > 0$, and thus $x_i \geq 1$. Bidder $i$ submits $x_i$ marginal bids, so $X_i(a_i^{x_i}, \va_{-i}) \leq x_i$, and $\beta_0(a_i^{x_i}, \va_{-i}) \leq \frac{v_i(x_i)}{x_i}$ holds whenever $X_i(a_i^{x_i}, \va_{-i}) \geq 1$, as one of her marginal bids then wins. Hence
\[
P_i(a_i^{x_i}, \va_{-i}) = \beta_0(a_i^{x_i}, \va_{-i})X_i(a_i^{x_i}, \va_{-i}) \le v_i(x_i),
\]
which is the first inequality. For the second inequality, we distinguish three cases, based on the value of $X_i(a_i^{x_i}, \va_{-i})$.

  \medskip
  \noindent \textbf{Case 1:} $X_i(a_i^{x_i}, \va_{-i}) = x_i$. In this case, the other $k - x_i$ winners are the $k - x_i$ highest opposing bids, i.e., the entries $\beta_{x_i+1}(\va_{-i})$ to $\beta_k(\va_{-i})$. The highest losing bid is then the next opposing bid, so $\beta_0(a_i^{x_i}, \va_{-i}) = \beta_{x_i}(\va_{-i})$, and hence $\z_i(a_i^{x_i}, \va_{-i}) = v_i(x_i) - x_i \beta_{x_i}(\va_{-i})$.

  \medskip
  \noindent \textbf{Case 2:} $1 \leq X_i(a_i^{x_i}, \va_{-i}) < x_i$. Here, the per-unit bid $\frac{v_i(x_i)}{x_i}$ is both a winning and a losing bid, so $\beta_0(a_i^{x_i}, \va_{-i}) = \frac{v_i(x_i)}{x_i}$. Hence
\[
\z_i(a_i^{x_i}, \va_{-i}) = v_i\bigl(X_i(a_i^{x_i}, \va_{-i})\bigr) - X_i(a_i^{x_i}, \va_{-i}) \cdot \frac{v_i(x_i)}{x_i} \ \ge\ 0 ,
\]
where the inequality follows as $v_i$ is concave with $v_i(0)=0$, so $\frac{v_i(y)}{y} \geq \frac{v_i(x_i)}{x_i}$ holds for every $y \in [x_i]$. On the other hand, more than $k - x_i$ opposing marginal bids win, so $\beta_{x_i}(\va_{-i})$ wins while one of her marginal bids loses. This gives $\beta_{x_i}(\va_{-i}) \geq \frac{v_i(x_i)}{x_i}$, and thus
\[
v_i(x_i) - x_i\beta_{x_i}(\va_{-i}) = x_i\Bigl(\frac{v_i(x_i)}{x_i} - \beta_{x_i}(\va_{-i})\Bigr) \leq 0 ,
\]
so the two bounds together give $\z_i(a_i^{x_i}, \va_{-i}) \geq 0 \geq v_i(x_i) - x_i\beta_{x_i}(\va_{-i})$.

    \medskip
  \noindent \textbf{Case 3:} $X_i(a_i^{x_i}, \va_{-i}) = 0$. Finally, $\z_i(a_i^{x_i}, \va_{-i}) = 0$ and all $k$ units go to opposing bids, so the per-unit bid $\frac{v_i(x_i)}{x_i}$ is at most every opposing winning bid. In particular, it is at most $\beta_1(\va_{-i}) \le \beta_{x_i}(\va_{-i})$, which gives $\z_i(a_i^{x_i}, \va_{-i}) = 0 \geq v_i(x_i) - x_i\beta_{x_i}(\va_{-i})$.

\medskip
In all three cases, $\z_i(a_i^{x_i}, \va_{-i}) \geq v_i(x_i) - x_i\beta_{x_i}(\va_{-i})$.
\end{proof}

\subsection{The Two-Point Optimal Profile}\label{sec:upa-two-point}

Given an optimal profile $\vchi^\ast$, we denote by $\bar{x}_i := \E_{x_i \sim \chi_i^\ast}[x_i]$ the expected number of units of bidder $i \in N$ under $\vchi^\ast$. Moreover, for every $i \in N$ and $r \in [0, k]$, we denote by $\chi_i^{r} \in \Delta(\cX_i)$ the probability distribution that assigns probability $r - \lfloor r \rfloor$ to $\lceil r \rceil$ and the remaining probability to $\lfloor r \rfloor$, so that $\E_{x_i \sim \chi_i^{r}}[x_i] = r$. In this section, we show that there is an optimal profile whose marginals are $\chi_1^{\bar{x}_1}, \dots, \chi_n^{\bar{x}_n}$. Essentially, since $v_i$ is concave, $\chi_i^{\bar{x}_i}$ maximizes the expected value of bidder $i$ among all distributions over $\cX_i$ with mean $\bar{x}_i$, and thus also her liquid value. Further, since $\sum_{i \in N} \bar{x}_i \le k$, these distributions can be combined into a distribution over feasible allocations. The following lemma states the two properties of $\chi_i^{r}$ that we use throughout the section.
\begin{lemma}\label{lem:upa-two-point}
Let $i \in N$ and $\chi_i \in \Delta(\cX_i)$ with $r := \E_{x_i \sim \chi_i}[x_i] > 0$. Then:
\begin{enumerate}
    \item $\E_{x_i \sim \chi_i}[v_i(x_i)] \le \E_{x_i \sim \chi_i^{r}}[v_i(x_i)]$.
    \item $\E_{x_i \sim \chi_i^{r}}[v_i(x_i)] \le \frac{r}{j}\, v_i(j)$ holds for every $j \in [k]$ with $j \le r$.
\end{enumerate}
\end{lemma}
\begin{proof}
Let $d := v_i(\lfloor r \rfloor + 1) - v_i(\lfloor r \rfloor)$ if $r < k$, and $d := v_i(k) - v_i(k-1)$ if $r = k$. Since $v_i$ is concave, the marginal values $v_i(x) - v_i(x-1)$ are non-increasing in $x \in [k]$. Hence the marginal values of the units above $\lfloor r \rfloor$ are at most $d$, and those of the units up to $\lfloor r \rfloor$ are at least $d$, which gives
\[
v_i(x) \le v_i(\lfloor r \rfloor) + d\,(x - \lfloor r \rfloor) \qquad \forall x \in \cX_i ,
\]
with equality for $x \in \set{\lfloor r \rfloor, \lceil r \rceil}$. Since $\chi_i$ and $\chi_i^{r}$ both have mean $r$, and $\chi_i^{r}$ is supported on $\set{\lfloor r \rfloor, \lceil r \rceil}$, taking expectations gives
\[
\E_{x_i \sim \chi_i}[v_i(x_i)] \le v_i(\lfloor r \rfloor) + d\,(r - \lfloor r \rfloor) = \E_{x_i \sim \chi_i^{r}}[v_i(x_i)] ,
\]
which is the first claim. For the second claim, let $j \in [k]$ with $j \le r$, and thus $j \le \lfloor r \rfloor$. Since $v_i$ is concave with $v_i(0) = 0$, the average value $\frac{v_i(x)}{x}$ is non-increasing in $x \in [k]$, so $v_i(x) \le \frac{x}{j}\, v_i(j)$ holds for every $x \in \set{\lfloor r \rfloor, \lceil r \rceil}$. Taking the expectation over $x_i \sim \chi_i^{r}$ gives $\E_{x_i \sim \chi_i^{r}}[v_i(x_i)] \le \frac{r}{j}\, v_i(j)$.
\end{proof}

Using the first claim of Lemma~\ref{lem:upa-two-point}, we now show that such an optimal profile exists.

\begin{proposition}\label{prop:upa-two-point-opt}
For every instance $I$ of the UPA, there is an optimal profile $\vchi^\ast$ such that $\chi_i^\ast = \chi_i^{\bar{x}_i}$ holds for every $i \in N$.%
\end{proposition}
\begin{proof}
Let $\vchi^\ast$ be an optimal profile, and let $\rho$ be drawn uniformly from $[0, 1)$. For every $i \in N$, let
\[
x_i := \Bigl\lceil \sum_{\ell \le i} \bar{x}_\ell - \rho \Bigr\rceil - \Bigl\lceil \sum_{\ell < i} \bar{x}_\ell - \rho \Bigr\rceil ,
\]
and let $\vchi'$ be the distribution of $\vx = (x_1, \dots, x_n)$. We first verify that $\vchi' \in \Delta(\cX)$ and that its marginals are $\chi_1^{\bar{x}_1}, \dots, \chi_n^{\bar{x}_n}$. The two arguments of the ceilings differ by $\bar{x}_i$, so $x_i$ is an integer in $(\bar{x}_i - 1, \bar{x}_i + 1)$, i.e., $x_i \in \set{\lfloor \bar{x}_i \rfloor, \lceil \bar{x}_i \rceil}$. For every real $s$, the random variable $\lceil s - \rho \rceil$ equals $\lceil s \rceil$ with probability $s - \lfloor s \rfloor$ and $\lfloor s \rfloor$ otherwise, so its expectation is $s$. Hence $\E[x_i] = \bar{x}_i$, and $x_i$ is distributed according to $\chi_i^{\bar{x}_i}$. Every allocation in the support of $\vchi^\ast$ assigns at most $k$ units, so $\sum_{i \in N} \bar{x}_i \le k$. Since the sum of the $x_i$ telescopes and $\lceil -\rho \rceil = 0$, we obtain
\[
\sum_{i \in N} x_i = \Bigl\lceil \sum_{i \in N} \bar{x}_i - \rho \Bigr\rceil \le \lceil k - \rho \rceil = k ,
\]
so $\vx \in \cX$ and $\vchi' \in \Delta(\cX)$.

It remains to show that $\vchi'$ is optimal. For every $i \in N$ with $\bar{x}_i > 0$, the first claim of Lemma~\ref{lem:upa-two-point} gives $\E_{x_i \sim \chi_i^\ast}[v_i(x_i)] \le \E_{x_i \sim \chi_i^{\bar{x}_i}}[v_i(x_i)]$, and both sides are $0$ if $\bar{x}_i = 0$, as $\chi_i^\ast$ and $\chi_i^{0}$ are then the point mass on $0$. Hence $\hat v_i(\chi_i^{\bar{x}_i}) \ge \hat v_i(\chi_i^\ast)$ holds for every $i \in N$, and thus $\LW(I, \vchi') \ge \LW(I, \vchi^\ast) = \OPT(I)$. Under $\vchi'$, the expected number of units of bidder $i \in N$ is $\bar{x}_i$ and her marginal is $\chi_i^{\bar{x}_i}$, so $\vchi'$ is an optimal profile with the required property.
\end{proof}

\subsection{A Randomized Deviation with Two Stages}\label{sec:upa-deviation}

Throughout this section, we fix a parameter $\alpha \in (0,1)$. For $r \in (0, k]$, let
\begin{equation}\label{eq:upa-lambda-mu}
\frac{1}{\gamma(r)} := \int_{\alpha k}^{k} \frac{dt}{\min\set{\lceil t \rceil, r}}
\qquad\text{and}\qquad
\lambda(r) := \frac{(1 - \alpha)\, k}{r}\, \gamma(r) ,
\end{equation}
and let $\lambda(0) := 1$ and $\gamma(0) := 0$.

Fix a bidder $i \in N$ and consider the following randomized unilateral deviation, which we call $A_i^{\bar{x}_i}$. If $\bar{x}_i = 0$, the bidder always opts out, and otherwise she proceeds in two stages. In the first stage, she draws $t$ from $(\alpha k, k]$ with probability density function
$
f(t) = \frac{\gamma(\bar{x}_i)}{\min\set{\lceil t \rceil, \bar{x}_i}}
$,
which integrates to one by~\eqref{eq:upa-lambda-mu}. In the second stage, she sets $j := \lceil t \rceil$ if $\lceil t \rceil \le \bar{x}_i$, and she draws $j \sim \chi_i^{\bar{x}_i}$ otherwise. Finally, she plays the action $a_i^{j}$ of Lemma~\ref{lem:upa-dev}, i.e., she submits $j$ copies of the per-unit bid $\frac{v_i(j)}{j}$ if $v_i(j) > 0$, and she opts out otherwise. A description of $A_i^{\bar{x}_i}$ is given in Algorithm~\ref{alg:upa-deviation}.

\begin{algorithm}[t]
\caption{The randomized deviation $A_i^{\bar{x}_i}$ of bidder $i \in N$}\label{alg:upa-deviation}
\DontPrintSemicolon
\KwIn{Valuation function $v_i$, $\alpha \in (0,1)$, and $\chi_i^\ast \in \Delta(\cX_i)$.}
Set $\bar{x}_i := \E_{x_i \sim \chi_i^\ast}[x_i]$.\;
If $\bar{x}_i = 0$, submit the opt-out action $\emptyset_i$ and \textbf{stop}.\;
Draw $t$ from $(\alpha k, k]$ with pdf $f(t) = \frac{\gamma(\bar{x}_i)}{\min\set{\lceil t \rceil, \bar{x}_i}}.$\;
\eIf{$\lceil t \rceil \le \bar{x}_i$}{Set $j := \lceil t \rceil$.\;}{Draw $j \sim \chi_i^{\bar{x}_i}$.\;}
\eIf{$v_i(j) > 0$}{Submit $a_i^{j} = \bigl(\frac{v_i(j)}{j}, j\bigr)$, i.e., $j$ copies of the per-unit bid $\frac{v_i(j)}{j}$.\;}{Submit the opt-out action $a_i^{j} = \emptyset_i$.\;}
\end{algorithm}
We define
\begin{equation}\label{eq:upa-charge}
\charge_i^{\bar{x}_i}(\va_{-i}) := \int_{\alpha k}^{k} \myexp{x_i \sim \chi_i^{\min\set{\lceil t \rceil, \bar{x}_i}}}{}{x_i\beta_{x_i}(\va_{-i})}\, f(t)\, dt,
\end{equation}
the expected value of $j\beta_j(\va_{-i})$ for the sampled quantity $j$ in the deviation $A_i^{\bar{x}_i}$. Indeed, for a fixed $t$, the quantity $j$ equals $\lceil t \rceil$ if $\lceil t \rceil \le \bar{x}_i$, and it is drawn from $\chi_i^{\bar{x}_i}$ otherwise. Since $\chi_i^{\lceil t \rceil}$ is the point mass on $\lceil t \rceil$, in both cases $j$ is distributed according to $\chi_i^{\min\set{\lceil t \rceil, \bar{x}_i}}$. We write $\E[\cdot \mid t]$ for the expectation over $j$ for a fixed $t$. In particular, $\E[j \mid t] = \min\set{\lceil t \rceil, \bar{x}_i}$, and thus $\E[j \mid t]\, f(t) = \gamma(\bar{x}_i)$ holds for every $t \in (\alpha k, k]$.

The charge of every bidder $i \in N$ is then defined by
\begin{equation}\label{eq:upa-charge-bidder}
\charge_i(x_i, \va_{-i}) :=
\begin{cases}
\charge_i^{\bar{x}_i}(\va_{-i}), & \bar{x}_i > 0, \\
0, & \bar{x}_i = 0,
\end{cases}
\end{equation}
for every $x_i \in \cX_i$ and $\va_{-i} \in \cA_{-i}$. It does not depend on $x_i$, so $\myexp{x_i \sim \chi_i^{\bar{x}_i}}{}{\charge_i(x_i, \va_{-i})} = \charge_i^{\bar{x}_i}(\va_{-i})$ holds whenever $\bar{x}_i > 0$. The following lemma bounds the expected payment and the expected utility of $A_i^{\bar{x}_i}$ in terms of the expected value $\E_{x_i \sim \chi_i^{\bar{x}_i}}[v_i(x_i)]$ and the charge~\eqref{eq:upa-charge-bidder}.

\begin{lemma}\label{lem:upa-key}
Let $i \in N$ with $\bar{x}_i > 0$. Then bidder~$i$ is conservatively $(\lambda(\bar{x}_i), 1)$-smooth w.r.t.~$\chi_i^{\bar{x}_i}$ for the charge~\eqref{eq:upa-charge-bidder}, with the deviation $A_i^{\bar{x}_i}$.
\end{lemma}
\begin{proof}
Fix $\va_{-i} \in \cA_{-i}$ and $t \in (\alpha k, k]$. We first relate the value of the sampled quantity $j$ to $\E_{x_i \sim \chi_i^{\bar{x}_i}}[v_i(x_i)]$.
If $\lceil t \rceil \le \bar{x}_i$, then $j = \lceil t \rceil$ is fixed, so $\E[j \mid t] = j$, and $j \le \lfloor \bar{x}_i \rfloor$ since $j$ is an integer. As $v_i$ is non-decreasing and $\chi_i^{\bar{x}_i}$ is supported on $\set{\lfloor \bar{x}_i \rfloor, \lceil \bar{x}_i \rceil}$, this gives $v_i(j) \le v_i(\lfloor \bar{x}_i \rfloor) \le \E_{x_i \sim \chi_i^{\bar{x}_i}}[v_i(x_i)]$. The second claim of Lemma~\ref{lem:upa-two-point}, applied to $\chi_i^{\bar{x}_i}$, also gives $v_i(j) \ge \frac{j}{\bar{x}_i}\, \E_{x_i \sim \chi_i^{\bar{x}_i}}[v_i(x_i)]$. If $\lceil t \rceil > \bar{x}_i$, then $j$ is drawn from $\chi_i^{\bar{x}_i}$, and thus $\E[j \mid t] = \bar{x}_i$ and $\E[v_i(j) \mid t] = \E_{x_i \sim \chi_i^{\bar{x}_i}}[v_i(x_i)]$. In both cases,
\begin{equation}\label{eq:upa-value-given-t}
\frac{\E[j \mid t]}{\bar{x}_i}\, \E_{x_i \sim \chi_i^{\bar{x}_i}}[v_i(x_i)] \le \E[v_i(j) \mid t] \le \E_{x_i \sim \chi_i^{\bar{x}_i}}[v_i(x_i)] .
\end{equation}
By the first inequality of Lemma~\ref{lem:upa-dev} and the second inequality in~\eqref{eq:upa-value-given-t}, we have $\E[P_i(a_i^{j}, \va_{-i}) \mid t] \le \E[v_i(j) \mid t] \le \E_{x_i \sim \chi_i^{\bar{x}_i}}[v_i(x_i)]$, and integrating over $t$ gives $\myexp{a_i' \sim A_i^{\bar{x}_i}}{}{P_i(a_i', \va_{-i})} \le \E_{x_i \sim \chi_i^{\bar{x}_i}}[v_i(x_i)]$. For the expected utility, we have
\begin{align*}
\myexp{a_i' \sim A_i^{\bar{x}_i}}{}{\z_i(a_i', \va_{-i})}
&= \int_{\alpha k}^{k} \E[\z_i(a_i^{j}, \va_{-i}) \mid t]\, f(t)\, dt \\
&\ge \int_{\alpha k}^{k} \Bigl( \E[v_i(j) \mid t] - \E[j\, \beta_j(\va_{-i}) \mid t] \Bigr) f(t)\, dt \\
&\ge \int_{\alpha k}^{k} \Bigl( \frac{\E[j \mid t]}{\bar{x}_i}\, \E_{x_i \sim \chi_i^{\bar{x}_i}}[v_i(x_i)] - \E[j\, \beta_j(\va_{-i}) \mid t] \Bigr) f(t)\, dt \\
&= \frac{\E_{x_i \sim \chi_i^{\bar{x}_i}}[v_i(x_i)]}{\bar{x}_i} \int_{\alpha k}^{k} \E[j \mid t]\, f(t)\, dt - \charge_i^{\bar{x}_i}(\va_{-i}) \\
&= \frac{(1 - \alpha)\, k\, \gamma(\bar{x}_i)}{\bar{x}_i}\, \E_{x_i \sim \chi_i^{\bar{x}_i}}[v_i(x_i)] - \charge_i^{\bar{x}_i}(\va_{-i}) \\
&= \lambda(\bar{x}_i)\, \E_{x_i \sim \chi_i^{\bar{x}_i}}[v_i(x_i)] - \charge_i^{\bar{x}_i}(\va_{-i}) .
\end{align*}
Here, the first equality follows by the definition of $A_i^{\bar{x}_i}$, and the first inequality follows by the second inequality of Lemma~\ref{lem:upa-dev}. The second inequality then follows by the first inequality in~\eqref{eq:upa-value-given-t}, and the second equality by~\eqref{eq:upa-charge}. Finally, the third equality follows since $\E[j \mid t]\, f(t) = \gamma(\bar{x}_i)$ holds for every $t \in (\alpha k, k]$, and the last equality follows by~\eqref{eq:upa-lambda-mu}. As $\va_{-i}$ was arbitrary, the first bound is the conservativeness requirement of Definition~\ref{def:smoothness}, and the second bound is the smoothness inequality with $\lambda = \lambda(\bar{x}_i)$ and $\mu = 1$, since the charge~\eqref{eq:upa-charge-bidder} does not depend on $x_i$.
\end{proof}

The main idea behind the deviation is the choice of the density $f$. Since $\E[j \mid t]\, f(t) = \gamma(\bar{x}_i)$ holds for every $t \in (\alpha k, k]$, the density cancels the expected number of units $\E[j \mid t]$. In the proof above, this turns the value term into an integral with the constant weight $\gamma(\bar{x}_i)$. In Section~\ref{sec:upa-lp}, the same cancellation bounds the charge by $\gamma(\bar{x}_i)$ times the sum of the $k$ highest opposing bids (see Lemma~\ref{lem:upa-coupling}).

\subsection{The Coupling Constraint of the UPA}\label{sec:upa-lp}

We now verify the coupling constraint~\eqref{eq:coupling-constraint} for the charges defined in~\eqref{eq:upa-charge-bidder}. To this end, we first bound the sum of the $k$ highest opposing bids.

\begin{lemma}\label{lem:upa-topk}
For every distribution $\vA \in \Delta(\cA)$ and every bidder $i \in N$,
\begin{equation}\label{eq:upa-topk}
\sum_{j=1}^{k} \myexp{\va_{-i} \sim \vA_{-i}}{}{\beta_j(\va_{-i})} \le \sum_{\ell \ne i} \myexp{\va \sim \vA}{}{W_\ell(\va)} + \myexp{\va \sim \vA}{}{P_i(\va)}.
\end{equation}
\end{lemma}
\begin{proof}
Fix a profile $\va \in \cA$. We first prove the inequality inside the expectations for $\va$. If $\sum_{\ell \in N} q_\ell < k$, every marginal bid wins, so $\beta_0(\va) = 0$ and $P_i(\va) = 0$. In this case, the $k$ highest opposing marginal bids consist of all opposing marginal bids, which all win. Hence $\sum_{j=1}^k \beta_j(\va_{-i}) = \sum_{\ell \ne i} b_\ell X_\ell(\va) = \sum_{\ell \ne i} W_\ell(\va)$, and the inequality holds with equality. Otherwise, all $k$ units are sold, so $k - X_i(\va)$ of the winning marginal bids belong to the bidders other than $i$, and their sum is $\sum_{\ell \ne i} b_\ell X_\ell(\va) = \sum_{\ell \ne i} W_\ell(\va)$. Each of these bids is among the $k$ highest marginal bids of all bidders, hence among the $k$ highest opposing marginal bids. The other $X_i(\va)$ of the $k$ highest opposing marginal bids are losing bids of $\va$ (or $0$), so each of them is at most $\beta_0(\va)$. Therefore,
\[
\sum_{j=1}^k \beta_j(\va_{-i}) \le \sum_{\ell \ne i} W_\ell(\va) + X_i(\va)\, \beta_0(\va) = \sum_{\ell \ne i} W_\ell(\va) + P_i(\va) ,
\]
and taking expectations over $\va \sim \vA$ gives~\eqref{eq:upa-topk}.
\end{proof}

We now bound the expected charge of the deviation $A_i^{\bar{x}_i}$. Recall that $\gamma(\cdot)$ is defined in~\eqref{eq:upa-lambda-mu} for the parameter $\alpha \in (0,1)$ fixed in Section~\ref{sec:upa-deviation}, and that $\E[j \mid t]\, f(t) = \gamma(\bar{x}_i)$ holds for every $t \in (\alpha k, k]$.

\begin{lemma}\label{lem:upa-coupling}
For every distribution $\vA \in \Delta(\cA)$ and every bidder $i \in N$ with $\bar{x}_i > 0$,
\begin{equation}\label{eq:upa-coupling-lemma}
\myexp{\va_{-i} \sim \vA_{-i}}{}{\charge_i^{\bar{x}_i}(\va_{-i})} \le \gamma(\bar{x}_i) \Bigl( \sum_{\ell \ne i} \myexp{\va \sim \vA}{}{W_\ell(\va)} + \myexp{\va \sim \vA}{}{P_i(\va)} \Bigr) .
\end{equation}
\end{lemma}
\begin{proof}
Fix $\va_{-i} \in \cA_{-i}$ and $t \in (\alpha k, k]$. Given $t$, the sampled quantity $j$ satisfies $j \le \lceil t \rceil$, since $j = \lceil t \rceil$ if $\lceil t \rceil \le \bar{x}_i$, and $j \le \lceil \bar{x}_i \rceil \le \lceil t \rceil$ otherwise. As $\beta_1(\va_{-i}) \le \dots \le \beta_k(\va_{-i})$, this gives $j\, \beta_j(\va_{-i}) \le j\, \beta_{\lceil t \rceil}(\va_{-i})$, and thus
\[
\E[j\, \beta_j(\va_{-i}) \mid t] \le \E[j \mid t]\, \beta_{\lceil t \rceil}(\va_{-i}) .
\]
Integrating over $t$, we obtain
\[
\charge_i^{\bar{x}_i}(\va_{-i}) = \int_{\alpha k}^{k} \E[j\, \beta_j(\va_{-i}) \mid t]\, f(t)\, dt \le \gamma(\bar{x}_i) \int_{\alpha k}^{k} \beta_{\lceil t \rceil}(\va_{-i})\, dt \le \gamma(\bar{x}_i) \sum_{m=1}^{k} \beta_m(\va_{-i}) .
\]
Here, the equality follows by~\eqref{eq:upa-charge}, and the first inequality follows since $\E[j \mid t]\, f(t) = \gamma(\bar{x}_i)$. The second inequality holds as the integrand is non-negative and $\int_{0}^{k} \beta_{\lceil t \rceil}(\va_{-i})\, dt = \sum_{m=1}^{k} \beta_m(\va_{-i})$. Finally, taking expectations over $\va_{-i} \sim \vA_{-i}$ and applying Lemma~\ref{lem:upa-topk} gives~\eqref{eq:upa-coupling-lemma}.
\end{proof}

\begin{lemma}\label{lem:upa-normalized-coupling}
Fix an instance $I$ of the UPA and an optimal profile $\vchi^\ast$, and let $\charge_i$ be the charges defined in~\eqref{eq:upa-charge-bidder}. Then, for every distribution $\vA \in \Delta(\cA)$, the coupling constraint~\eqref{eq:coupling-constraint} holds with $\lambda_i = \lambda(\bar{x}_i)$, $\mu_i = 1$,
\[
\kappa_i = \frac{\gamma(\bar{x}_i)}{\lambda(\bar{x}_i)} = \frac{\bar{x}_i}{(1 - \alpha)\, k}
\qquad \text{and} \qquad
\nu_i = \sum_{\ell \ne i} \frac{\gamma(\bar{x}_\ell)}{\lambda(\bar{x}_\ell)} = \sum_{\ell \ne i} \frac{\bar{x}_\ell}{(1 - \alpha)\, k} .
\]
\end{lemma}
\begin{proof}
Fix $\vA \in \Delta(\cA)$, and let $C_i$, $P_i$ and $W_i$ be the expected statistics of Section~\ref{sec:poa-bounds}. The identity $\frac{\gamma(r)}{\lambda(r)} = \frac{r}{(1 - \alpha)\, k}$ holds for every $r \in (0, k]$ by~\eqref{eq:upa-lambda-mu}, and also for $r = 0$ by the convention $\lambda(0) = 1$ and $\gamma(0) = 0$. This gives the two expressions for $\kappa_i$ and $\nu_i$ in the statement.

Now let $i \in N$. If $\bar{x}_i = 0$, then $C_i = 0$ by the definition of $\charge_i$, and $\kappa_i = 0$. Otherwise, $\charge_i(x_i, \va_{-i}) = \charge_i^{\bar{x}_i}(\va_{-i})$ does not depend on $x_i$, so $C_i = \myexp{\va_{-i} \sim \vA_{-i}}{}{\charge_i^{\bar{x}_i}(\va_{-i})}$, and Lemma~\ref{lem:upa-coupling} gives
\[
\frac{C_i}{\lambda(\bar{x}_i)} \le \frac{\gamma(\bar{x}_i)}{\lambda(\bar{x}_i)} \Bigl( \sum_{\ell \ne i} W_\ell + P_i \Bigr) = \kappa_i P_i + \kappa_i \sum_{\ell \ne i} W_\ell.
\]
In both cases, we thus have $\frac{C_i}{\lambda(\bar{x}_i)} \le \kappa_i P_i + \kappa_i \sum_{\ell \ne i} W_\ell$, and summing this inequality over $i \in N$ gives
\[
\sum_{i \in N} \frac{C_i}{\lambda(\bar{x}_i)} \le \sum_{i \in N} \kappa_i P_i + \sum_{\ell \in N} \Bigl( \sum_{i \ne \ell} \kappa_i \Bigr) W_\ell = \sum_{i \in N} \kappa_i P_i + \sum_{\ell \in N} \nu_\ell W_\ell ,
\]
where the equality exchanges the order of summation and uses the definition of $\nu_{\ell}$. This is~\eqref{eq:coupling-constraint} with $\mu_i = 1$.
\end{proof}

\subsection{The Proof of Theorem~\ref{thm:upa-individual-certificate}}\label{subsec:upa-proof}

We can now prove Theorem~\ref{thm:upa-individual-certificate}. 

\begin{proof}[Proof of Theorem~\ref{thm:upa-individual-certificate}]
Let $\cI$ be the class of all instances of the UPA, and fix $I \in \cI$, an optimal profile $\vchi^\ast$ as in Proposition~\ref{prop:upa-two-point-opt}, and $\alpha \in (0, 1)$. By Lemma~\ref{lem:upa-key}, every bidder $i \in N$ with $\bar{x}_i > 0$ is conservatively $(\lambda(\bar{x}_i), 1)$-smooth w.r.t.~$\chi_i^{\bar{x}_i}$, which equals $\chi_i^\ast$ by the choice of $\vchi^\ast$. If $\bar{x}_i = 0$, then $\chi_i^\ast$ is the point mass on $0$, and the same holds with $\lambda(0) = 1$ by Remark~\ref{rem:wlog-positive-value}. Since conservative smoothness implies conservative semi-smoothness for the same distribution, Theorem~\ref{thm:smoothness-reduction} shows that every bidder $i \in N$ is budget-feasible $(\lambda(\bar{x}_i), 1)$-semi-smooth. Together with Lemma~\ref{lem:upa-normalized-coupling}, the hypotheses of Theorem~\ref{thm:ultimate-POA} hold with $\lambda_i = \lambda(\bar{x}_i)$, $\mu_i = 1$, and the coefficients $\kappa_i$ and $\nu_i$ of this lemma. Since $\cI$ contains instances with budgets, we apply~\eqref{eq:normalized-nob-poa} with $\iota = 1$, which gives
\begin{equation}\label{eq:upa-coefficient}
\CCEPOA(\cI) \le \sup_{I \in \cI} \max_{i \in N} \Bigl( \nu_i + \max\set{ \frac{1}{\lambda(\bar{x}_i)},\, \kappa_i + 1 } \Bigr) .
\end{equation}
Each term in~\eqref{eq:upa-coefficient} is the maximum of $\nu_i + \frac{1}{\lambda(\bar{x}_i)}$ and $\nu_i + \kappa_i + 1$, two quantities that we bound separately.

To bound $\nu_i + \kappa_i + 1$, recall from Lemma~\ref{lem:upa-normalized-coupling} that $\kappa_i = \frac{\bar{x}_i}{(1 - \alpha)\, k}$ and $\nu_i = \sum_{\ell \ne i} \frac{\bar{x}_\ell}{(1 - \alpha)\, k}$. Every allocation in the support of $\vchi^\ast$ assigns at most $k$ units, so $\sum_{i \in N} \bar{x}_i \le k$, and hence, for every $i \in N$,
\[
\kappa_i + \nu_i = \sum_{\ell \in N} \frac{\bar{x}_\ell}{(1 - \alpha)\, k} \le \frac{1}{1 - \alpha}
\qquad \text{and} \qquad
\nu_i = \sum_{\ell \ne i} \frac{\bar{x}_\ell}{(1 - \alpha)\, k} \le \frac{k - \bar{x}_i}{(1 - \alpha)\, k}.
\]
In particular, $\nu_i + \kappa_i + 1 \le \frac{2 - \alpha}{1 - \alpha}$.

It remains to bound $\nu_i + \frac{1}{\lambda(\bar{x}_i)}$, for which we distinguish two cases, based on the value of $\bar{x}_i$.

\medskip
\noindent \textbf{Case 1:} $\bar{x}_i \le \alpha k$. We first show that $\lambda(\bar{x}_i) = 1$, which holds by definition if $\bar{x}_i = 0$. Otherwise, every $t > \alpha k$ satisfies $\lceil t \rceil \ge t > \alpha k \ge \bar{x}_i$, so $\min\set{\lceil t \rceil, \bar{x}_i} = \bar{x}_i$, and~\eqref{eq:upa-lambda-mu} gives $\frac{1}{\gamma(\bar{x}_i)} = \frac{(1 - \alpha)\, k}{\bar{x}_i}$ and thus $\lambda(\bar{x}_i) = 1$. Since $\nu_i \le \frac{k - \bar{x}_i}{(1 - \alpha)\, k} \le \frac{1}{1 - \alpha}$, we obtain
\[
\nu_i + \frac{1}{\lambda(\bar{x}_i)} \le \frac{1}{1 - \alpha} + 1 = \frac{2 - \alpha}{1 - \alpha} .
\]

\medskip
\noindent \textbf{Case 2:} $\bar{x}_i > \alpha k$. We first bound $\frac{1}{\lambda(\bar{x}_i)}$ as follows:
\[
\frac{(1 - \alpha)\, k}{\lambda(\bar{x}_i)} = \frac{\bar{x}_i}{\gamma(\bar{x}_i)} = \int_{\alpha k}^{k} \frac{\bar{x}_i}{\min\set{\lceil t \rceil, \bar{x}_i}}\, dt \le \int_{\alpha k}^{k} \frac{\bar{x}_i}{\min\set{t, \bar{x}_i}}\, dt .
\]
Here, both equalities follow by~\eqref{eq:upa-lambda-mu}, and the inequality follows as $\lceil t \rceil \ge t$. To evaluate the last integral, we split it at $t = \bar{x}_i$, as the integrand equals $\frac{\bar{x}_i}{t}$ on $(\alpha k, \bar{x}_i]$ and $1$ on $(\bar{x}_i, k]$. Hence
\[
\frac{(1 - \alpha)\, k}{\lambda(\bar{x}_i)} \le \int_{\alpha k}^{\bar{x}_i} \frac{\bar{x}_i}{t}\, dt + \int_{\bar{x}_i}^{k} dt = \bar{x}_i \ln\frac{\bar{x}_i}{\alpha k} + k - \bar{x}_i .
\]
Dividing this bound by $(1 - \alpha)\, k$ and adding the bound $\nu_i \le \frac{k - \bar{x}_i}{(1 - \alpha)\, k}$ shown above gives
\[
\nu_i + \frac{1}{\lambda(\bar{x}_i)} \le \frac{\bar{x}_i \ln\frac{\bar{x}_i}{\alpha k} + 2\, (k - \bar{x}_i)}{(1 - \alpha)\, k} .
\]
To bound the numerator, let $h(t) := t \ln\frac{t}{\alpha k} + 2\, (k - t)$ for $t \in [\alpha k, k]$, so that the numerator equals $h(\bar{x}_i)$. As $h''(t) = \frac{1}{t} > 0$, the function $h$ is convex and attains its maximum on $[\alpha k, k]$ at an endpoint, and thus
\[
h(\bar{x}_i) \le \max\set{h(\alpha k), h(k)} = \max\set{2\, (1 - \alpha)\, k,\, k \ln\frac{1}{\alpha}} .
\]
Combining the last two bounds, we obtain
\[
\nu_i + \frac{1}{\lambda(\bar{x}_i)} \le \frac{\max\set{\ln\frac{1}{\alpha},\, 2\, (1 - \alpha)}}{1 - \alpha} .
\]

\medskip
In both cases, $\nu_i + \frac{1}{\lambda(\bar{x}_i)}$ is at most $\frac{\max\set{\ln\frac{1}{\alpha},\, 2 - \alpha}}{1 - \alpha}$, since $2\, (1 - \alpha) < 2 - \alpha$, and so is $\nu_i + \kappa_i + 1$, as shown above. Since $I \in \cI$ and $i \in N$ were arbitrary, every term in~\eqref{eq:upa-coefficient} is at most this value, and we obtain
\[
\CCEPOA(\cI) \le \frac{\max\set{\ln\frac{1}{\alpha},\, 2 - \alpha}}{1 - \alpha} \qquad \text{for every } \alpha \in (0, 1) .
\]
We now choose $\alpha = \alpha_0 := -\mathcal{W}_0(-e^{-2}) \approx 0.159$. By the definition of $\mathcal{W}_0$, we have $\alpha_0 e^{-\alpha_0} = e^{-2}$, i.e., $\ln\frac{1}{\alpha_0} = 2 - \alpha_0$, so both terms in the maximum equal $2 - \alpha_0$, and the bound becomes
\[
\CCEPOA(\cI) \le \frac{2 - \alpha_0}{1 - \alpha_0} = \frac{2 + \mathcal{W}_0(-e^{-2})}{1 + \mathcal{W}_0(-e^{-2})} \approx 2.188 .
\]
Together with the matching lower bound of Theorem~\ref{thm:UPA-LB}, this proves the theorem.
\end{proof}

\section*{Acknowledgements and AI Disclosure} 
The authors used large language models (Claude, ChatGPT) to assist with editing and improving the exposition.
All mathematical results, proofs, and technical content were developed by the authors, with the following exceptions. The initial lower-bound constructions for the single-item first-price auction with budgets (Theorem~\ref{thm:lb-fpa}), GSP auctions (Theorems~\ref{thm:lb-gsp-free} and~\ref{thm:lb-gsp-budgets}), and all-pay auctions (Theorem~\ref{thm:lb-apa}) were generated by ChatGPT 5.6 Sol (high). For the GSP auction in particular, this required extensive discussion. Furthermore, ChatGPT 6 Astra helped with the upper bound for the uniform price auction (Theorem~\ref{thm:upa-individual-certificate}), where it used our directions and previous starting point to come up with the proof. These constructions and this proof were subsequently reviewed, verified, and substantially rewritten by the authors, who used the AI-generated arguments only as an initial starting point. The authors take full responsibility for the correctness and content of the manuscript.

\bibliographystyle{plainnat}
\bibliography{arxiv_bib}

\clearpage
\appendix

\section{Detailed Comparison with other LP-based Approaches}\label{app:duality}
Duality has been used to bound the price of anarchy of games through several approaches. \citet{nadav10} take the equilibrium distribution as the primal variables of a linear program whose objective value equals the PoA and whose dual yields $(\lambda,\mu)$-smoothness inequalities; their main result characterizes the class of distributions to which every smoothness bound applies, a superclass of CCE. A related line of work~\cite{kulkarni14, kim18} instead models the underlying optimization problem as a convex program or an exponential-size configuration LP and resorts to (Fenchel or LP) duality, so that the smoothness inequalities arise as dual constraints. A third approach, due to \citet{bilo18} (see also~\cite{bilo23}), focuses on the PoA of weighted congestion games: the primal LP variables are the parameters defining the players' payoffs, and complementary slackness guides the construction of matching worst-case instances. All of these works consider games with fixed strategy spaces, and none accommodates budget constraints.

Thematically, the closest related work is the work of \citet{kim25}, which also develops a duality-based approach for PoA bounds with respect to liquid welfare, combining a configuration LP with the KKT conditions of an inner convex program. This work establishes a tight PoA bound of $2$ for the proportional (Kelly) mechanism at pure Nash equilibria, a tight bound of $2$ at Bayesian CCE for a \emph{modified} proportional mechanism restricted to diminishing-returns valuations, and a tight bound of $2$ at (Bayesian) CCE for the auto-bidding pacing metagame \cite{feng24}, all under budget and return-on-spend constraints with hybrid value/utility objectives. The scope of \citet{kim25} is complementary to ours: the KKT-based analysis exploits the differentiable, concave structure of the proportional and pacing mechanisms and does not directly extend to the combinatorial auction formats we consider (discrete allocations and XOS or subadditive valuations). In contrast, we bound the inefficiency of CCE directly for simultaneous, multi-unit, and position auctions.

Our formulation is different from all the approaches above in what the LP variables represent. Rather than the equilibrium distribution, the underlying optimization problem, or the payoff parameters, our linear program ranges over aggregate per-bidder \emph{equilibrium statistics}: payments, willingness-to-pay, and charges, with the bidders partitioned according to whether their valuation at equilibrium already meets their budget. A single coupling constraint relates the penalty term to the total payments or to the total willingness-to-pay, so that our LP captures first-price, all-pay, and second-price payment rules in a unified way. The key technical ingredient is a reduction showing that smoothness can be achieved by deviations that remain budget-feasible against every distribution of the opponents' actions; this is precisely what makes smoothness arguments applicable in our setting, where a deviation must respect the budget constraint in expectation. This is an issue that does not arise in prior duality-based frameworks, where every deviation is always feasible. Finally, in contrast to the analysis of \citet{baldeschi2026}, which is tailored to first-price auctions and is not LP-based, our framework recovers the first-price guarantees for budget-constrained utility maximizers and extends them, through a single LP, to second-price and all-pay item auctions, multi-unit auctions, and position auctions.

\section{Additional Preliminaries}\label{app:additionalPrelims}

\subsection{Other Equilibrium Notions}\label{app:equilibrium-notions}
The budget-feasible deviation sets are those of Section~\ref{sec:auction-model}. We retain the following conventions for correlated, mixed Nash, and pure Nash equilibria.
\begin{definition}\label{def:other-equilibria}
Fix an instance $I$ and a budget-feasible distribution $\vA\in\Delta(\cA)$.
\begin{itemize}
\item $\vA$ is a \emph{correlated equilibrium (\ce)} if for every bidder $i \in N$ and every swap function $\swap: \supp(\A_i) \to \Delta(\cA_i)$ satisfying $
\iexp{\va \sim \vA}{a'_i \sim \swap(a_i)}{P_i(a'_i, \va_{-i})}
\le \budget_i$, 
it holds that
\begin{equation}\label{eq:ce-def}
\EX_{\va \sim \vA}\left[\z_i(\va)\right]
\;\ge\;
\iexp{\va \sim \vA}{a'_i \sim \swap(a_i)}{\z_i(a'_i, \va_{-i})}.
\end{equation}

\item $\vec{\A}$ is a \emph{mixed Nash equilibrium (\mne)} if $\vec{\A} = \prod_{i \in N} \A_i$ and for every bidder $i \in N$:%
\begin{equation}
\label{eq:mne-def}
\myexp{\va \sim \vA}{}{\z_i(\va)}
\ge 
\myexp{a'_i \sim A'_i}{\vec{a}_{-i} \sim \vA_{-i}}{\z_i(a'_i, \vec{a}_{-i})} \qquad \forall \A'_i \in \cC_i(\vec{\A}_{-i}).
\end{equation}

\item $\vA$ is a \emph{pure Nash equilibrium (\pne)} if $\vA = \prod_{i \in N} \A_i$, $\supp(\A_i) = \set{a_i}$ for all $i \in N$, and for every bidder $i \in N$:
\begin{equation}
\label{eq:pne-def}
\z_i(\va) \ge \z_i(a'_i, \va_{-i}) \qquad \forall A'_i \in \cC_i(\va_{-i}),\ \supp(A'_i) = \set{a'_i}.
\end{equation}

\end{itemize}
\end{definition}
For CE, the swap's budget constraint is imposed in expectation over the recommendations and the swap's randomization, rather than separately for each recommendation.
These define $\CE(I)$, $\MNE(I)$, and $\PNE(I)$, respectively, with $\MNE(I)\subseteq\CE(I)\subseteq\CCE(I)$. PNE tests only pure deviations, whereas MNE tests all budget-feasible mixed deviations. Under expected budgets, a feasible mixture can use an action that would exceed the budget if played alone. Thus the usual inclusion of PNE in MNE need not hold under this convention.

\begin{example}\label{ex:pne-not-mne}
Consider a first-price auction with two bidders and one item. Let $v_1=10$, $v_2=1$, $\budget_1=1$ and $\budget_2=1$. Suppose that ties are broken in favor of bidder 2 and consider the bid profile $\va=(1,1)$. There is no deterministic budget-feasible deviation bidder 1 can make to win the item. Hence, she has no incentive to deviate. The utility of bidder 2 under $\va$ is 0 and by bidding lower she would lose the item, while bidding higher would only increase her payment. Hence, for bidder 2 there is also no profitable deviation, showing that $\va$ is a PNE.

Now, consider the mixed deviation where bidder 1 bids $2$ with probability $0.5$ and $0$ otherwise, and call it $A_1'$. Her expected payment under $(A_1',a_2)$ is 1. Hence, this deviation is budget-feasible, and her expected utility improves from $0$ to
\[
 \myexp{a_1'\sim A_1'}{}{\z_1(a_1',a_2)}=\frac12(10-2)=4,
\]
proving that $\va$ is neither an MNE nor a CCE.
\end{example}

\subsection{The Liquid Welfare Benchmark}\label{app:liquid-welfare-benchmark}
Regarding our benchmark, since deterministic allocation profiles are degenerate random allocation profiles, our benchmark becomes stronger by taking the optimum over random allocation profiles, as opposed to optimizing over deterministic allocation profiles. Since $\cX$ is finite, $\Delta(\cX)$ is compact and liquid welfare is continuous in the distribution, so the optimum is attained. Furthermore, the liquid welfare w.r.t. the optimal deterministic allocation profile can be arbitrarily small compared to the liquid welfare at a CCE, as the following example shows.

\begin{example}\label{ex:randomized-benchmark}
If the optimum is taken over deterministic allocations instead, the ratio of the optimum over liquid welfare at a CCE may go below $1$. Consider a single-item first-price auction with $n+1$ bidders, for $n\geq2$, each having a value of $v_i=n$ and a budget of $\budget_i=1$. Ties are broken in favor of the bidder with the smallest index. So, bidder 1 will always win the tie, bidder 2 will win any tie except for ties with bidder 1, and so on. The optimal liquid welfare that can be achieved by any deterministic allocation is 1. Consider $\vA$ to be the uniform distribution over the bid profiles
\[
 (n,0,\ldots,0,n),\quad (0,n,0,\ldots,0,n),\quad\ldots,\quad (0,\ldots,0,n,n).
\]
Each profile has a probability of $\frac1n$ of occurring and thus each bidder $i\in[n]$ wins with probability $\frac1n$. Thus, her expected payment is exactly 1, meaning that she satisfies the budget constraint in expectation. Under this distribution of bid profiles, bidder $i=n+1$ has no chance of winning and pays zero. Note that the expected utility of each bidder is 0 under this distribution. Also, no deviation by a bidder $i\in N$ can yield a strictly positive utility because bidder $n+1$ always bids $n$. Bidder $n+1$ also always faces an opposing bid of $n$. Thus, every bidder must pay at least her value to win after any deviation, and neither a pure nor a mixed deviation can yield positive expected utility. Hence, $\vA$ is a CCE. The expected value of bidder $i\in [n]$ is, however, exactly 1. Hence, we see that
\[
 \LW(I,\vA)=\sum_{i\in [n]}\LW_i(I,\vA)=\sum_{i\in [n]}1=n>1.
\]
This example shows that the ratio of the optimal liquid welfare achieved by a deterministic allocation and the liquid welfare achieved by a CCE can be arbitrarily close to $0$.
\end{example}

\begin{remark}\label{rem:ex-ante-lw}
Given an instance $I$ and a random allocation $\vchi \in \Delta(\cX)$, the \emph{ex-post} liquid welfare is defined as $\LW_{\textsc{ex-post}}(I, \vchi) := \sum_{i \in N} \myexp{\vec{x} \sim \vchi}{}{\min\setl{v_i(x_i), \budget_i}}$. By concavity of the minimum, $\LW(I, \vchi) \geq \LW_{\textsc{ex-post}}(I, \vchi)$ holds and this is satisfied with equality for deterministic allocations. In this work, similarly to~\cite{deng2024, baldeschi2026, caragiannis16}, we choose to measure the performance of mechanisms with the ex-ante benchmark, consistently with the in-expectation budget constraints in \eqref{eq:agents-problem}. Together with non-negative expected utility, these constraints bound each bidder's expected payment by her ex-ante liquid value, as in Lemma~\ref{lem:IR}.

For simultaneous first-price auctions with additive valuations, an $\Omega(n)$ price of anarchy at MNE with respect to $\LW_{\textsc{ex-post}}$ is shown in~\cite[Theorem~C.2]{azar2015liquidpriceanarchy}. Their result uses hard budgets, under which a realized payment above the budget gives utility $-\infty$. It therefore does not by itself establish the same lower bound under our expected-payment constraints. See also the discussion of ex-ante and ex-post liquid welfare in~\cite{baldeschi2026}.
\end{remark}

\subsection{No-Overbidding as an Equilibrium Refinement}\label{app:sec:NOB-discussion}
\begin{remark}\label{rem:nob-comparison}
In the literature, there are two approaches to impose no-overbidding: NOB as a \emph{strategy refinement} restricts the action space of the bidders so that overbidding actions are unavailable (see, e.g., \cite{ST13, christodoulou16, bhawalkar11}), whereas NOB as an \emph{equilibrium refinement} does not restrict the full action space and instead restricts attention to equilibria whose distributions satisfy a no-overbidding assumption (e.g., \cite{birmpas2019tight, fu12}).
Here, we adopt NOB as an equilibrium refinement: under NOB as a strategy refinement, a profile may \emph{artificially} be an equilibrium only because a profitable (overbidding) deviation has been ruled out by assumption, whereas under an equilibrium refinement every budget-feasible mixed deviation remains available and must be unprofitable. We further stress that our condition is imposed in expectation, for each bidder separately: expected willingness to pay must not exceed either expected value or the budget. It need not hold for every realized action profile.

For a PYB mechanism, Lemma~\ref{lem:IR} and $W_i=P_i$ imply expected NOB at every equilibrium. For other mechanisms, some form of no-overbidding is necessary for a bounded PoA in general; see Example~\ref{ex:unboundedPOA}.
\end{remark}

\begin{example}\label{ex:unboundedPOA}
Consider a single-item second-price auction with two bidders: bidder~$1$ has $v_1 = 1$ and $\budget_1 = B$, and bidder~$2$ has $v_2 = B$ and $\budget_2 = B$, for $B>1$. The profile $\va=(B+1,0)$ is a pure Nash equilibrium: bidder~$1$ wins and pays~$0$, attaining her maximum possible utility, while bidder~$2$ would have to pay $B+1$, more than her value, to win. In fact, the bid profile is also an MNE and a CCE: no mixed deviation can improve either bidder's expected utility, even if its payments exceed the budget in some realizations. The resulting liquid welfare is $\min\{1, B\} = 1$, whereas the optimum is $\min\{B, B\} = B$, giving an unbounded price of anarchy as $B$ increases. Note that bidder~$1$ violates the NOB assumption, since $W_1(\va) = B + 1 > 1 = \LW_1(I,\va)$.
\end{example}

\section{Lower Bounds}\label{app:lower-bounds}

In this appendix,  we show the lower bounds for the auction formats of Sections~\ref{sec:first-price}, \ref{sec:gsp} and~\ref{sec:upa}, and for the all-pay auctions of Section~\ref{sec:all-pay}.

\subsection{First-Price Settings}\label{sec:first-price-lb}

We prove lower bounds for simultaneous first-price auctions under uniform-multiplier bidding and for first-price auctions with a single item and budgets. In the instances considered here, ties are broken by fixed priority orders over the bidders. Moreover, first-price auctions are PYB, so their equilibria are NOB-feasible by Lemma~\ref{lem:IR}. Our constructions with budgets for the first-price formats use the constant $\alpha_0 := -\mathcal{W}_0(-e^{-2}) \approx 0.159$. By the definition of $\mathcal{W}_0$, we have $\alpha_0 e^{-\alpha_0} = e^{-2}$, i.e.,
\begin{equation}\label{eq:alpha-zero-defining}
\ln\alpha_0 = \alpha_0 - 2 .
\end{equation}
The construction for first-price auctions with a single item uses a random variable $t$ with $\mathbb P(t \le x) = \frac{a}{1 - x}$ for $x \in [0, 1 - a]$, where $a \in (0, 1)$ is a parameter, and it sets $a = \alpha_0$. This distribution function increases from $a$ at $x = 0$ to $1$ at $x = 1 - a$, so $t$ has an atom of $a$ at zero and no other atom. In particular, $\mathbb P(t < b) = \mathbb P(t \le b) = \frac{a}{1 - b}$ holds for $b \in (0, 1 - a]$. A bidder with value $1$ who bids $b \in (0, 1 - a]$ and wins at the price $b$ whenever $t < b$, or whenever $t \le b$, therefore has expected utility $\frac{a}{1 - b}(1 - b) = a$, whatever her bid. In the two constructions with budgets below, the limit of the ratio of the optimal liquid welfare to that of the equilibrium is $\frac{1 + \E[t]}{a + \E[t]}$. The following lemma computes $\E[t]$ and explains the choice $a = \alpha_0$.

\begin{lemma}\label{lem:lb-price}
Let $a \in (0, 1)$. Then $\E[t] = 1 - a + a\ln a$, which equals $(1 - \alpha_0)^2 - \alpha_0$ for $a = \alpha_0$, and
\[
\frac{1 + \E[t]}{a + \E[t]} = \frac{2 - a + a\ln a}{1 + a\ln a}
\]
is maximized over $a \in (0, 1)$ at $a = \alpha_0$, where it equals $\frac{2 - \alpha_0}{1 - \alpha_0} = \frac{2 + \mathcal{W}_0(-e^{-2})}{1 + \mathcal{W}_0(-e^{-2})}$.
\end{lemma}
\begin{proof}
We have
\[
\E[t] = \int_0^{1 - a} \Bigl( 1 - \frac{a}{1 - x} \Bigr) dx = 1 - a + a\ln a ,
\]
and for $a = \alpha_0$, \eqref{eq:alpha-zero-defining} gives $\E[t] = 1 - \alpha_0 + \alpha_0(\alpha_0 - 2) = (1 - \alpha_0)^2 - \alpha_0$. Let $g(a) := \frac{2 - a + a\ln a}{1 + a\ln a}$ for $a \in (0, 1)$. Its denominator is positive, as $a\ln a \ge -\frac{1}{e}$, and its derivative $g'(a)$ has the sign of
\[
\ln a\,(1 + a\ln a) - (2 - a + a\ln a)(1 + \ln a) = a - 2 - \ln a .
\]
This expression is decreasing in $a$ and vanishes at $a = \alpha_0$ by~\eqref{eq:alpha-zero-defining}, so $g$ is maximized at $a = \alpha_0$. There, it equals $\frac{1 + (1 - \alpha_0)^2 - \alpha_0}{(1 - \alpha_0)^2} = \frac{(1 - \alpha_0)(2 - \alpha_0)}{(1 - \alpha_0)^2} = \frac{2 - \alpha_0}{1 - \alpha_0}$, and $\frac{2 - \alpha_0}{1 - \alpha_0} = \frac{2 + \mathcal{W}_0(-e^{-2})}{1 + \mathcal{W}_0(-e^{-2})}$ by the definition of $\alpha_0$.
\end{proof}

\subsubsection{Simultaneous First-Price Auctions with Uniform Bids}\label{sec:lb-uniform-bids}

Recall that under uniform-multiplier bidding, bidder $i$ chooses a bid multiplier $\alpha_i \in [0,1]$ and bids $\alpha_i v_{ij}$ on every item $j$, where $v_{ij}$ is her value for item $j$. The following theorem shows that both bounds $\frac{e}{e-1}$ and $2.188$ of Section~\ref{sec:first-price} are tight for this interface.
\begin{theorem}\label{thm:lb-uniform-bids}
The PoA of mixed Nash equilibria of simultaneous first-price auctions with additive valuations under uniform-multiplier bidding is at least $\frac{e}{e-1} \approx 1.582$ for budget-free instances and at least $\frac{2 + \mathcal{W}_0(-e^{-2})}{1 + \mathcal{W}_0(-e^{-2})} \approx 2.188$ with budgets, already for pure profiles.
\end{theorem}
\begin{proof}
We first prove the budget-free bound. Fix an integer $m \ge 2$, and let $\ell := \lceil \frac{m}{e} \rceil$, so that $1 \le \ell < m$. Consider a simultaneous first-price auction with uniform bids, $m$ items and three bidders with additive valuations. For every item $j \le m - \ell$, bidders~$1$ and~$2$ have value $v_{1j} = v_{2j} = \frac{j}{\ell + j}$, and they have value zero for the other $\ell$ items. Bidder~$3$ has value $v_{3j} = 1$ for every item $j$. Ties favor bidder~$3$ over bidders~$1$ and~$2$, and bidder~$1$ over bidder~$2$. All budgets are infinite.

We consider the profile $\va$ in which bidders~$1$ and~$2$ use the multiplier $1$ and bidder~$3$ uses the multiplier $0$. Bidder~$1$ then wins the items $j \le m - \ell$ and pays her values for them, while bidder~$2$ wins nothing. Bidder~$3$ wins the other $\ell$ items at the price zero, so the utilities of the three bidders are $0$, $0$ and $\ell$.

We now show that no pure deviation is profitable, so that $\va$ is a mixed Nash equilibrium. Bidder~$2$ never wins an item that she values, as she bids at most bidder~$1$'s bid on it and loses ties to her, so her utility is zero for every multiplier. Bidder~$1$ wins the items $j \le m - \ell$ with the multiplier $1$ and none of them with a smaller multiplier, as bidder~$2$ bids her values, so her utility is zero for every multiplier as well. Suppose now that bidder~$3$ uses the multiplier $\alpha_3 \in [0,1]$. Besides the $\ell$ items that only she values, she wins the items $j \le s$ for the largest $s \in \set{0, \dots, m - \ell}$ such that $\frac{s}{\ell+s} \le \alpha_3$ holds, since $\frac{j}{\ell + j}$ increases in $j$ and ties favor her. Her utility is then
\[
(\ell + s)(1 - \alpha_3) \le (\ell + s)\Bigl(1 - \frac{s}{\ell + s}\Bigr) = \ell ,
\]
which is her utility in $\va$.

It remains to evaluate the welfare. We have $\LW(I, \va) = \ell + \sum_{j=1}^{m-\ell} \frac{j}{\ell+j}$, and giving all items to bidder~$3$ shows that $\OPT(I) \ge m$. Since $\frac{j}{\ell+j} = 1 - \frac{\ell}{\ell+j}$, we obtain
\[
\frac{\OPT(I)}{\LW(I, \va)} \ge \frac{m}{m - \ell \sum_{j=1}^{m-\ell} \frac{1}{\ell+j}} \ge \frac{m}{m - \ell \ln\frac{m+1}{\ell+1}} ,
\]
where the second inequality follows from $\sum_{j=1}^{m-\ell} \frac{1}{\ell+j} \ge \int_1^{m-\ell+1} \frac{dx}{\ell + x} = \ln\frac{m+1}{\ell+1}$, as the integrand is decreasing. By the choice of $\ell$, we have $\lim_{m \to \infty} \frac{\ell}{m} = \frac{1}{e}$ and $\lim_{m \to \infty} \ln\frac{m+1}{\ell+1} = 1$, and hence
\[
\lim_{m \to \infty} \frac{m}{m - \ell \ln\frac{m+1}{\ell+1}} = \frac{1}{1 - \frac{1}{e}} = \frac{e}{e-1} .
\]

We now prove the bound with budgets. Fix an integer $m \ge 7$, and let $\ell := \lfloor \alpha_0 m \rfloor$, so that $1 \le \ell < m$. We modify the instance of the first part, with this choice of $\ell$, by giving bidder~$1$ the budget $\budget_1 := \sum_{j=1}^{m-\ell} \frac{j}{\ell + j}$ and the value $v_{1j} = (\ell + 1)\budget_1 \frac{j}{\ell + j}$ for every item $j \le m - \ell$. In the profile $\va$, bidder~$1$ uses the multiplier $\frac{1}{(\ell + 1)\budget_1}$, which is at most $1$, as $\budget_1 \ge \frac{1}{\ell + 1}$. Bidders~$2$ and~$3$ use the multipliers $1$ and $0$, as before. All bids are then as in the first part, so bidder~$1$ wins the items $j \le m - \ell$ and pays $\budget_1$, her budget, while bidder~$3$ wins the other $\ell$ items at the price zero.

The argument of the first part applies to bidders~$2$ and~$3$, whose opponents bid as before. A smaller multiplier makes bidder~$1$ lose every item $j \le m - \ell$ to bidder~$2$, and a larger one wins the same items at higher prices. Since her utility in $\va$ is $(\ell + 1)\budget_1^2 - \budget_1 \ge 0$, no pure deviation is profitable for her either, whether or not it meets her budget. By averaging, no randomized deviation is profitable for her, and $\va$ is a mixed Nash equilibrium.

Bidder~$1$ has value $(\ell + 1)\budget_1^2 \ge \budget_1$, so $\LW(I, \va) = \budget_1 + \ell$. Giving item~$1$ to bidder~$1$, who values it at $\budget_1$, and all other items to bidder~$3$ shows that $\OPT(I) \ge \budget_1 + m - 1$. Comparing the sum with integrals, as in the first part, gives $\ln\frac{m+1}{\ell+1} \le \sum_{j=1}^{m-\ell} \frac{1}{\ell+j} \le \ln\frac{m}{\ell}$. Since $\budget_1 = (m - \ell) - \ell\sum_{j=1}^{m-\ell} \frac{1}{\ell+j}$ and $\lim_{m \to \infty} \frac{\ell}{m} = \alpha_0$, this gives $\lim_{m \to \infty} \frac{\budget_1}{m} = 1 - \alpha_0 + \alpha_0\ln\alpha_0$. Dividing by $m$ and applying Lemma~\ref{lem:lb-price}, we obtain
\[
\lim_{m \to \infty} \frac{\budget_1 + m - 1}{\budget_1 + \ell} = \frac{2 - \alpha_0 + \alpha_0\ln\alpha_0}{1 + \alpha_0\ln\alpha_0} = \frac{2 - \alpha_0}{1 - \alpha_0} .
\]
\end{proof}

\subsubsection{First-Price Auctions with a Single Item}\label{sec:lb-fpa}

A first-price auction with a single item is a simultaneous first-price auction with $m = 1$, a discriminatory auction with $k = 1$ and a GFP auction with one slot. For budget-free instances, the bound $\frac{e}{e-1}$ is tight for CCE already in this case, by an example of Syrgkanis~\cite{syrgkanis-thesis}. The following theorem shows that the bound $2.188$ with budgets is tight in this case as well.

\begin{theorem}\label{thm:lb-fpa}
With budgets, the PoA of CCE of first-price auctions with a single item is at least $\frac{2 + \mathcal{W}_0(-e^{-2})}{1 + \mathcal{W}_0(-e^{-2})} \approx 2.188$.
\end{theorem}
\begin{proof}
Let $H > 1$, and let $t$ be the random variable defined above with $a = \alpha_0$. Consider a first-price auction with a single item and two bidders with values $H$ and $1$ and budgets $\budget_1 := \E[t]$ and $\budget_2 := \infty$, respectively, where ties at positive bids favor bidder~$1$ and ties at zero favor bidder~$2$. In the distribution $\vA$, both bidders submit the same bid $t$. Bidder~$1$ thus wins whenever $t > 0$ and pays $t$, so her expected payment is $\E[t] = \budget_1$, her expected value is $H(1 - \alpha_0)$, and her expected utility is $H(1 - \alpha_0) - \E[t]$. Bidder~$2$ wins at the price zero whenever $t = 0$, so her expected utility is $\alpha_0$. In particular, $\vA$ is budget-feasible.

We now show that $\vA$ is a CCE, starting with bidder~$2$. The bid zero gives her utility $\alpha_0$. A bid $b \in (0, 1 - \alpha_0]$ wins whenever $t < b$ and gives her utility $\frac{\alpha_0}{1 - b}(1 - b) = \alpha_0$, as noted above, and a bid $b > 1 - \alpha_0$ wins surely and gives $1 - b < \alpha_0$. Hence no bid of bidder~$2$ is profitable, and by averaging, no randomized deviation is profitable either.

A bid of bidder~$1$ can be profitable, but then it exceeds her budget. To make this precise, we claim that every bid $b \ge 0$ of bidder~$1$ satisfies
\begin{equation}\label{eq:lb-fpa-multiplier}
\myexp{\va_{-1} \sim \vA_{-1}}{}{\z_1(b, \va_{-1})} \le H\alpha_0 + (H - 1)\, \myexp{\va_{-1} \sim \vA_{-1}}{}{P_1(b, \va_{-1})} .
\end{equation}
The bid zero loses, so the left-hand side is zero. A bid $b \in (0, 1 - \alpha_0]$ wins whenever $t \le b$, i.e., with probability $\frac{\alpha_0}{1 - b}$, so her expected payment is $\frac{\alpha_0 b}{1 - b}$ and her expected utility is $\frac{\alpha_0 (H - b)}{1 - b}$, and~\eqref{eq:lb-fpa-multiplier} holds with equality. A bid $b > 1 - \alpha_0$ wins surely, so her utility is $H - b$ and her payment is $b$, and the right-hand side of~\eqref{eq:lb-fpa-multiplier} exceeds the left-hand side by $H(\alpha_0 + b - 1) > 0$. Now let $A_1' \in \cC_1(\vA_{-1})$. Its expected payment is at most $\budget_1 = \E[t]$, so averaging~\eqref{eq:lb-fpa-multiplier} over $b \sim A_1'$ gives
\[
\myexp{b \sim A_1'}{\va_{-1} \sim \vA_{-1}}{\z_1(b, \va_{-1})} \le H\alpha_0 + (H - 1)\, \E[t] = H(1 - \alpha_0) - \E[t] - H\alpha_0 (1 - \alpha_0) ,
\]
where the equality uses $\E[t] = (1 - \alpha_0)^2 - \alpha_0$ from Lemma~\ref{lem:lb-price}. Hence no deviation in $\cC_1(\vA_{-1})$ gives bidder~$1$ more than her utility in $\vA$, and $\vA$ is a CCE.

It remains to compare the liquid welfare of $\vA$ with the optimum. Since $\E[t] < 1 - \alpha_0 < H(1 - \alpha_0)$, we have $\LW(I, \vA) = \E[t] + \alpha_0$. Giving the item to bidder~$1$ with probability $\frac{\E[t]}{H}$ and to bidder~$2$ otherwise yields bidder~$1$ the expected value $\E[t]$, her budget, so $\OPT(I) \ge \E[t] + 1 - \frac{\E[t]}{H}$. The ratio is therefore at least $\frac{1 + \E[t] - \E[t]/H}{\alpha_0 + \E[t]}$, and
\[
\lim_{H \to \infty} \frac{1 + \E[t] - \E[t]/H}{\alpha_0 + \E[t]} = \frac{1 + \E[t]}{\alpha_0 + \E[t]} = \frac{2 - \alpha_0}{1 - \alpha_0} = \frac{2 + \mathcal{W}_0(-e^{-2})}{1 + \mathcal{W}_0(-e^{-2})} ,
\]
where the second equality follows by Lemma~\ref{lem:lb-price}.
\end{proof}

\subsection{Generalized Second-Price Auctions}\label{sec:gsp-lb}

The two constructions below match the upper bounds $2.310$ and $2.582$ of Section~\ref{sec:gsp}. Both constructions use the random variable $t$ of Appendix~\ref{sec:first-price-lb} with $a = e^{-s}$, and the following properties of it.

\begin{lemma}\label{lem:lb-gsp-price}
Let $s > 0$, and let $t$ be the random variable of Appendix~\ref{sec:first-price-lb} with $a = e^{-s}$, i.e., $\mathbb P(t \le x) = \frac{e^{-s}}{1 - x}$ for $x \in [0, 1 - e^{-s}]$. Then the following hold.
\begin{enumerate}[label=(\roman*)]
\item\label{item:lb-gsp-mean} $\E[\max\set{c, t}] = 1 - e^{-s}\bigl(1 + s + \ln(1 - c)\bigr)$ holds for every $c \in [0, 1 - e^{-s}]$.
\item\label{item:lb-gsp-utility} For all $0 \le c < b \le 1 - e^{-s}$,
\[
\E\bigl[ (1 - \max\set{c, t})\, \mathbbm{1}_{[t \le b]} \bigr] = e^{-s} \Bigl( 1 + \ln\frac{1 - c}{1 - b} \Bigr) .
\]
\item\label{item:lb-gsp-move} Moving probability mass of $t$ from zero to a positive value does not increase the left-hand side of~\ref{item:lb-gsp-utility}.
\end{enumerate}
\end{lemma}
In the constructions, a bidder with value $1$ per click who bids $b \in (0, 1 - e^{-s}]$ gets at most one click, only if $t \le b$, and she then pays at least $\max\set{c, t}$ per click for a suitable $c < b$. The left-hand side of~\ref{item:lb-gsp-utility} is therefore an upper bound on her expected utility.

\begin{proof}
For~\ref{item:lb-gsp-mean}, fix $c \in [0, 1 - e^{-s}]$. Since $\max\set{c, t} = c + \int_c^{1 - e^{-s}} \mathbbm{1}_{[t > x]}\, dx$, we have
\[
\E[\max\set{c, t}] = c + \int_c^{1 - e^{-s}} \Bigl( 1 - \frac{e^{-s}}{1 - x} \Bigr) dx = 1 - e^{-s} - e^{-s} \ln\frac{1 - c}{e^{-s}} ,
\]
which is the claimed expression, as $\ln\frac{1 - c}{e^{-s}} = s + \ln(1 - c)$.

For~\ref{item:lb-gsp-utility}, the distribution of $t$ has the atom $e^{-s}$ at zero and the density $\frac{e^{-s}}{(1 - x)^2}$ on $(0, 1 - e^{-s})$. For $c < b$, splitting the expectation at $t = c$ gives
\[
\E\bigl[ (1 - \max\set{c, t})\, \mathbbm{1}_{[t \le b]} \bigr] = (1 - c)\frac{e^{-s}}{1 - c} + \int_c^b \frac{e^{-s}}{1 - x}\, dx = e^{-s} \Bigl( 1 + \ln\frac{1 - c}{1 - b} \Bigr) ,
\]
where the first term uses $\mathbb P(t \le c) = \frac{e^{-s}}{1 - c}$, and the integrand is $1 - x$ times the density.

For~\ref{item:lb-gsp-move}, let $g(x) := (1 - \max\set{c, x})\, \mathbbm{1}_{[x \le b]}$ for $x \ge 0$, so that the left-hand side of~\ref{item:lb-gsp-utility} equals $\E[g(t)]$. The function $g$ is constant on $[0, c]$, decreasing on $(c, b]$ and zero on $(b, \infty)$, where $g(b) = 1 - b \ge 0$. Hence $g$ is non-increasing, and moving probability mass of $t$ upwards does not increase $\E[g(t)]$.
\end{proof}

\begin{theorem}\label{thm:lb-gsp-free}
Under NOB, the PoA of CCE of budget-free GSP auctions is at least $\frac{1}{1 - \mathcal{W}_0(1)} \approx 2.310$, where $\mathcal{W}_0(1) \approx 0.567$ is the solution of $w\, e^{w} = 1$.
\end{theorem}
\begin{proof}
Fix $s \in (0, 1)$ and integers $k \ge 4$ and $r \ge 2$, and let $\delta := \frac{s}{k}$ and $c_j := 1 - e^{-j\delta}$ for $j = 0, \dots, k$. Then $c_0 = 0$ and $c_k = 1 - e^{-s}$, and $\ln\frac{1 - c_j}{1 - c_{j+1}} = \delta$ holds for every $j < k$. Let further $q := \frac{1 + \delta - s}{1 - c_1}$, which satisfies $0 < q < 1$, as $s < 1$ and $c_1 < \delta < s - \delta$.

Draw $t$ as in Lemma~\ref{lem:lb-gsp-price}, and then move the probability mass $qe^{-s}$ from zero to $c_1$, which is possible, as $q < 1$. Since $c_1 \le c_j$ holds for $j \ge 1$, this move does not change $\max\set{c_j, t}$ for $j \ge 1$. Lemma~\ref{lem:lb-gsp-price}\ref{item:lb-gsp-mean} therefore gives
\[
\E[\max\set{c_j, t}] = 1 - e^{-s}(1 + s - j\delta) \qquad \text{for } j = 1, \dots, k .
\]
The move raises $\mathbb P(t > 0)$ by $qe^{-s}$ and $\E[t]$ by $qe^{-s}c_1$, so
\[
\mathbb P(t > 0) - \E[t] = \bigl( 1 - e^{-s} + qe^{-s} \bigr) - \bigl( 1 - e^{-s}(1 + s) + qe^{-s}c_1 \bigr) = e^{-s}(1 + \delta) ,
\]
where the last equality uses $q(1 - c_1) = 1 + \delta - s$. Without the move, this difference would only be $se^{-s}$, which is too small for the regular bidders below.

Consider a GSP auction with one slot of click-through rate $r$ and $r^2 + k$ slots of click-through rate $1$. There are $r^2 + 1$ regular bidders and $k - 1$ grid bidders, all with value $1$ per click, a special bidder with value $V := 1 + re^{-s}(1 + \delta)$ per click, and $r^2$ dummy bidders with value zero. All budgets are infinite.

In the distribution $\vA$, one regular bidder, called \emph{distinguished}, is drawn uniformly at random, independently of $t$. The distinguished bidder bids $V$, the special bidder bids $1 - e^{-s}$, grid bidder $j \in [k - 1]$ bids $\max\set{c_j, t}$, and all other bidders bid $t$. We refer to $c_j$ as the level of grid bidder $j$. Ties at positive bids favor the special bidder, then the grid bidders in decreasing order of index, then the regular bidders and then the dummy bidders. Ties at zero favor the dummy bidders, then the regular bidders, then the grid bidders and then the special bidder.

We first determine the outcome. The distinguished bidder gets the slot with rate $r$ and pays $1 - e^{-s}$ per click, the bid of the special bidder. The special bidder and grid bidder $j$ get slots with rate $1$ and pay $\max\set{c_{k-1}, t}$ and $\max\set{c_{j-1}, t}$ per click, respectively. If $t > 0$, the remaining $r^2$ slots go to the other regular bidders, who pay $t$ per click, and if $t = 0$, they go to the dummy bidders.

We then verify that $\vA$ is NOB-feasible and compute the expected utilities. Since there are more bidders than slots, opponents can bid arbitrarily close below a bidder without changing her slot. The willingness to pay of a bidder in a slot is therefore its click-through rate times her bid. A regular bidder is distinguished with probability $p := \frac{1}{r^2 + 1}$, and then her value $r$ falls short of her willingness to pay $rV$ by $r^2e^{-s}(1 + \delta)$. If she is not distinguished, she gets a slot if and only if $t > 0$ and pays $t$ per click. Her expected value then exceeds her expected willingness to pay by $\mathbb P(t > 0) - \E[t] = e^{-s}(1 + \delta)$. Altogether, her expected value minus her expected willingness to pay is
\[
-p\, r^2e^{-s}(1 + \delta) + (1 - p)\, e^{-s}(1 + \delta) = (1 - p - pr^2)\, e^{-s}(1 + \delta) = 0 ,
\]
so~\eqref{eq:NOB} holds for her with equality. Her expected utility is $pre^{-s} + (1 - p)e^{-s}(1 + \delta)$, which is at least $e^{-s}(1 + \delta)$, as $r \ge 1 + \delta$.

Grid bidder $j$ always gets a slot with rate $1$. Her willingness to pay $\E[\max\set{c_j, t}]$ is less than her value $1$, and her expected utility is $1 - \E[\max\set{c_{j-1}, t}]$. For $j \ge 2$, this equals $e^{-s}(1 + s - (j - 1)\delta)$, which is at least $e^{-s}(1 + 2\delta)$, as $(j + 1)\delta \le k\delta = s$. For $j = 1$, it equals $1 - \E[t] = e^{-s}(1 + s - qc_1)$, which is at least $e^{-s}(1 + 2\delta)$, as $qc_1 < \delta \le s - 2\delta$. The special bidder has willingness to pay $1 - e^{-s} < V$, and as $\E[\max\set{c_{k-1}, t}] = 1 - e^{-s}(1 + \delta)$, her expected utility is $V - 1 + e^{-s}(1 + \delta)$. The dummy bidders win slots only at the bid zero, so their willingness to pay and their utility are zero. Hence $\vA$ is NOB-feasible.

We now show that $\vA$ is a CCE, i.e., that no bidder has a profitable deviation to a bid $b$. Every other bidder bids at least $t$, and there are $2r^2 + k$ of them, at least as many as slots. If the deviator gets a slot, some other bidder is therefore ranked below her, and the next bid is at least $t$. We distinguish three cases, based on the value of $b$.

\medskip
\noindent \textbf{Case 1:} $b \in (0, 1 - e^{-s}]$. Let $c$ be the largest level below $b$ among $c_0$ and the levels of the grid bidders other than the deviator. If $t > b$, the deviator gets no slot, as all other bidders bid more than her. If $t \le b$, she gets at most one click, since the slot with rate $r$ goes to the distinguished bidder, or to the special bidder if the deviator is distinguished. The next bid is then at least $t$, and if $c > t$, it is at least $c$, as the grid bidder with level $c$ bids $c < b$ and is ranked below the deviator. Hence the utility of a deviator with value $1$ per click is at most $(1 - \max\set{c, t})\, \mathbbm{1}_{[t \le b]}$, and by Lemma~\ref{lem:lb-gsp-price}\ref{item:lb-gsp-utility} and~\ref{item:lb-gsp-move}, her expected utility is at most $e^{-s}(1 + \ln\frac{1 - c}{1 - b})$.

Write $c = c_j$. For a regular bidder, $b \le c_{j+1}$ holds, since either $c_{j+1}$ is the level of a grid bidder and thus at least $b$, or $j + 1 = k$ and $b \le c_k$. Her expected utility is therefore at most $e^{-s}(1 + \delta)$. For a grid bidder, $c_{j+1}$ may be her own level, which is missing, so we only get $\ln\frac{1 - c}{1 - b} \le 2\delta$ and an expected utility of at most $e^{-s}(1 + 2\delta)$. The special bidder also gets at most one click, and her value exceeds $1$ by $V - 1$ per click, so her expected utility is at most $V - 1 + e^{-s}(1 + \delta)$. Each of these bounds is at most the utility of the respective bidder in $\vA$.

\medskip
\noindent \textbf{Case 2:} $b > 1 - e^{-s}$. A distinguished regular bidder then gets the slot with rate $r$ at the price $1 - e^{-s}$ per click, as in $\vA$. Any other regular bidder or grid bidder gets the slot with rate $1$ right below the distinguished bidder at the price $1 - e^{-s}$ per click. If $b \ge V$, she may instead get the slot with rate $r$, but then she pays $V > 1$ per click. Such bids thus give a regular bidder at most $pre^{-s} + (1 - p)e^{-s}$ and a grid bidder at most $e^{-s}$, less than their utilities in $\vA$. For the special bidder, a bid in $(1 - e^{-s}, V)$ changes nothing, while a bid of at least $V$ gives her the slot with rate $r$ at the price $V$ per click, her value.

\medskip
\noindent \textbf{Case 3:} $b = 0$. If $t > 0$, all other bidders bid more than the deviator, so she gets no slot. Suppose that $t = 0$. If the deviator is the special bidder or a grid bidder, the other bidders with positive bids take $k$ slots. The remaining $r^2 + 1$ slots are taken by the $2r^2$ dummy and regular bidders who bid zero, as they precede the deviator at zero. A regular bidder who is not distinguished gets no slot, as the $r^2$ dummy bidders precede her and take the $r^2$ slots left by the $k + 1$ bidders with positive bids. A distinguished regular bidder gets at most one click at the price zero, so the bid zero gives a regular bidder at most $pe^{-s}$ in expectation.

\medskip
In all three cases, no deviation is profitable. Since the dummy bidders have value zero, no deviation is profitable for them either, and $\vA$ is a CCE.

It remains to compare the welfare of $\vA$ with the optimum. We have $\LW(I, \vA) = r + r^2\, \mathbb P(t > 0) + k - 1 + V$, where the four terms come from the distinguished bidder, the other regular bidders, the grid bidders and the special bidder. Let $\vx$ be the allocation that gives the slot with rate $r$ to the special bidder and the other slots to the regular and grid bidders, so that $\OPT(I) \ge \LW(I, \vx) = rV + r^2 + k$. Since $V = 1 + re^{-s}(1 + \delta)$ and $\mathbb P(t > 0) = 1 - e^{-s} + qe^{-s}$, both welfare values are quadratic in $r$, and their leading coefficients give
\[
\lim_{r \to \infty} \frac{\LW(I, \vx)}{\LW(I, \vA)} = \frac{1 + e^{-s}(1 + \delta)}{1 - e^{-s} + qe^{-s}} .
\]
Moreover, $\lim_{k \to \infty} \delta = \lim_{k \to \infty} c_1 = 0$ and $\lim_{k \to \infty} q = 1 - s$, so
\[
\lim_{k \to \infty} \frac{1 + e^{-s}(1 + \delta)}{1 - e^{-s} + qe^{-s}} = \frac{1 + e^{-s}}{1 - se^{-s}} .
\]
For $s = \mathcal{W}_0(1)$, i.e., $s = e^{-s}$, this limit equals $\frac{1 + s}{1 - s^2} = \frac{1}{1 - s}$, and the PoA is therefore at least $\frac{1}{1 - \mathcal{W}_0(1)}$, as $\OPT(I) \ge \LW(I, \vx)$ holds for all $k$ and $r$.
\end{proof}

\begin{theorem}\label{thm:lb-gsp-budgets}
Under NOB, the PoA of CCE of GSP auctions with budgets is at least $1 + \frac{e}{e-1} \approx 2.582$.
\end{theorem}
\begin{proof}
Fix integers $k \ge 4$ and $r \ge 2$, and modify the construction in the proof of Theorem~\ref{thm:lb-gsp-free} as follows. Let $s = 1$, so that $\delta = \frac{1}{k}$, and do not move mass. Lemma~\ref{lem:lb-gsp-price}\ref{item:lb-gsp-mean} then shows that $\E[\max\set{c_j, t}] = 1 - \frac{2 - j\delta}{e}$ holds for every $j \in \set{0, \dots, k}$, and in particular $\E[t] = 1 - \frac{2}{e}$.

Consider a GSP auction with one slot of click-through rate $r$, $r^2 + k$ slots of click-through rate $1$ and $r^2$ slots of click-through rate $\delta$. The regular, grid and special bidders and their bids are as before, with $V := 1 + \frac{r(1 + \delta)}{e}$, and grid bidder $j$ now has budget $1 - \frac{2 - j\delta}{e}$. Moreover, there are $r^2$ high bidders $i$ with value $\bar v_i = k(1 - \frac{2}{e})$ per click and budget $1 - \frac{2}{e}$, and one dummy bidder with value zero. The high bidders and the dummy bidder bid $t$, and $\bar v_i > 1$ holds, as $k \ge 4$. All other budgets are infinite. Ties at positive bids favor the special bidder, then the grid bidders in decreasing order of index, then the high bidders, then the regular bidders and then the dummy bidder. Ties at zero favor the regular bidders, then the high bidders, then the dummy bidder, then the grid bidders and then the special bidder.

In the resulting distribution $\vA$, the distinguished, special and grid bidders get the same slots at the same prices as before. If $t > 0$, the high bidders get the remaining slots with rate $1$, the other regular bidders get the slots with rate $\delta$, and all of them pay $t$ per click. If $t = 0$, the high bidders and the other regular bidders exchange their slots and pay nothing. The dummy bidder gets no slot.

We now verify budget feasibility and NOB and compute the expected utilities. A regular bidder who is not distinguished gets the slot with rate $1$ if $t = 0$, which has probability $\frac{1}{e}$, and a slot with rate $\delta$ otherwise. Her expected value is therefore $\frac{1}{e} + \delta(1 - \frac{1}{e})$, while her expected payment and her expected willingness to pay are both $\delta\, \E[t] = \delta(1 - \frac{2}{e})$. Her value thus exceeds her willingness to pay by $\frac{1 + \delta}{e}$, which is $e^{-s}(1 + \delta)$ for $s = 1$. The argument of the proof of Theorem~\ref{thm:lb-gsp-free} therefore gives~\eqref{eq:NOB} for the regular bidders, and their expected utility is $\frac{pr}{e} + (1 - p)\frac{1 + \delta}{e} \ge \frac{1 + \delta}{e}$, where $p = \frac{1}{r^2 + 1}$.

A high bidder gets the slot with rate $1$ if $t > 0$ and a slot with rate $\delta$ if $t = 0$. Her expected value is therefore $\bar v_i(1 - \frac{1}{e} + \frac{\delta}{e}) \ge \bar v_i\delta = 1 - \frac{2}{e}$, and her expected payment and her expected willingness to pay are $\E[t] = 1 - \frac{2}{e}$, her budget. She is thus budget-feasible and satisfies~\eqref{eq:NOB}. Grid bidder $j$ has willingness to pay $\E[\max\set{c_j, t}] = 1 - \frac{2 - j\delta}{e}$, her budget, which is less than her value $1$, and her expected payment $\E[\max\set{c_{j-1}, t}]$ is smaller. Her expected utility is $\frac{2 - (j - 1)\delta}{e} \ge \frac{1 + 2\delta}{e}$, as $(j + 1)\delta \le 1$. The special bidder pays $\E[\max\set{c_{k-1}, t}] = 1 - \frac{1 + \delta}{e}$ per click in expectation and has willingness to pay $1 - \frac{1}{e} < V$, so her expected utility is $V - 1 + \frac{1 + \delta}{e}$. Since the dummy bidder gets no slot, $\vA$ is budget-feasible and NOB-feasible.

We now turn to deviations. There are $2r^2 + k + 1$ slots, and every bidder has $2r^2 + k + 1$ other bidders, all of whom bid at least $t$. Hence the case analysis in the proof of Theorem~\ref{thm:lb-gsp-free} applies to the regular, grid and special bidders, where a slot with rate $\delta$ gives less than one click. It shows that no deviation of these bidders is profitable, even without their budgets. The only changes concern Case~3. There, the regular, high and dummy bidders who bid zero precede the grid bidders and the special bidder, and they take all slots left by the bidders with positive bids. A regular bidder, however, now gets a slot whenever $t = 0$, at the price zero, which yields at most $\frac{1}{e}$.

It remains to consider the high bidders, for whom we need their budget. For a pure bid $b$ of a high bidder $i$, let $Q(b)$ and $P(b)$ be her expected number of clicks and her expected payment. The difference $Q(b) - P(b)$ is the expected utility of $b$ for a bidder with value $1$ per click in her place, and we claim that
\begin{equation}\label{eq:lb-gsp-high}
Q(b) - P(b) \le \frac{1 + \delta}{e} \qquad \text{for every } b \ge 0 .
\end{equation}
For $b \in (0, 1 - \frac{1}{e}]$, this follows by the argument of Case~1 for the regular bidders. For $b = 0$, she gets at most a slot with rate $\delta$, only if $t = 0$ and at the price zero, so $Q(b) \le \frac{\delta}{e}$ and $P(b) = 0$. For $b \in (1 - \frac{1}{e}, V)$, she gets the slot with rate $1$ right below the distinguished bidder at the price $1 - \frac{1}{e}$, so $Q(b) - P(b) = \frac{1}{e}$. For $b \ge V$, she gets the slot with rate $r$ at the price $V > 1$ per click, so $Q(b) - P(b) < 0$.

By~\eqref{eq:lb-gsp-high}, every pure bid $b$ gives her utility
\[
\bar v_iQ(b) - P(b) = \bar v_i\bigl(Q(b) - P(b)\bigr) + (\bar v_i - 1)P(b) \le \bar v_i\, \frac{1 + \delta}{e} + (\bar v_i - 1)P(b) .
\]
By averaging, a randomized deviation that meets her budget $1 - \frac{2}{e}$ gives her at most $\bar v_i\, \frac{1 + \delta}{e} + (\bar v_i - 1)(1 - \frac{2}{e})$, as $\bar v_i > 1$. This is her expected utility in $\vA$, since
\[
\bar v_i \Bigl( 1 - \frac{1}{e} + \frac{\delta}{e} \Bigr) - \Bigl( 1 - \frac{2}{e} \Bigr) = \bar v_i\, \frac{1 + \delta}{e} + (\bar v_i - 1)\Bigl( 1 - \frac{2}{e} \Bigr) ,
\]
and hence $\vA$ is a CCE.

The liquid welfare of $\vA$ is
\[
\LW(I, \vA) = r^2 \Bigl( 1 - \frac{2}{e} \Bigr) + r + r^2 \Bigl( \frac{1}{e} + \delta \Bigl( 1 - \frac{1}{e} \Bigr) \Bigr) + \sum_{j=1}^{k-1} \Bigl( 1 - \frac{2 - j\delta}{e} \Bigr) + V .
\]
Here, the first term is the liquid value of the high bidders, and the next two terms are that of the regular bidders. The last two terms are those of the grid bidders and of the special bidder. Consider now the allocation $\vx$ that gives the slot with rate $r$ to the special bidder, the slots with rate $1$ to the regular and grid bidders, and the slots with rate $\delta$ to the high bidders. A high bidder then has liquid value $\bar v_i\delta = 1 - \frac{2}{e}$, so
\[
\OPT(I) \ge \LW(I, \vx) = rV + r^2 + 1 + \sum_{j=1}^{k-1} \Bigl( 1 - \frac{2 - j\delta}{e} \Bigr) + r^2 \Bigl( 1 - \frac{2}{e} \Bigr) .
\]
Since $V = 1 + \frac{r(1 + \delta)}{e}$, both liquid welfare values are quadratic in $r$, with leading coefficients $(1 - \frac{1}{e})(1 + \delta)$ and $2 - \frac{1}{e} + \frac{\delta}{e}$, respectively. This gives
\[
\lim_{r \to \infty} \frac{\LW(I, \vx)}{\LW(I, \vA)} = \frac{2 - \frac{1}{e} + \frac{\delta}{e}}{(1 - \frac{1}{e})(1 + \delta)} ,
\]
and since $\delta = \frac{1}{k}$,
\[
\lim_{k \to \infty} \frac{2 - \frac{1}{e} + \frac{\delta}{e}}{(1 - \frac{1}{e})(1 + \delta)} = \frac{2 - \frac{1}{e}}{1 - \frac{1}{e}} = 1 + \frac{e}{e-1} .
\]
As $\OPT(I) \ge \LW(I, \vx)$ holds for all $k$ and $r$, the PoA is at least $1 + \frac{e}{e-1}$.
\end{proof}

\subsection{Uniform Price Auction}\label{sec:upa-lb}

The bound of Theorem~\ref{thm:upa-individual-certificate} is tight, already for pure Nash equilibria of budget-free instances. The construction below adapts the lower bound of~\cite{birmpas2019tight} for the standard bidding interface to uniform bidding.
\begin{theorem}\label{thm:UPA-LB}
    Under NOB, the PoA of pure Nash equilibria of the budget-free UPA under uniform bidding is at least $\frac{2 + \mathcal{W}_0(-e^{-2})}{1 + \mathcal{W}_0(-e^{-2})} \approx 2.188$.
\end{theorem}
\begin{proof}
Fix an integer $n \ge 2$, and let $\ell := \bigl\lceil \frac{\alpha_0}{1 - \alpha_0} (n-1)^2 \bigr\rceil$, where $\alpha_0 = -\mathcal{W}_0(-e^{-2}) \approx 0.159$ is the constant of Appendix~\ref{sec:first-price-lb}, which satisfies~\eqref{eq:alpha-zero-defining}. We construct an instance of the uniform price auction with $n$ bidders and $k := \ell + (n-1)^2$ units. Bidder~$n$ has the valuation $v_n(q) = \min\set{q, k - n + 1}$. Every bidder $i < n$ is unit-demand with value $(n-1)\, b_i$, where
\[
b_i := \frac{i\, (n-1)}{\ell + i\, (n-1)} ,
\]
i.e., $v_i(q) = (n-1)\, b_i \min\set{q, 1}$. All budgets are infinite.

We consider the bid profile $\va$ with $a_n = (0, \ell)$ and $a_i = (b_i, n - 1)$ for every $i < n$. The total demand is $\ell + (n-1)^2 = k$, so every marginal bid wins and $\beta_0(\va) = 0$. Bidder~$n$ wins $\ell$ units, and every bidder $i < n$ wins $n - 1$ units.

We first verify that $\va$ is NOB-feasible. In the UPA, we have $W_i(\va) = b_i\, X_i(\va)$, so the left-hand side of~\eqref{eq:NOB} is $0$ for bidder~$n$ and $(n-1)\, b_i = v_i(n - 1)$ for every bidder $i < n$. Hence~\eqref{eq:NOB} holds for every bidder, with equality for the bidders $i < n$.

We now show that $\va$ is a pure Nash equilibrium. Every bidder $i < n$ obtains her value $(n-1)\, b_i$ at the price zero, which is the largest utility she can attain. Bidder~$n$ obtains $\z_n(\va) = v_n(\ell) = \ell$, as $\ell \le k - n + 1$. Consider any deviation $a_n'$ of bidder~$n$, and let $x := X_n(a_n', \va_{-n})$. If $x \le \ell$, then $\z_n(a_n', \va_{-n}) \le v_n(x) \le \ell$. Otherwise, the other bidders demand $(n-1)^2 = k - \ell$ units, but only $k - x$ units remain for them, so their $x - \ell$ lowest marginal bids lose. Since $b_1 < \dots < b_{n-1}$ and every bidder $i < n$ submits $n - 1$ marginal bids, a marginal bid of bidder $j := \lceil \frac{x - \ell}{n - 1} \rceil \le n - 1$ loses, and thus $\beta_0(a_n', \va_{-n}) \ge b_j$. As $x \le \ell + j\, (n-1)$, we obtain
\[
\z_n(a_n', \va_{-n}) \le x\, (1 - b_j) = \frac{x\, \ell}{\ell + j\, (n-1)} \le \ell = \z_n(\va) .
\]
No deviation of bidder~$n$ is therefore profitable, and $\va$ is a pure Nash equilibrium.

It remains to compare the liquid welfare of $\va$ with the optimum. As all budgets are infinite, we have $\LW(I, \va) = \ell + \sum_{i < n} (n-1)\, b_i$. Giving $k - n + 1$ units to bidder~$n$ and one unit to every other bidder is feasible, so $\OPT(I) \ge k - n + 1 + \sum_{i < n} (n-1)\, b_i$. Since $(n-1)\, b_i = (n-1) - \frac{(n-1)\, \ell}{\ell + i\, (n-1)}$ and $k - n + 1 - \ell = (n-1)(n-2)$, this gives
\[
\frac{\OPT(I)}{\LW(I, \va)} \ge 1 + \frac{(n-1)(n-2)}{k - \ell \sum_{i=1}^{n-1} \frac{n-1}{\ell + i\, (n-1)}} \ge 1 + \frac{(n-1)(n-2)}{k - \ell \ln\frac{k + n - 1}{\ell + n - 1}} .
\]
Here, the second inequality follows from
\[
\sum_{i=1}^{n-1} \frac{n-1}{\ell + i\, (n-1)} \ge \int_1^n \frac{n-1}{\ell + z\, (n-1)}\, dz = \ln\frac{\ell + n\, (n-1)}{\ell + n - 1} = \ln\frac{k + n - 1}{\ell + n - 1} ,
\]
as the integrand is decreasing.

We conclude by evaluating this lower bound in the limit, where $\ell$ and $k$ depend on $n$. By the definition of $\ell$, we have $\lim_{n \to \infty} \frac{\ell}{(n-1)^2} = \frac{\alpha_0}{1 - \alpha_0}$, and since $k = \ell + (n-1)^2$, this gives
\[
\lim_{n \to \infty} \frac{\ell}{k} = \alpha_0 \qquad \text{and} \qquad \lim_{n \to \infty} \frac{(n-1)(n-2)}{k} = 1 - \alpha_0 .
\]
Moreover, $\lim_{n \to \infty} \frac{n}{k} = 0$, so the continuity of the logarithm gives $\lim_{n \to \infty} \ln\frac{k + n - 1}{\ell + n - 1} = \ln\frac{1}{\alpha_0}$. Dividing the numerator and the denominator by $k$, we obtain
\[
\lim_{n \to \infty} \frac{(n-1)(n-2)}{k - \ell \ln\frac{k + n - 1}{\ell + n - 1}} = \frac{1 - \alpha_0}{1 - \alpha_0 \ln\frac{1}{\alpha_0}} = \frac{1 - \alpha_0}{(1 - \alpha_0)^2} = \frac{1}{1 - \alpha_0} ,
\]
where the second equality uses $1 - \alpha_0 \ln\frac{1}{\alpha_0} = 1 + \alpha_0 (\alpha_0 - 2) = (1 - \alpha_0)^2$, by~\eqref{eq:alpha-zero-defining}. Since $\va$ is a NOB-feasible pure Nash equilibrium for every $n \ge 2$, the PoA is at least
\[
1 + \frac{1}{1 - \alpha_0} = \frac{2 - \alpha_0}{1 - \alpha_0} = \frac{2 + \mathcal{W}_0(-e^{-2})}{1 + \mathcal{W}_0(-e^{-2})} \approx 2.188 .
\]
\end{proof}

\subsection{All-Pay Auctions}\label{sec:lb-apa}

All-pay auctions are PYB, so their equilibria are NOB-feasible by Lemma~\ref{lem:IR}. Since every valuation for a single item is additive, the following bounds also hold for simultaneous all-pay auctions with XOS valuations.

\begin{theorem}\label{thm:lb-apa}
The PoA of CCE of budget-free single-item all-pay auctions is at least $2$, and the PoA of mixed Nash equilibria of single-item all-pay auctions with budgets is at least $3$.
\end{theorem}
\begin{proof}
We first prove the budget-free bound. Fix an integer $n \ge 4$, and consider an all-pay auction with a single item and $n$ bidders. Bidder~$n$ has value $1$, and every bidder $i < n$ has value $\frac{n-1}{2(n-2)} < 1$. All budgets are infinite. Let $z$ be uniformly distributed on $[0,1]$, and draw, independently of $z$, an ordered pair of distinct bidders from $N \setminus \set{n}$ uniformly at random. In the distribution $\vA$, the first bidder of the pair bids $z$, the second bids $\frac{z}{n-2}$, and all other bidders bid zero. Since $\frac{z}{n-2} < z$ holds whenever $z > 0$, the first bidder of the pair wins the item almost surely.

Every bidder $i < n$ is drawn first with probability $\frac{1}{n-1}$ and second with probability $\frac{1}{n-1}$. Her expected value is therefore $\frac{1}{n-1} \cdot \frac{n-1}{2(n-2)} = \frac{1}{2(n-2)}$, and her expected payment is
\[
\frac{1}{n-1} \cdot \frac{1}{2} + \frac{1}{n-1} \cdot \frac{1}{2(n-2)} = \frac{1}{2(n-2)} ,
\]
so her expected utility is zero. Bidder~$n$ never wins and pays nothing, so her utility is zero as well.

We now show that $\vA$ is a CCE, i.e., that no bid $b \ge 0$ gives a bidder a positive expected utility. For bidder~$n$, the highest opposing bid is $z$, so the expected utility of $b$ is $\min\set{b, 1} - b \le 0$. A bidder $i < n$ faces the highest opposing bid $\frac{z}{n-2}$ if she is drawn first, which happens with probability $\frac{1}{n-1}$, and the highest opposing bid $z$ otherwise. The expected utility of $b$ is therefore
\[
\frac{n-1}{2(n-2)} \Bigl( \frac{\min\set{(n-2)b, 1}}{n-1} + \frac{(n-2)\min\set{b, 1}}{n-1} \Bigr) - b =
\begin{cases}
0, & b \le \frac{1}{n-2}, \\
\frac{1}{2} \bigl( \frac{1}{n-2} - b \bigr), & \frac{1}{n-2} \le b \le 1, \\
\frac{n-1}{2(n-2)} - b, & b \ge 1 .
\end{cases}
\]
In each case, this is at most zero, where the last case uses $\frac{n-1}{2(n-2)} < 1$. Finally, $\LW(I, \vA) = \frac{n-1}{2(n-2)}$ and $\OPT(I) = 1$, so the ratio is $\frac{2(n-2)}{n-1}$, and $\lim_{n \to \infty} \frac{2(n-2)}{n-1} = 2$.

We now prove the bound with budgets. Let $H > 1$, and consider an all-pay auction with a single item and two bidders with values $H$ and $1$ and budgets $\frac{1}{2}$ and $\infty$, respectively, where ties at zero favor bidder~$1$. In the product distribution $\vA$, bidder~$1$ bids uniformly from $[0, 1]$, and bidder~$2$ bids zero with probability $1 - \frac{1}{H}$ and uniformly from $[0,1]$ otherwise. Bidder~$1$ pays her budget $\frac{1}{2}$ in expectation, so $\vA$ is budget-feasible.

We now show that $\vA$ is a mixed Nash equilibrium. A bid $b \in [0,1]$ of bidder~$1$ wins with probability $1 - \frac{1}{H} + \frac{b}{H}$, so its expected utility is $H - 1$. Higher bids win surely and give $H - b < H - 1$. For bidder~$2$, a bid $b \in [0, 1]$ wins with probability $b$, so its expected utility is zero, while a bid $b > 1$ gives $1 - b < 0$. Hence no bid gives a bidder more than her expected utility in $\vA$, whether or not it meets her budget, and by averaging, no randomized deviation does either.

Finally, bidder~$2$ wins with probability $\frac{1}{2H}$, namely when she bids from $[0,1]$ and above bidder~$1$. Bidder~$1$ thus has expected value $H - \frac{1}{2} > \frac{1}{2}$, and $\LW(I, \vA) = \frac{1}{2} + \frac{1}{2H}$. Giving the item to bidder~$1$ with probability $\frac{1}{2H}$ and to bidder~$2$ otherwise shows that $\OPT(I) \ge \frac{3}{2} - \frac{1}{2H}$. The ratio is thus at least $\frac{3H - 1}{H + 1}$, and $\lim_{H \to \infty} \frac{3H - 1}{H + 1} = 3$.
\end{proof}

\end{document}